\documentclass[10pt,a4paper]{article}

\usepackage[utf8]{inputenc}
\usepackage{amsmath, amsthm, amssymb}
\usepackage{graphicx}
\usepackage{enumerate}
\usepackage{enumitem}
\usepackage{listings}
\usepackage{graphicx}
\usepackage{grffile}
\usepackage[makeroom]{cancel}
\usepackage{float}
\usepackage{setspace}
\usepackage{afterpage}
\usepackage{multirow}
\usepackage{xcolor}
\usepackage{bm}         %For bolding with greek letters %\bm{}
\usepackage{bbm} %\mathbbm{1} Indicator Function
\usepackage[citecolor=blue,colorlinks=true,linkcolor=blue]{hyperref} %for links to theorem environments
\usepackage{placeins} %to use \FloatBarrier

\usepackage[style=apa]{biblatex}

\usepackage[parfill]{parskip}

\usepackage{geometry,pdfpages}
\theoremstyle{plain}
\newtheorem{theorem}{Theorem}
\newtheorem{lemma}{Lemma}

\newtheorem{remark}{Remark} %\newtheorem*{remark}{Remark} to remove numbering
\newtheorem{proposition}{Proposition}
\newtheorem{assumption}{Assumption}
\newtheorem{definition}{Definition}
\newtheorem{corollary}{Corollary}

\usepackage{mathrsfs} % for likelihood L
\DeclareMathOperator*{\argmin}{arg\,min} %\argmin command

\newcommand{\dint}{\mathrm{d}} %d in integral

\usepackage{caption}
\usepackage{subcaption} %for 2 subfigures with graphicx

\usepackage{comment}
\specialcomment{itcomment}{\begingroup\itshape}{\endgroup}

\usepackage[ruled,vlined,linesnumbered]{algorithm2e} %for algorithms!!

\usepackage{url}

\newcommand{\ppx}[2]{\frac{\partial #1}{\partial #2}}
\newcommand{\ddt}[1]{\frac{\mathrm{d}#1}{\mathrm{d}t}}
\newcommand{\ddx}[1]{\frac{\mathrm{d}#1}{\mathrm{d}x}}
\newcommand{\ddxx}[2]{\frac{\mathrm{d}#1}{\mathrm{d}#2}}

\usepackage[noblocks]{authblk}

\title{\sc NuZZ: numerical Zig-Zag for general models}

\author[1,*]{Filippo Pagani}
\author[2]{Augustin Chevallier}
\author[3]{Sam Power}
\author[4]{Thomas House}
\author[4]{Simon Cotter}

\affil[1]{MRC Biostatistics Unit, University of Cambridge}
\affil[3]{School of Mathematics, University of Bristol} 
\affil[2]{Department of Mathematics and Statistics, Lancaster University}
\affil[4]{Department of Mathematics, University of Manchester}
\affil[*]{Corresponding author(s): Filippo Pagani. E-mail(s): filippo.pagani@mrc-bsu.cam.ac.uk}

\date{\today}

\begin{document}

% INSTRUCTIONS: you need to compile from the file main-<name>.tex
% if you want to change template, you need to clear the cache
% the button for that is at the bottom of the error log.

\maketitle

\begin{abstract}
\noindent{Markov chain Monte Carlo (MCMC) is a key algorithm in computational statistics, and as datasets grow larger and models grow more complex, many popular MCMC algorithms become too computationally expensive to be practical. Recent progress has been made on this problem through development of MCMC algorithms based on Piecewise Deterministic Markov Processes (PDMPs), irreversible processes that can be engineered to converge at a rate which is independent of the size of data. While there has understandably been a surge of theoretical studies following these results, PDMPs have so far only been implemented for models where certain gradients can be bounded, which is not possible in many statistical contexts. Focusing on the Zig-Zag process, we present the Numerical Zig-Zag (NuZZ) algorithm, which is applicable to general statistical models without the need for bounds on the gradient of the log posterior. This allows us to perform numerical experiments on: (i) how the Zig-Zag dynamics behaves on some test problems with common challenging features; and (ii) how the error between the target and sampled distributions evolves as a function of computational effort for different MCMC algorithms including NuZZ. Moreover, due to the specifics of the NuZZ algorithms, we are able to give an explicit bound on the Wasserstein distance between the exact posterior and its numerically perturbed counterpart in terms of the user-specified numerical tolerances of NuZZ.}
\end{abstract}

%used \input as \include induces a page-break
\section{Introduction}
\label{intro}

Markov Chain Monte Carlo (MCMC) methods are a mainstay of modern statistics~\parencite{Brooks:2011}, and are employed in a wide variety of applications, such as astrophysics (e.g.\ \textcite{des2017}), epidemiology (e.g.\ \textcite{o1999bayesian}, \textcite{house2016}), chemistry (e.g.\ \textcite{cotter2019,cotter2020transport}) and finance (e.g.\ \textcite{kim1998}). There are a diverse range of challenges that MCMC algorithms can be hindered by, including multimodality of the posterior, complex non-linear dependencies between parameters, large number of parameters (e.g.\ \textcite{cotter2013mcmc}), or a large number of observations (e.g.\ \textcite{bardenet2017}).  Common MCMC methods such as the Random Walk Metropolis (RWM) algorithm~\textcite{metropolis1953} offer simplicity and flexibility, but also have several limitations. 

Throughout the history of MCMC there have been many attempts to improve algorithmic efficiency~\parencite{robert-casella2011}, including the work of~\textcite{hastings1970} in generalising the RWM approach to non-symmetric proposals, and the use of gradient information to explore the state space according to a discretised Langevin diffusion in the Metropolis-Adjusted Langevin Algorithm (MALA)~\parencite{roberts1996}. Both of these methods rely on reversible dynamics that satisfies a detailed balance condition, used to ensure convergence to the appropriate limiting distribution.

Lifting the process simulated to a larger state space can aid exploration of the original state space, with one of the most successful such approaches being the Hamiltonian Monte Carlo (HMC) algorithm in Euclidean space \parencite{duane1987}, and on Riemannian manifolds (RMHMC)~\parencite{girolami2011}. HMC relies on adding momentum variables to the state space, while maintaining detailed balance and reversibility thanks to a Metropolis acceptance step.

Recently, there has been growing interest in irreversible algorithms, which have been shown to converge faster and have lower asymptotic variance than their reversible counterparts~\parencite{chen2013, ottobre2016}.  This was initially investigated by \textcite{diaconis2000} and \textcite{turitsyn2011}, where the authors achieve irreversibility via ``lifting" and injecting vortices in the state space. In \textcite{ma-fox2016} the authors provide a framework for developing irreversible algorithms.

Not all irreversible algorithms are built from reversible ones. In the physics literature, \textcite{rapaport2009} and \textcite{peters-dewith2012} laid the foundations for the Bouncy Particle Sampler (BPS), which was then thoroughly studied by \textcite{bouchard2017}. The BPS abandons diffusions in favour of completely irreversible, piecewise deterministic dynamics based on the theory developed in \textcite{davis1984, davis1993}. The close relationship between the BPS and HMC is explored in \textcite{deligiannidis2018}.

The Zig-Zag (ZZ) process~\parencite{bierkens2016, bierkens-roberts2016, bierkens-duncan2016}, a variant of the Telegraph process~\parencite{kolesnik2013}, is another Piecewise Deterministic Markov Process (PDMP) closely related to the BPS, even coinciding in one dimension. Its properties were first explored in work by \textcite{bierkens-duncan2016} and \textcite{bierkens2019bigdata}.  The ZZ sampler is particularly appealing as when properly modified with upper bounds for the switching rate, it has a computational cost per iteration that need not grow with the size of the data. This feature makes it particularly well suited for `tall, thin' data, e.g.\ multivariate datasets with a large number of observations, $n$, and a relatively small number of parameters, $d$.

The BPS sampler and the ZZ sampler are both PDMPs that change direction at random times according to the interarrival times of a Poisson process whose rate depends on the gradient of the target density. The main difference between them is the way in which the new direction is chosen. Conditions for ergodicity have been established for both algorithms~\parencite{bierkens2019ergodicity, durmus2019}, and the Generalised Bouncy Particle Sampler~\parencite{wu2017} partially bridges the gap between the two.

Given their desirable properties, there has been a wealth of theoretical analysis of PDMPs in MCMC developed in recent years~\parencite{bierkens2019highdimensional, zhao2019, andrieu2021hypocoercivity, bierkens2019, chevallier2020, andrieu2021, chevallier2021}. In \textcite{vanetti2018}, the authors provide insights into a general framework for construction of continuous and discrete time algorithms based on PDMPs, and discuss some improvements that ameliorate some of the current challenges that PDMP samplers face.  There remains, however, a lack of guidance about which classes of statistical problems are well-suited to the use of PDMP algorithms, and how those algorithms are best to be implemented. 

In this work we aim to study the Zig-Zag dynamics and how well they explore distributions with specific features, often found in practical Bayesian statistics. In order to do this, we develop a numerical approximation to calculate the time to the next switch for general distributions, without requiring analytical solutions or bounds for the switching rate.

Moreover, we obtain results in controlling the numerical error induced by our algorithm. A first result in this sense was achieved by \textcite{riedler2013}, where the author shows that under general assumptions, sample paths resulting from numerical approximations of PDMPs converge almost surely to the paths followed by the exact process as the numerical tolerances approach zero. A more relevant result can be found in \textcite{bertazzi2021}, where the authors study a general class of Euler-type integration routines, which yield a Markov process at each point of the mesh defined by the integrator.  Their framework is sufficiently general to accomodate PDMPs with trajectories that are not linear and need to be approximated numerically.  On the contrary, the process derived using the adaptive integration methods examined in this work is Markov only at points on the skeleton chain, but not at points on the deterministic portion of the trajectory.

The paper is organised as follows.  In Section \ref{sec:reviewZZ} we review the technical details of the Zig-Zag process and fix the notation for the rest of this work. In Section \ref{sec:NuZZ} we present an alternative algorithm, the Numerical Zig-Zag (NuZZ), which uses numerical methodology to compute switching times of the Zig-Zag process. In Section \ref{sec:err} we study how the numerical tolerances affect the quality of the NuZZ MCMC sample.  In Section \ref{sec:exp} we present some numerical experiments which test the performance of a range of MCMC algorithms, including NuZZ, for a range of test problems with features often found in real-world models. Finally in Section \ref{sec:con} we conclude with some discussions of the results.

\section{The Zig-Zag sampler}
\label{sec:reviewZZ} 

The Zig-Zag sampler is an MCMC algorithm designed for drawing samples from a probability measure $\pi(\mathbf{x})$ with a smooth density on $\mathcal{X} = \mathbb{R}^d$. In the same vein as methods like HMC, the Zig-Zag sampler operates on an extended state space where the position variable $\mathbf{x}$ is augmented by a velocity variable $\mathbf{v}$.

The Zig-Zag sampler is driven by the Zig-Zag process, a continuous time PDMP which describes the motion of a particle $\mathbf{x}$, with a piecewise constant velocity which is subject to intermittent jumps. The dynamics and jumps are constructed in such a way that the process converges in law to a unique invariant measure which admits $\pi$ as its marginal.

Let the Zig-Zag process be defined as $\mathbf{Z}_t = (\mathbf{X}_t, \mathbf{V}_t)$ on the extended space $E = \mathcal{X} \times \mathcal{V}$, with $\mathcal{X}=\mathbb{R}^d$ and $\mathcal{V} = \{-1,1\}^d$. From an initial state $(\mathbf{x}, \mathbf{v})$, the process evolves in time following the deterministic linear dynamics
\begin{equation}
\label{linearDyn}
\dot{\mathbf{x}} = \mathbf{v}\, , \qquad \dot{\mathbf{v}} = \mathbf{0} \, ,
\end{equation}
where the dots denote total differentiation with respect to time. The probability of changing direction (i.e.\ the sign of one component of the vector $\mathbf{v}$) increases as the process heads into the tails of $\pi(\mathbf{x})$. The event of changing the sign of $v_i$ is modelled as an inhomogeneous Poisson process with component-wise rate
\begin{equation}
\label{lambda}
\lambda_i (\mathbf{x}, \mathbf{v}) = 0 \vee - v_i \partial_i \log \pi (\mathbf{x}) \, ,
\end{equation}
where $\partial_i$ is the partial derivative w.r.t.\ the $i$th component. Hence, the probability of having an event for coordinate $i$ in $[t,t+ \dint t]$ is $\lambda_i(\mathbf{z}_t) \dint t + o( \dint t)$. This induces the ZZ process to spend more time in regions of high density of $\pi(\mathbf{x})$.
When a velocity switch is triggered, a switching time $\tau$ is obtained, and the process moves forward from $\mathbf{x}$ to $\mathbf{x}' = \mathbf{x}+\tau \mathbf{v}$. At which point the transition kernel $Q(\mathbf{v}' | \mathbf{x}',\mathbf{v})$ selects the velocity component to switch and flips its sign. From $(\mathbf{x}',\mathbf{v}')$ the dynamics proceed following Equation \ref{linearDyn} until another velocity switch occurs. The trajectories of the process produce `zig-zag' patterns such as those seen in Figure \ref{ZZrosenTraj}, motivating its name.

\begin{figure}[!h]
	\centering
	\includegraphics[width=.7\linewidth]{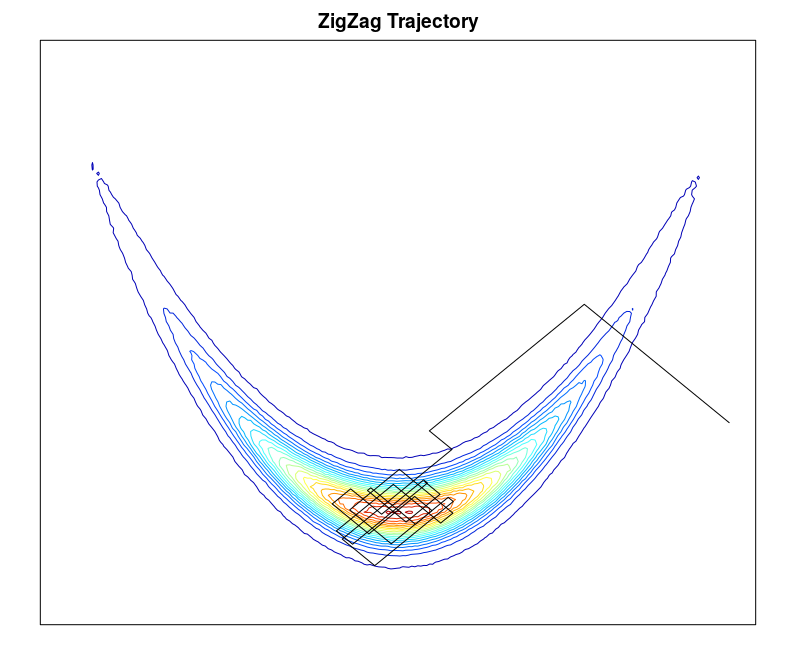}
	\caption{Example Zig-Zag trajectory. The target is a 2-dimensional Rosenbrock
		distribution.}
	\label{ZZrosenTraj}
\end{figure}

The linear dynamics in Equation \ref{linearDyn}, the jump rates $\lambda_i(\mathbf{x},\mathbf{v})$, and the transition kernel $Q(\mathbf{v}' | \mathbf{x}',\mathbf{v})$ uniquely define the Zig-Zag process. Its infinitesimal generator is known, and can be given in terms of the three above quantities as
\ifdefined\doublecolumn
  \begin{multline}
    \label{eq:generatorZZ}
    \mathcal{U}_{\mathrm{ZZ}}f(\mathbf{x},\mathbf{v}) = 
    \sum_{i=1}^d \biggl \{ v_i \frac{\partial f}{\partial x_i} +
    \lambda_i(\mathbf{x},\mathbf{v}) \big(f(\mathbf{x},\mathbf{F}_i(\mathbf{v})) - \\
    f(\mathbf{x},\mathbf{v})\big) \biggr \}
    \, ,
  \end{multline}
\else
  \begin{equation}
    \label{eq:generatorZZ}
    \mathcal{U}_{\mathrm{ZZ}}f(\mathbf{x},\mathbf{v}) = 
    \sum_{i=1}^d \left \{ v_i \frac{\partial f}{\partial x_i} +
    \lambda_i(\mathbf{x},\mathbf{v}) \left(f(\mathbf{x},\mathbf{F}_i(\mathbf{v})) -
    f(\mathbf{x},\mathbf{v})\right) \right \}
    \, ,
  \end{equation}
\fi
where $F_i(\mathbf{v})$ represents the operation of flipping the sign of the $i$th component.

What is described above is the Zig-Zag process with canonical rates, which has been proven to be ergodic for a large class of targets \parencite{bierkens2019ergodicity}. Taking the rate 
\begin{equation}
\label{lambda_gamma}
\lambda_i (\mathbf{x}, \mathbf{v}) = \left ( 0 \vee - v_i \partial_i \log \pi (\mathbf{x}) \right ) + \gamma_i (\mathbf{x}, \mathbf{v})
\end{equation}
(compare with Equation \eqref{lambda}), with $\gamma_i$ satisfying
\begin{equation}
    \lambda_i(\mathbf{x},\mathbf{v}) - \lambda_i(\mathbf{x},F_i(\mathbf{v}) ) = - v_i \partial_i \log \pi (\mathbf{x}) \, , \qquad \forall i=1,...,d \, ,
\end{equation}
also yields a process that targets the desired distribution, and the process is ergodic on a larger class of models. However, increasing $\gamma_i$ induces diffusive behaviour, in that the process tends to change direction more often than it otherwise would, and to cover less distance in the same amount of time. Hence, if the functions $\gamma_i$ are taken to be strictly positive, they are generally set to be as small as possible.

In this work we take $\gamma_i(\mathbf{x}, \mathbf{v}) = \gamma$, which gives an overall baseline switching rate of $\Gamma = \sum_i \gamma_i = d \gamma$. For $\mathbf{z} = (\mathbf{x}, \mathbf{v})$, we will often refer to the switching rate as a function of time as $\lambda_i^\mathbf{z} (t) := \lambda_i ( \mathbf{x} + t \cdot \mathbf{v}, \mathbf{v})$.

\subsection{Sampling the switching time}
\label{switchTimep}

In this section we outline the approach followed in \cite{bierkens2019bigdata}. Switches in $\mathbf{v}$ occur as the interarrival times of an inhomogeneous Poisson process with rate $\lambda_i(\mathbf{x},\mathbf{v})$. Starting the process at $t=0$, each $\tau_i$, the time to the next switch in the $i$th component conditional on no other switches occurring, has cumulative distribution function (cdf)
\begin{equation}
\label{eq:pdfTau}
G_{T_i}(\tau_i) = 1 - \exp
\left( -\int_0^{\tau_i} \lambda_i^\mathbf{z}(s) \, \dint{s} \right)
\, .
\end{equation}
%where $\mathbf{x}(s) = \mathbf{x} + s \mathbf{v}$. 
A switching time $\tau_i$ is sampled separately for each component with law given by \eqref{eq:pdfTau}, and the final switching time $\tau$ is selected as
\begin{equation*}
\tau = \min_{i=1,\dots,d} \tau_i \, , \qquad \mbox{with } \qquad
i^* := \argmin_{i=1,\dots,d} \tau_i \, .
\end{equation*}
Once $\tau$ has been calculated, the system follows the deterministic dynamics from $\mathbf{x}(s) = \mathbf{x}(0) + s \mathbf{v}$ to $\mathbf{x}(s) = \mathbf{x}(0) + \tau^- \mathbf{v}$, where a velocity flip occurs, and the process starts following the new dynamics $\mathbf{x}(s) = \mathbf{x}(\tau) +\mathbf{F}_{i^*}(\mathbf{v}) (s-\tau)$. All the larger $ \{ \tau_i \}_{i \neq i^*}$ are usually discarded, as they refer to a state of the system that can no longer be achieved. The same procedure is repeated at each iteration of the algorithm.

This method for sampling the switching times requires realisations of $d$ random variables with densities given by \eqref{eq:pdfTau}, where $d$ is the number of components of the Zig-Zag process, i.e.\ the number of parameters to be estimated via MCMC.

\subsection{Poisson sampling via thinning}

Sampling $\tau_i$ from Equation \eqref{eq:pdfTau} requires solving the integral in the exponential. This is often a difficult task, as the rates $\lambda_i$ depend on the posterior in a complex way as shown in Equation \eqref{lambda_gamma}. Analytic solutions are available only for simple targets and are problem-dependent.

In~\cite{bierkens2019bigdata} the authors propose the method of \emph{thinning} to sample the waiting times to the next switch, $\tau_i$. If the Jacobian or Hessian of the target $\pi(\mathbf{x})$ is bounded, then there exists a bound $\bar{\lambda}_i(\mathbf{x},\mathbf{v}) \ge \lambda_i(\mathbf{x},\mathbf{v})$, $\forall \mathbf{x},\mathbf{v}, i$. The bound $\bar{\lambda}_i$ can be used instead of $\lambda_i$ in Equation \eqref{eq:pdfTau}, so that $G_{T_i}(\tau_i)$ can be replaced with a simpler function. Once the $\tau_i$ are sampled and the time $\tau$ to the next switch is found, the proposed velocity switch occurs with probability $\lambda_{i*} / \bar{\lambda}_{i*}$, to correct for the fact that $\tau$ was not sampled from the correct distribution. If the switch is rejected, there is no change in $\mathbf{v}$, the process moves linearly to $\mathbf{x}+\tau \mathbf{v}$, and new switching times are sampled.

Sampling the switching times via thinning extends the applicability of the Zig-Zag sampler to a wider class of models than those where $\tau$ can be sampled by inverting the cdf, and forms the basis for the ZZCV algorithm \cite{bierkens2019bigdata}. Additionally, in the recent work of \cite{sutton2021}, the authors propose a method for systematically identifying numerical bounds which are valid on a given time horizon. Their method significantly extends the applicability of PDMPs, but again retain the cost of needing to construct a numerical bound at each iteration, and requires a priori knowledge of the target measure.

As such, despite these developments, the Zig-Zag sampler is still limited in its applicability to a narrow class of models, which has so far hampered the study of PDMP-based MCMC from a practical point of view.

\section{Numerical Zig-Zag}
\label{sec:NuZZ}

The Zig-Zag sampler described in Section \ref{sec:reviewZZ} relies on the availability of an easy way to sample from Equation \eqref{eq:pdfTau}, or on having knowledge of a bound for $\lambda_i$. In this section we extend the Zig-Zag algorithm to sample from arbitrary targets via numerical approximation of the switching times.

\subsection{The Sellke Zig-Zag process}\label{sec:SeZZ}

In order to ameliorate the issues described in Section \ref{intro} and \ref{sec:reviewZZ}, it is possible to reformulate the problem of finding the time to the next velocity switch in a more convenient way.  This approach is based on the Sellke construction~\parencite{sellke1983}, and we note that picking $\tau_i$ with law \eqref{eq:pdfTau} is equivalent to finding $\tau_i >0$ such that
\begin{equation}
\label{eq:rootfunct}
g(\tau_i) = \int_0^{\tau _i} \lambda_i(\mathbf{x}(s),\mathbf{v}) \, \dint{s} -R_i = 0 \, ,
\end{equation}
where $R_i \sim {\rm Exp}(1)$.  Ergo, finding $\tau_i$ is equivalent to finding the root of the function $g(\tau_i)$.

We demonstrate this equivalence in the appendix in Section \ref{app:SeZZ}, in which we describe SeZZ, the Zig-Zag process where the switching times are given by roots of \eqref{eq:rootfunct}, and where we show that the SeZZ generator is stationary with respect to the target measure. The Numerical Zig-Zag, or NuZZ, arises from numerical approximations of the roots of \eqref{eq:rootfunct}.

\subsection{Numerical integration}
\label{sec:integration}
 Unlike other Zig-Zag processes, which can rely on bounds which are either unknown or might not exist, numerical approximations of the roots of \eqref{eq:rootfunct} can be accomplished for a very wide class of models. There are, however, a number of problems to overcome. The function $g$ is in the set $C^1(\mathbb{R})\backslash C^2(\mathbb{R})$, and in practice it often increases very steeply when moving into the tails of the distribution.  The numerical method used to solve the integral in \eqref{eq:rootfunct} up to some $\widetilde{\tau}$ should be chosen on the basis of its efficient and accurate approximation of such a problem. The function $g$ is only once differentiable, and therefore numerical integrators like Runge-Kutta may perform unreliably. In particular, the first derivative of $\lambda_i$ is only piecewise continuous, which may also create problems for common root finding methods like Newton-Raphson.

The method we use to find the time to the next switch combines the QAGS (Quadrature Adaptive Gauss-Kronrod Singularities) integration routine from the GSL library \parencite{gsl}, with Brent's method \parencite{press2007} for finding the root.

\subsection{Brent's Algorithm}

Common root-finding algorithms such as Newton-Raphson and the secant method have good theoretical properties, such as superlinear convergence to the root. However, if the target function is not extremely well behaved, these algorithms may converge very slowly, or not converge at all. On the other hand, iteratively bisecting intervals is a more reliable method to find the root, as it converges linearly even when the function is not so well behaved. However, bisection has slower linear convergence.

Brent's method \parencite{press2007} was originally developed by Van Wijngaarden and Dekker in the 1960s, and it was only successively improved by Brent in 1973. Dekker's root-finding method combined the secant method with bisection, to converge super linearly when the function is well behaved, while maintaining linear convergence if the target is particularly difficult. 
Unfortunately, it is possible for Dekker's algorithm to choose to apply the secant method even when the change to the current guess of the root is arbitrarily small, leading to very slow convergence.

Brent's method introduces an additional check on the change that the secant method would apply to the current root guess. If the update is too small, the algorithm then performs a bisection step instead, which guarantees linear convergence in the worst case scenario.
%This ensures that in the worst case scenario, Brent's method only takes $N^2$ iterations to converge, where $N$ is the number of iterations that the bisection method would take to converge. 

Moreover, Brent's method is superior to previous methods in that when it determines the new root guess via the secant method, it uses inverse quadratic interpolation instead of the usual linear interpolation. 

%If the secant method determines the new guess for the root according to the formula
%\begin{equation*}
%\mathbf{x}_{n+1} = \mathbf{x}_{n} - f(\mathbf{x}_{n}) \frac{\mathbf{x}_{n} - \mathbf{x}_{n-1}}{f(\mathbf{x}_{n}) - f(\mathbf{x}_{n-1})} \, ,
%\end{equation*}
%Brent's method updates it as
%\begin{align*}
%	\mathbf{x}_{n+1} & = 
	%
%	\frac{ f(\mathbf{x}_{n-1}) f(\mathbf{x}_{n}) }{ (f(\mathbf{x}_{n-2}) - f(\mathbf{x}_{n-1}) ) (f(\mathbf{x}_{n-2}) - f(\mathbf{x}_{n})) } \mathbf{x}_{n-2} \\
	%
%	& + \frac{f(\mathbf{x}_{n-2}) f(\mathbf{x}_{n})}{ (f(\mathbf{x}_{n-1}) - f(\mathbf{x}_{n-2})) ( f(\mathbf{x}_{n-1}) - f(\mathbf{x}_{n}) ) } \mathbf{x}_{n-1} \\
	%
%	& + \frac{f(\mathbf{x}_{n-2}) f(\mathbf{x}_{n-1})}{ (f(\mathbf{x}_{n}) - f(\mathbf{x}_{n-2})) (f(\mathbf{x}_{n}) - f(\mathbf{x}_{n-1})) } \mathbf{x}_{n} \, .
%\end{align*}
Brent's method deals with numerical instabilities caused by the denominators of the secant update being too small by bracketing the root. That is to say, by storing and updating the range where the root is known to be by the intermediate value theorem, as the iterations are performed and the search for the root continues.

%As many root finding algorithms, Brent's method will stop when the error between two consecutive guesses is smaller than a user defined tolerance $\varepsilon_\mathrm{Bre}$.

In this work, for a Zig-Zag process as defined in Section \ref{switchTimep}, the root returned by Brent's method for each of the components, $\widetilde{\tau}_i$, satisfies the following inequality
\begin{equation}
    \left| \int_0^{\tilde{\tau}_i} \widetilde{\lambda}_i(\mathbf{x}(s), \mathbf{v} ) \, \dint s - R_i \right| \leqslant \varepsilon_{\mathrm{Bre}} \, , \qquad i=1,...,d \, ,
\end{equation}
where $\varepsilon_{\mathrm{Bre}}$ is the user-specified numerical tolerance, and $\widetilde{\lambda}_i$ is the piecewise linear approximation to to $\lambda$ induced by the numerical integration routine.

In Section \ref{sec:err} we discuss how we model the numerical tolerances for the NuZZ process, and what effect they have on the resulting MCMC sample. In the next sub-section, we discuss our actual implementation, where we reduce the number of integrals to be computed from $d$ to $1$.

\subsection{Efficient implementation}

The Zig-Zag Sampler discussed in Section \ref{switchTimep} relies on sampling a switching time $\tau_i$ for each component, then taking the minimum $\tau = \min_{i=1,\ldots,d} \tau_i$. This approach requires solving $d$ integrals for each iteration. However, it is possible to reduce the cost to only one integral per iteration, by using the following procedure. First, sample $\tau$ such that
\begin{equation}
\label{eq:sellketGeneral}
\int_0^\tau \Lambda(\mathbf{x}(s),\mathbf{v}) \, \dint{s} = R \sim \mathrm{Exp}(1) \, ,
\end{equation}
where
$\Lambda(\mathbf{x}(s),\mathbf{v}) := \sum_{i=1}^d \lambda_i(\mathbf{x}(s),\mathbf{v})$ is the total rate of switching in any component. Then sample the index of the first component to switch, $i^*$, as a multinomial distribution with $i$th cell probability $\lambda_i(\mathbf{x}(\tau),\mathbf{v}) / \Lambda(\mathbf{x}(\tau),\mathbf{v})$. The last step can be achieved efficiently by taking $i^*$ such that
\ifdefined\doublecolumn
\begin{multline}
i^* = \min \bigg\{ l \in \{1,\ldots,d \} \,\bigg|\, \\ \sum_{i=1}^l
\frac{\lambda_i(\mathbf{x}(\tau),\mathbf{v})}{\Lambda(\mathbf{x}(\tau),\mathbf{v})}
\ge \mathrm{Unif}(0,1) \bigg\} \, ,
\end{multline}
\else
\begin{equation}
i^* = \min \bigg\{ l \in \{1,\ldots,d \} \,\bigg|\, \sum_{i=1}^l
\frac{\lambda_i(\mathbf{x}(\tau),\mathbf{v})}{\Lambda(\mathbf{x}(\tau),\mathbf{v})}
\ge \mathrm{Unif}(0,1) \bigg\} \, ,
\end{equation}
\fi
and finding $i^*$ by bisection search. This sampling method is used in the Gillespie algorithm, for Monte Carlo simulation of chemical reaction networks~\parencite{gillespie1977}, and is a standard result for Poisson processes. 

\begin{proposition}
	\label{prop:gillespie}
	A Zig-Zag process where switches happen at rate
	$\Lambda(\mathbf{x}(s),\mathbf{v}) = \sum_{i=1}^d
	\lambda_i(\mathbf{x}(s),\mathbf{v})$, and the index of the velocity component
	to switch is distributed as 
	\begin{equation}
	\label{multinomialindex}
	i^* \sim \mathrm{Multi} \left ( \frac{\lambda_i(\mathbf{x}(\tau),\mathbf{v})}{ \Lambda(\mathbf{x}(\tau),\mathbf{v})} \right ) \, ,
	\end{equation}
	targets the same invariant distribution $\pi(\mathbf{x})$ as the Zig-Zag process
	with generator given in \eqref{eq:generatorZZ}.
\end{proposition}

\begin{proof}
	This follows from standard results on Poisson processes. For a proof tailored to this case see Appendix \ref{gillespieProof}
\end{proof}

Importantly, sampling $\tau$ in this way reduces the number of integrals to be computed numerically from $d$ to 1, at the price of a slight increase in the complexity of the integrand. Moreover, there may be more efficient implementations than bisection to find $i^*$, which may reduce the computational complexity further.

\medskip

\begin{algorithm}%[H] H option not allowed on 2 cols
	\caption{Numerical Zig-Zag (NuZZ)}
	\DontPrintSemicolon
	\BlankLine
	Initialise: $t=0$; $\mathbf{x}(0)=\mathbf{x}_{\mathrm{init}}$;
	$\mathbf{v}(0)=\mathbf{v}_{\mathrm{init}}$; $N=$ number of desired switching points.\;
	\For{$k = 1:N$}{
		Sample $R \sim \mbox{Exp}(1)$\;
		Calculate numerically $\tau$ such that $\int_0^\tau \Lambda(\mathbf{x}(s),\mathbf{v}) \,
		\dint{s} -R=0$ \;
		Set $\mathbf{x}(t+\tau) = \mathbf{x}(t) + \tau \mathbf{v}(t)$\;
		Sample $i^* \sim \mathrm{Multi} \left(
		\frac{\lambda_i(\mathbf{x}(t+\tau),\mathbf{v}(t))}{
			\Lambda(\mathbf{x}(t+\tau),\mathbf{v}(t))} \right)$\;
		Set $\mathbf{v}(t+\tau) = \mathbf{F}_{i^*}(\mathbf{v}(t))$\;
		Record switching point $\mathbf{X}_k = \mathbf{x}(t+\tau)$,
		$\mathbf{V}_k = \mathbf{v}(t+\tau)$, $T_k = t+\tau$.\;
		%\If{$\mathcal{U}(0,1) < \frac{\lambda_{i^*}(\bm{x}(\tau),\bm{v})}{\Lambda(\bm{x}(\tau),\bm{v})}$}{
		%Accept: record $\bm{x}_k$ as a switching point\;
		%$k=k+1$\;
		%} %{ Reject: record $\bm{x}_k=\bm{x}_{k-1}$ } %\eIf
	}
\end{algorithm}

\medskip

It is worth pointing out that the benefits of this implementation of the Zig-Zag Sampler, i.e.\ working with the sum of rates, are not only applicable to the NuZZ algorithm. The approach can be applied to other Zig-Zag-based algorithms and PDMPs, and in fact it synergises particularly well with the Zig-Zag algorithm with control variates, as showed in Appendix \ref{proofZZCVbound}. 

Even though the efficient implementation has the same theoretical cost as the original (without taking numerical approximations into consideration), our experiments suggest that the efficient implementation may be significantly more efficient in practice.

\subsection{Reusing information}

As integration over a similar range is repeated at every root finding iteration, we have modified the integration routine so that information can be shared through subsequent iterations of Brent's method. 

The QAGS routine performs a 10--21 Gauss-Kronrod numerical integral on the interval $[0,\widetilde{\tau}]$ at each root finding iteration, where $\widetilde{\tau}$ is the current proposed root. If the estimated integration error is greater than the threshold $\varepsilon_{\mathrm{Int}}$, the interval $[0,\widetilde{\tau}]$ is split into subintervals, on which new 10--21 numerical integrals are computed, and so on recursively. 

Note that the value of the numerically approximated integral of the true rate $\Lambda$ on $[0,\widetilde{\tau}]$ is equal to the value of the exact integral of $\widetilde{\Lambda}$, the piecewise linear approximation of the rate $\Lambda$ obtained joining the nodes of the numerical integral.

If at a given iteration of the root finding algorithm the following inequality holds,
\begin{equation*}
g(\widetilde{\tau})= \int_0^{\tilde{\tau}} \widetilde{\Lambda}(s) \, \dint s - R < 0 \, ,
\end{equation*}
then the quantities computed between $0$ and $\widetilde{\tau}$ can be saved and reused during the next iteration of the root finder. At the next iteration, the new proposed root would be $\widetilde{\tau}'$, and the objective function can be decomposed as
\begin{align*}
g(\widetilde{\tau}') & = \int_0^{\tilde{\tau}'} \widetilde{\Lambda}(s) \, \dint s -R \\
& = \int_0^{\tilde{\tau}} \widetilde{\Lambda}(s) \, \dint s + \int_{\tilde{\tau}}^{\tilde{\tau}'} \widetilde{\Lambda}(s) \, \dint s -R \, ,
\end{align*}
where now we only have to compute the second integral, which typically requires fewer resources to achieve the desired level of precision. 
%Because of this, the Gauss-Kronrod routine's cost was decreased from the original 10-21 points to a 7-15 points one.
Because of this, the Gauss-Kronrod number of points was decreased from the original 10--21 to 7--15 points.

\section{Error and convergence analysis} 
\label{sec:err}

In this section we study how the error induced by numerical procedures impacts the quality of the samples which are produced by NuZZ. In order to do this, we will consider NuZZ as a stochastic process that is a perturbation of the original process, the Zig-Zag sampler.

Our aim is to show that the ergodic averages which we calculate via NuZZ are close to the true equilibrium expectations of interest, with an error that can be quantified, and whose dependence on the numerical tolerances can be made transparent. 

However, the continuous time NuZZ process is not Markov, as numerical integration induces a mesh on the deterministic part of the trajectory that depends on both end-points 
It is possible to enlarge the state space to make the NuZZ process Markov, but calculations quickly become unwieldy.
Therefore we will follow an indirect route using the skeleton chain. 

Recall that the full Zig-Zag process $\mathbf{Z}_t = (\mathbf{X}_t,
\mathbf{V}_t)$ is continuous-time, however the differential equations
\eqref{linearDyn} describing its behaviour in between velocity switches are trivially solvable and so the full process can be reconstructed from the skeleton chain of points at which velocities change, $\mathbf{Y}_k = (\mathbf{X}_{T_k}, \mathbf{V}_{T_k})$. Explicitly, 
\begin{equation}
t \in [T_k, T_{k+1}) \Rightarrow \mathbf{X}_{t} =
 \mathbf{X}_{T_{k}} + (t-T_{k}) \mathbf{V}_{T_{k}} \, .
\label{interp}
\end{equation}
Our strategy for bounding the impact of numerical errors on ergodic averages therefore involves bounding the discrepancy between the skeleton chain from the exact process, and the skeleton chain from the NuZZ process. One can then reconstruct ergodic averages via interpolation using~\eqref{interp}.

\subsection{Results}

Our final result, contained in Theorem \ref{th:mu_ergodic_measure}, can be summarised with the following corollary. %, though compactness is not quite necessary, as explained at the end of this section.
Let $\varepsilon_\mathrm{int}, \varepsilon_\mathrm{Bre}$ be the user-defined error tolerances for the numerical (QAGS) integration routine and for Brent's method. 

Let us assume now (we will motivate these assumptions later) that
\begin{enumerate}
    \item $\eta_\mathrm{int} \leqslant \varepsilon_\mathrm{int}$, i.e.\ the true numerical integration error is always smaller than the tolerance,
    
    \item Every element of the Hessian of $\log \pi(\mathbf{x})$ is globally bounded from above by the constant $M$.
    
    \item The skeleton chain $\mathbf{Y}_k$ of the Zig-Zag process.
without numerical errors, $\mathbf{Z}_t$, is Wasserstein geometrically ergodic. %with constants $\rho$ and $C$,
    
    \item The skeleton chain $\widetilde{\mathbf{Y}}_k$ of the Zig-Zag process with numerical errors (the NuZZ process), $\widetilde{\mathbf{Z}}_t$, is Wasserstein geometrically ergodic.

    \item The total refreshment rate $\Gamma$ is strictly positive.
\end{enumerate}
Then we will be able to obtain the following result.

\begin{corollary}
\label{cor:final_results2}
Let $\mathbf{Z}_t$ be the exact Zig-Zag process with stationary measure $\mu$, and let $\widetilde{\mathbf{Z}}_t$ be the NuZZ process. There exists a probability measure $\widetilde{\mu}$ on $E$ such that for all functions $f$ bounded and Lipschitz, the limit
\begin{equation}
    \lim_{t \to \infty} \int_0^t f(\widetilde{\mathbf{z}}_s) \, \dint s
\end{equation}
exists and takes the value
\begin{equation}
    \widetilde{\mu}(f) = \int_E f(\widetilde{\mathbf{z}}) \, \widetilde{\mu}(\dint \widetilde{\mathbf{z}}) \, .
\end{equation}
Moreover, defining the Kantorovich-Rubinstein distance as
\begin{equation}
    d_{KR}(\mu_1, \mu_2) = \sup_{\|f\|_\infty \le 1, \, |f|_{Lip} \le 1} | \mu_1(f) - \mu_2(f)| \, , 
\end{equation}
it holds that
\begin{equation}
\label{eq:order_bound}
    d_{KR}(\mu, \widetilde{\mu}) \leq \mathcal{O} \left (\Gamma^{-\frac{1}{2}} \max (\varepsilon_\mathrm{Tot}, \sqrt{M d \, \varepsilon_\mathrm{Tot}}) \right ) \, ,
\end{equation}
where $\varepsilon_\mathrm{Tot}$ is the sum between the numerical tolerance for the integration routine, $\varepsilon_\mathrm{int}$, and the tolerance for Brent's method, $\varepsilon_\mathrm{Bre}$.
%
%There exists constants $M_1, M_2, M_3 > 0$ which depend only on the ergodicity properties of $\mathbf{Z}_t$ and $\widetilde{\mathbf{Z}}_t$, and a class $\mathcal{F}$ which includes (but is not limited to) bounded Lipschitz functions with compact domains, such that for every $f : E \to \mathbb{R}$, $f \in \mathcal{F}$, 

%\ifdefined\doublecolumn

%\begin{align}
%    \left| \lim_{t \to +\infty} \frac{1}{t} \int_0^t f(\mathbf{z}_s) \, \dint s - \lim_{t\rightarrow +\infty} \frac{1}{t} \int_0^t f(\widetilde{\mathbf{z}}_s) \, \dint s \right| \nonumber \\
%    %
%    \leqslant M_3 \left(M_1 |f|_\mathrm{Lip} + M_2 \| f \|_{\infty} \right)
%\end{align}

%\else

%\begin{equation}
%    \left| \lim_{t \to +\infty} \frac{1}{t} \int_0^t f(\mathbf{z}_s) \, \dint s - \lim_{t\rightarrow +\infty} \frac{1}{t} \int_0^t f(\widetilde{\mathbf{z}}_s) \, \dint s \right|
    %
%    \leqslant M_3 \left(M_1 |f|_\mathrm{Lip} + M_2 \| f \|_{\infty} \right)
%\end{equation}
%\fi
%where $\widetilde{\mathbf{Z}}_t$ is the NuZZ process, with
%\begin{equation}
%M_3 =  \sqrt{2} \, \Gamma^{-\frac{1}{2}} \left (  (\varepsilon_\mathrm{Bre} + \varepsilon_\mathrm{int})^2 + 16 M d \, (\varepsilon_\mathrm{Bre} + \varepsilon_\mathrm{int}) \right )^{1/2} \, .
%\end{equation}
\end{corollary}

This bound motivates the following considerations.

\begin{remark}
We note that the bound in Lemma \ref{lm:tau_tau_tilde}, which then determines the bound in Corollary \ref{cor:final_results2}, may not be tight, and so it may be possible to use other techniques to obtain similar stability results in the degenerate case in which $\Gamma=0$. It would also be natural to devise a specialised root-finding routine which is adapted to this setting. We leave such specialisations for future work. 
\end{remark}

\begin{remark}
\label{rm:tuning_recom}

Consider the bound in Equation \eqref{eq:order_bound}, and a sequence of model problems in dimension $d$, with Hessian bounds growing as $M \sim d^a$, numerical tolerances tightening as $\varepsilon_\mathrm{Tot} \sim d^{-b}$, and aggregate refreshment rates shrinking as $\Gamma \sim d^{-c}$ with $a, b, c \geqslant 0$. We choose to consider bounded $\Gamma$ as it represents a minimal perturbation to the non-reversible character of the original process. 

In order to have the error bound on the ergodic averages decaying to $0$ as the dimension grows, we must have 
\begin{equation}
\max \left( \frac{c}{2} - b, \frac{a + 1 - b + c}{2} \right) < 0 \, .
\end{equation}
Noting that the cost per unit time of simulating the process grows with $b$, efficient algorithms will seek to make this quantity as small as possible. Since we must have $b > a + c + 1$, it suffices to take $c = 0$ and $b > a + 1$, i.e. let $\Gamma$ be of constant order, and take $\varepsilon_\mathrm{Tot} = o \left( 1 / (Md) \right)$.

One outcome of this argument is to confirm the intuition that for more irregular targets (i.e.\ with larger $M$), a tighter numerical tolerance $\varepsilon_\mathrm{Tot}$ is required in order to faithfully resolve the event times and jump decisions of the process numerically, and hence to accurately reproduce equilibrium expectation quantities.
\end{remark}

In order to produce these results, we start by recalling some facts about
Markov chains and Wasserstein distances that will be of use later.

\subsection{Wasserstein distance}

In \cite{rudolf2018} the authors derive perturbation results for Markov chains which are formulated in terms of the 1-Wasserstein distance between transition kernels. We found that in our calculations, it is more convenient to work instead with the 2-Wasserstein distance. Hence, in this section, we report some basic results on Markov chains and Wasserstein distances, and prove that the relevant results in \cite{rudolf2018} do indeed extend to the 2-Wasserstein case. For brevity, many of the results of this section have been moved to Appendix \ref{app:wasserstein}, as they are only technical conditions and are not used directly in the results that we chose to include in the main body of this work. Let us now define the following.

Let $(E, \mathcal{B}(E))$ be a Polish space with its associated Borel $\sigma$-algebra, and $\mathcal{P}$ the set of Borel probability measures on $E$. 
The 2-Wasserstein distance between $\mu_1,\mu_2 \in \mathcal{P}$ is
\begin{equation}
\label{eq:wasserstein_dist_def}
d_W(\mu_1,\mu_2) = \left ( \inf_{\xi \in \Xi(\mu_1,\mu_2)} \int_{E \times E} \| \mathbf{x}_1-\mathbf{x}_2 \|_2^2 \, \dint{\xi (\mathbf{x}_1,\mathbf{x}_2)} \right )^{\frac{1}{2}} \, ,
\end{equation}
where $\Xi$ is the set of all couplings of $\mu_1$ and $\mu_2$ on $E \times E$ that have $\mu_1$ and $\mu_2$ as marginals, and $\| \cdot \|_2$ is the 2-norm. We will use the following notation for the expectation of a function $f$ with respect to the measure $\mu_1$:
\begin{equation}
    \mu_1 (f) := \int_E f(\mathbf{x}_1) \, \dint \mu_1 (\mathbf{x}_1) \, .
\end{equation}
Let $P:\mathcal{P} \to \mathcal{P}$ be a linear operator and a transition kernel on $(E, \mathcal{B}(E))$ such that measures of interest ($\mu_1$, $\mu_2$ etc.) are elements of $\mathcal{P}$. Then
\begin{equation}
    \left( \mu_1 \right) P(A) = \int_E P(\mathbf{x}_1,A) \, \dint \mu_1 (\mathbf{x}_1) \, , \qquad A \in \mathcal{B}(E)\, .
\end{equation}
This implies that $\delta_{\mathbf{x}_1} P(A) = P(\mathbf{x}_1,A)$, and for a measurable function $f: E \to \mathbb{R}$,
\begin{equation}
    \int_E f(\mathbf{x}_1) \, \dint{(\mu_1 P)(\mathbf{x}_1)} = \int_E \left(P f \right) (\mathbf{x}_1) \, \dint{\mu_1(\mathbf{x}_1)} \, ,
\end{equation}
where
\begin{equation}
    \left(P f \right)(\mathbf{x}_1) = \int_E P(\mathbf{x}_1, \dint \mathbf{x}_2) f(\mathbf{x}_2) \, .
\end{equation}

\begin{definition}[Wasserstein Geometric Ergodicity of Markov kernels]
\label{def:Wass_geom_ergo}
We say that the transition kernel $P$ is Wasserstein geometrically ergodic with constants $(C, \rho) \in [0, \infty) \times [0, 1)$ if, for all $n \in \mathbb{N}$, it holds that
\begin{equation}
\forall \mathbf{x}_1, \mathbf{x}_2 \in E, \quad d_W \left( P^n(\mathbf{x}_1,\cdot), P^n(\mathbf{x}_2,\cdot) \right) \leqslant C \rho^n \cdot \| \mathbf{x}_1 - \mathbf{x}_2 \|_2 \, .
\end{equation}
\end{definition}

\begin{proposition} 
\label{pr:close_2_measures_n_step}
Let $P$ and $\widetilde{P}$ be Markov kernels such that $P$ is Wasserstein geometrically ergodic, and $d_{W}(\delta_\mathbf{x} P,\delta_\mathbf{x} \widetilde{P}) \leqslant \epsilon$ for all $\mathbf{x}$. Let $\pi_n$ and $\widetilde{\pi}_n$ be the measures defined by $\pi_n = \pi_0 P^n$ and $\widetilde{\pi}_n = \widetilde{\pi}_0 \widetilde{P}^n$. It then holds that
\begin{align}
    d_{W}(\pi_n,\widetilde{\pi}_n) \leqslant d_{W}(\pi_0,\widetilde{\pi}_0) C \rho^n + \epsilon C \frac{1-\rho^n}{1-\rho} \, .
\end{align}
\end{proposition}
\begin{proof}
    See Appendix \ref{app:wasserstein}.
\end{proof}

\begin{proposition}
\label{pr:P_unique_pi}
If a Markov kernel $P$ is Wasserstein geometrically ergodic, then it has an ergodic invariant measure $\pi$ with finite second moment, which is therefore unique. Furthermore, for any measure $\nu$ with finite second moment, it holds that $\nu P^n \rightarrow \pi$ in the Wasserstein distance.
\end{proposition}
\begin{proof}
See Appendix \ref{app:wasserstein}.
\end{proof}

We are now ready to discuss results specific to the Zig-Zag process.
%, and in order to do that we briefly recall the relevant notation in the following section.

\subsection{Zig-Zag notation}

Recall that $\mathbf{Z}_t=(\mathbf{X}_t, \mathbf{V}_t)$ is the exact (Zig-Zag) continuous time Markov process defined on the state space $E = \mathcal{X} \times \mathcal{V}$, with stationary measure $\mu (\dint \mathbf{x}, \dint \mathbf{v}) = \pi (\mathbf{x}) \otimes \psi (\mathbf{v})$. Let $T_k$ be the $k$th switching time of the Zig-Zag process, and let $\mathbf{Y}_k = (\mathbf{X}_{T_k}, \mathbf{V}_{T_k})$ for $k=1,...,n$ be the skeleton Markov chain, with stationary measure $\nu (\dint \mathbf{x}, \dint \mathbf{v})$.

%Let us also define the objects $P$ and $\mathbf{U}_k$, which are related to the other mathematical objects reviewed in this section, and will be used in the following proofs.

Let $P:E\rightarrow E$ denote the Markov transition kernel of the skeleton chain $\mathbf{Y}_k$. Intuitively, applying $P$ to a state $(\mathbf{x},\mathbf{v})$ of the process involves sampling a new $\tau$, pushing the dynamics forward from $(\mathbf{x}, \mathbf{v})$ to $(\mathbf{x}+\tau \mathbf{v}, \mathbf{v})$, and sampling a new velocity from the distribution associated with the transition kernel $Q(\mathbf{v}'|\mathbf{x}+\tau \mathbf{v},\mathbf{v})$. 

%For our purposes, it is important to note that $\delta_{\mathbf{y}} P(A) = P(\mathbf{y},A)$.

Let us also define $\mathbf{U}_k = ( \mathbf{Y}_{k-1}, \mathbf{Y}_k)$ on $E \times E$ (or equivalently, $E^2$) as the `extended' skeleton chain, which includes two consecutive states of the $\mathbf{Y}_k$ chain, with invariant measure $\zeta$. As mentioned above, our proof strategy is based on the interpolation of the skeleton chain, and the process $\mathbf{U}_k$ will be a key component.

Lastly, we place tildes above all of the quantities above to denote their perturbed versions under numerical approximation. In this spirit, if $\mathbf{Z}_t$ is the exact Zig-Zag process, $\widetilde{\mathbf{Z}}_t$ is the NuZZ process. If $\mathbf{Y}_k$ is the skeleton chain of the Zig-Zag process, then $\widetilde{\mathbf{Y}}_k$ is the skeleton chain of NuZZ, and if $P$ is the Markov transition kernel that moves the skeleton chain of the exact process one step forward, then $\widetilde{P}$ is its numerical equivalent. The same applies to the other relevant objects. We assume that all of these quantities exist, and we will address their properties below. 
In particular, in the next section we will describe the analytical form that numerical errors take for the NuZZ algorithm, which will be useful in bounding discrepancies between the exact and numerical approximation of the objects described above.

For notational simplicity, we use the unlabelled norm $\| \cdot \|$ to represent a 2-norm, and norms defined on elements of the extended state space $E$ and $E \times E$ (denoted as $E^2$) are defined as
\begin{equation}
\label{eq:2-norm-E}
    \| \mathbf{z} \|_E = \| \mathbf{x} \|_2 + \| \mathbf{v} \|_2 \, ,
\quad
%\label{eq:2-norm-E-square}
    \| \mathbf{u} \|_{E^2} = \| (\mathbf{x}, \mathbf{x}') \|_E + \| (\mathbf{v}, \mathbf{v}') \|_E \, .
\end{equation}

\subsection{Perturbation bound for the skeleton chain}
\label{err:numericalError}

Under appropriate conditions that we detail below, a sufficient condition for the processes $\mathbf{Z}_t$ and $\widetilde{\mathbf{Z}}_t$ to produce comparable ergodic behaviour is that $\mathbf{Z}_t$ contracts towards equilibrium sufficiently rapidly, and that the transition kernels $P$ and $\widetilde{P}$ should be close to each other in a uniform sense. 
%By analysing the numerical error, we will show in Proposition \ref{prop:wass_dist_P_Ptilde} that the latter property holds for the NuZZ.
By construction, the latter property holds for the NuZZ if on average, (i) the new locations $\mathbf{x}', \widetilde{\mathbf{x}}'$ are close to one another, and (ii) the new velocities $\mathbf{v}', \widetilde{\mathbf{v}}'$ are close to one another. 
By analysing the numerical error, we will show in Proposition \ref{prop:wass_dist_P_Ptilde} that this holds. 
%in this case.
%We confirm this in the following lemmas.

Recall that in Section \ref{sec:integration} we described how NuZZ relies on two different numerical methods to solve Equation \eqref{eq:rootfunct}: an adaptive integration routine and Brent's method, each requiring a user-set tolerance $\varepsilon_{\mathrm{int}}$ and $\varepsilon_{\mathrm{Bre}}$, respectively. 

In our specific setting, fixing $R$ at the beginning of each MCMC iteration, the error at each iteration of Brent's method can be written as
\begin{equation}
\label{errorbrent}
\eta_{\mathrm{Bre}} = \int_0^{\tilde{\tau}} \widetilde{\Lambda}^\mathbf{z}(s) \, \dint{s} -R \, ,
\qquad R \sim \mbox{Exp}(1) \, ,
\end{equation}
where $\tilde{\Lambda}$ is the piecewise polynomial approximation to $\Lambda$ induced by the integration routine, and $\widetilde{\tau}$ $(\neq \tau)$ is the root currently returned by the algorithm. Using Equation \eqref{eq:sellketGeneral} and adding and subtracting terms to Equation \eqref{errorbrent} we obtain the following error decomposition

\ifdefined\doublecolumn
\begin{align*}
\eta_{\mathrm{Bre}} & = \underbrace{\int_0^{\tilde{\tau}} \widetilde{\Lambda}^\mathbf{z}(s) \, \dint{s}
	- \int_0^{\tilde{\tau}} \Lambda^\mathbf{z}(s) \, \dint{s}}_{\eta_{\mathrm{int}}} \\
& \quad + \underbrace{\int_0^{\tilde{\tau}} \Lambda^\mathbf{z}(s) \, \dint{s} -
	\int_0^\tau \Lambda^\mathbf{z}(s) \, \dint{s}}_{\eta_{\mathrm{root}}} \, .
\end{align*}
\else
\begin{equation*}
\eta_{\mathrm{Bre}} = \underbrace{\int_0^{\tilde{\tau}} \widetilde{\Lambda}^\mathbf{z}(s) \, \dint{s}
	- \int_0^{\tilde{\tau}} \Lambda^\mathbf{z}(s) \, \dint{s}}_{\eta_{\mathrm{int}}}
+ \underbrace{\int_0^{\tilde{\tau}} \Lambda^\mathbf{z}(s) \, \dint{s} -
	\int_0^\tau \Lambda^\mathbf{z}(s) \, \dint{s}}_{\eta_{\mathrm{root}}} \, .
\end{equation*}
\fi
The term $\eta_{\mathrm{int}}$ represents the component of the total Brent error deriving from the numerical integration routine, while the term $\eta_{\mathrm{root}}$ represents the component of the total Brent error deriving from the fact that, conditional on integrating on the exact rate, the integral is calculated up to $\widetilde{\tau}$ instead of to $\tau$.

Let us now make an assumption on the smoothness of the event rates of the process.
\begin{assumption}
\label{as:eta_bounded_epsilon}
The function $\Lambda$ is 
such
%sufficiently smooth 
that the numerical error $\eta_\mathrm{int}$ is always bounded by the given tolerance $\varepsilon_{\mathrm{int}}$.
\end{assumption}

In our specific case, the integration routine is adaptive, and it stops improving the result when the size of the estimated $\eta_{\mathrm{int}}$, the integration error between the two layers of Gauss-Kronrod nodes, is below the tolerance threshold $\varepsilon_{\mathrm{int}}$. An irregular $\Lambda$ with very tall and thin spikes can in principle fool the integrator, but the QAGS routine places more nodes around areas of high variation of the target, to minimise the impact of the irregularities on the final value of the integral.
Therefore, only a very pathological $\Lambda$, e.g.\ one with very large, or infinite, values of certain derivatives are likely to cause Assumption \ref{as:eta_bounded_epsilon} to be false, and hence we expect this to hold for most densities
arising in statistical applications, including all targets studied in this work. Let us now proceed with the discussion.

Brent's method will continue iterating until $|\eta_{\mathrm{Bre}}| \le \varepsilon_{\mathrm{Bre}}$. If Assumption \ref{as:eta_bounded_epsilon} holds, the root error $\eta_{\mathrm{root}}$ at the last iteration of Brent's method can be bounded as
\begin{equation*}
|\eta_{\mathrm{root}}| = |\eta_{\mathrm{Bre}} - \eta_{\mathrm{int}} |
\le \varepsilon_{\mathrm{Bre}} + \varepsilon_{\mathrm{int}} \, .
\end{equation*}
Effectively, as $\varepsilon_{\mathrm{Bre}}, \varepsilon_{\mathrm{int}} \to 0$, then $\eta_{\mathrm{root}} \to 0$. This does not guarantee that $\widetilde{\tau} \to \tau$, which will be the subject of Lemma \ref{lm:tau_tau_tilde}.

\begin{lemma}
\label{lm:tau_tau_tilde}
Let $\tau$ be the exact switching time found by solving Equation \eqref{eq:rootfunct}, and $\widetilde{\tau}$ be the approximated switching time obtained by solving Equation \eqref{errorbrent}. 
The difference between the exact and numerically approximate switching time can be bounded (deterministically) as
\begin{equation}
\label{eq:tau_tau_tilde}
    |\tau - \widetilde{\tau}| \leqslant \frac{\varepsilon_\mathrm{Bre} + \varepsilon_\mathrm{int}}{\Lambda_\mathrm{min}} \, .
\end{equation}
\end{lemma}

\begin{proof}
See Appendix \ref{app:perturbation_bounds}.
\end{proof}

Note that when implemented without a refreshment procedure, Equation \eqref{eq:tau_tau_tilde} can diverge as $\Lambda_\mathrm{min}$ may degenerate to $0$.
The degradation of this bound corresponds to the possibility of serious discrepancies between the inter-jump times for the exact and numerical chain, which is a concrete possibility for PDMPs which have vanishing event rates over nontrivial parts of the state space. This problem can be circumvented by taking a baseline position-independent refreshment rate of $\gamma_i(\mathbf{z}) \equiv \gamma_i>0$, with $\Lambda_{\mathrm{min}} = \Gamma = \sum \gamma_i = d \gamma$. 

One practical interpretation of this bound is that if the numerical tolerances are high, then a large value of $\Gamma$ is necessary in order to avoid large discrepancies in computation of the inter-jump times. Likewise, as the numerical tolerances become smaller, the baseline refreshment rate $\Gamma$ can also be taken smaller.

Let us now prove the second component needed for Proposition \ref{prop:wass_dist_P_Ptilde}. 
Intuitively, if $\tau$ and $\widetilde{\tau}$ are close to each other, then the Zig-Zag process and NuZZ both starting from $(\mathbf{x}, \mathbf{v})$ will remain close at $(\mathbf{x} + \tau \mathbf{v}, \mathbf{v})$ and $(\mathbf{x} + \widetilde{\tau} \mathbf{v}, \mathbf{v})$ respectively. With the following lemma, we ensure that when the two processes reach their respective new switching points, the new velocities which are sampled according to the transition kernel $Q(\mathbf{v}'|\mathbf{x},\mathbf{v})$ are likely to remain the same.

\begin{lemma}
\label{lem:lip_jump}
Assume that the entries of the Hessian of $\log \pi(\mathbf{x})$ are uniformly bounded in absolute value by a constant $M$. Assume also that the process is implemented with coordinate-wise refreshment rates of $\gamma_i > 0$, with sum $\Gamma = \sum_i \gamma_i = d\gamma$. Then, for all $(\mathbf{x}, \mathbf{v}, \mathbf{v}')$, it holds that
\begin{align*}
    \lvert Q(\mathbf{v}'|\mathbf{x},\mathbf{v}) - Q(\mathbf{v}'|\mathbf{x}',\mathbf{v}) \rvert \leqslant 2 \Gamma^{-1} M d^{3/2} \|\mathbf{x} - \mathbf{x}'\|_2.
\end{align*}
\end{lemma}

\begin{proof}
See Appendix \ref{app:perturbation_bounds}.
\end{proof}

We can now use Lemma \ref{lm:tau_tau_tilde} and \ref{lem:lip_jump} to show that the the transition kernels $P$ and $\widetilde{P}$ are close to each other.

\begin{proposition}[Wasserstein closeness of $P$ and $\widetilde{P}$]
\label{prop:wass_dist_P_Ptilde}
Let $\widetilde{P}$ be the Markov transition kernel for NuZZ, parametrised by the numerical tolerances $\varepsilon_\mathrm{Bre}$ and $\varepsilon_\mathrm{int}$. Let $\Gamma > 0$ be the baseline switching rate of both ZZS and NuZZ.
For any point $\mathbf{z} = (\mathbf{x},\mathbf{v}) \in E$, it holds that
\begin{align*}
    \label{wasserstein_distance}
    d_W \left( \delta_\mathbf{z} P, \delta_\mathbf{z} \widetilde{P} \right) \leqslant \epsilon,
\end{align*}
% $P \left( \mathbf{z} \to \cdot \right), \widetilde{P} \left( \mathbf{z} \to \cdot \right)$
where $d_W$ is the 2-Wasserstein distance as defined in Equation \eqref{eq:wasserstein_dist_def}, and

%\begin{equation}
%    \epsilon = \frac{\varepsilon_\mathrm{Bre} + \varepsilon_\mathrm{int}}{ \gamma } \sqrt{d} \, .
%\end{equation}

\begin{equation}
    \epsilon = \left ( 2 \Gamma^{-1} d (\varepsilon_\mathrm{Bre} + \varepsilon_\mathrm{int})^2
    + 32 \Gamma^{-1} M d^2 (\varepsilon_\mathrm{Bre} + \varepsilon_\mathrm{int}) \right )^{1/2}
\label{eq:epsdef}
\end{equation}

\end{proposition}

\begin{proof}
See Appendix \ref{app:perturbation_bounds}.
\end{proof}

\subsection{Invariant Measures and Convergence}

It is challenging to make statements on the asymptotic law $\widetilde{\mu}$ of $\widetilde{\mathbf{Z}}_t$, as NuZZ is not a Markov process. While there are ways of augmenting the state space such that $\widetilde{\mathbf{Z}}_t$ is a sub-component of a Markov process, we have found that in our case these approaches are somewhat opaque and algebraically complicated. Moreover, establishing the a priori existence of a stationary distribution for the numerical process is in fact not crucial for Monte Carlo purposes. With this in mind, we chose to pursue results directly in terms of ergodic averages along the trajectories of our processes. A byproduct of our approach is that we will be able to describe these ergodic averages as spatial averages against a probability measure $\widetilde{\mu}$, which encapsulates the long-run behaviour of the process $\widetilde{\mathbf{Z}}_t$, without requiring the strict property of being an invariant measure for the process.

In this Section we will use Proposition \ref{prop:wass_dist_P_Ptilde} to derive results about the stationary measures of $\mathbf{Y}_k, \widetilde{\mathbf{Y}}_k, \mathbf{U}_k$ and $\widetilde{\mathbf{U}}_k$. We will then define the interpolation operator $\mathcal{J}f$ and use it, together with the results on the stationary measures, to produce Theorem \ref{th:mu_ergodic_measure}, the main result of this section.

Recall that $P$ is defined as the Markov kernel corresponding to the skeleton chain of the exact Zig-Zag Process.
If the Zig-Zag process $\mathbf{Z}_t$ converges geometrically to the stationary distribution, then one generally expects that this assumption should also hold for the skeleton chain $\mathbf{Y}_k$. However, we do not know of existing results which guarantee this implication.

\begin{assumption}[Wasserstein geometric ergodicity for $\mathbf{Y}_k$]
\label{as:Wergodic_Y_and_P}
The transition kernel $P$ is Wasserstein geometrically ergodic with constants $\rho \in [0,1)$ and $C\in(0,\infty)$.
\end{assumption}

Moreover, let us make the following additional assumption.

\begin{assumption}[Wasserstein geometric ergodicity for $\widetilde{\mathbf{Y}}_k$]
\label{as:Wergodic_Y_and_P_tilde}
%The invariant measure $\widetilde{\nu}$ of the Markov transition kernel $\widetilde{P}$ is unique.
The transition kernel $\widetilde{P}$ is Wasserstein geometrically ergodic.
\end{assumption}

Once again, as $\widetilde{P}$ relates to the skeleton chain rather than the continuous time process, this assumption is not particularly strict. Note that Assumptions \ref{as:Wergodic_Y_and_P} and \ref{as:Wergodic_Y_and_P_tilde} imply that $\nu$ and $\widetilde{\nu}$ are unique.

%As the choices of $\Gamma, \varepsilon_{\mathrm{Int}}$ and $\varepsilon_{\mathrm{Brent}}$ determine how noisy the numerically approximated process is with respect to the exact process, it is natural to assume that they also influence the value of the constants $C$ and $\rho$. However, we expect this influence to be only weak, as the numerical tolerances are small, and the measures $\nu$ and $\widetilde{\nu}$ remain close to each other.

Under these assumptions, the following bound on the invariant measures of the skeleton chains holds.

\begin{proposition}[Wasserstein distance between skeleton chains]
\label{prop:WdistYchain}
Let $\nu$ and $\widetilde{\nu}$ be the unique invariant distributions of $\mathbf{Y}_k$ and $\widetilde{\mathbf{Y}}_k$.
If $P$ satisfies Assumption \ref{as:Wergodic_Y_and_P}, then:
\begin{equation}
    d_W(\nu,\widetilde{\nu}) \leqslant \frac{C \epsilon}{1-\rho} \, ,
\end{equation}
where $C$ and $\rho$ are the constants defined in Assumption \ref{as:Wergodic_Y_and_P}.
\end{proposition}

\begin{proof}
    See Appendix \ref{app:perturbation_bounds}.
\end{proof}

Proposition \ref{prop:WdistYchain} above makes use of the $\epsilon$ described in Proposition \ref{prop:wass_dist_P_Ptilde}, and thus already provides a bound on the Wasserstein distance. However, the computed bound applies only to the stationary distribution of the skeleton chain. In what follows we try to extend these results on the skeleton chain $\widetilde{\mathbf{Y}}_k$ to the continuous time Zig-Zag process $\mathbf{Z}_t$ by making use of the extended skeleton chain $\mathbf{U}_k$ and the interpolation operator $\mathcal{J}$.

Let us now state a couple of technical results on the extended skeleton chains $\mathbf{U}_k$ and $\widetilde{\mathbf{U}}_k$.

\begin{lemma}
\label{lm:U_Utilde_unique_Pi}
The Markov chains $\mathbf{U}_k$ and $\widetilde{\mathbf{U}}_k$ each have a unique invariant measure.
\end{lemma}

\begin{proof}%[Proof of Lemma \ref{lm:U_Utilde_unique_Pi}]
See Appendix \ref{app:invariant_measure}.
\end{proof}

We can now bound the distance between $\mathbf{U}_k$ and $\widetilde{\mathbf{U}}_k$ in terms of the chains $\mathbf{Y}_k$ and $\widetilde{\mathbf{Y}}_k$ as follows.

\begin{proposition}[Wasserstein distance between $\mathbf{U}_k$ and $\widetilde{\mathbf{U}}_k$]
\label{pr:W_U_Utilde}
Let $\zeta$ and $\widetilde{\zeta}$ be the invariant distributions of the augmented chains $\mathbf{U}_k$ and $\widetilde{\mathbf{U}}_k$. Then:
\begin{equation}
    d_W^2(\zeta, \widetilde{\zeta}) \le \left(1 + 2 C^2 \rho^2 \right) d_W^2(\nu, \widetilde{\nu}) + 2 \epsilon^2 .
\end{equation}
\end{proposition}

\begin{proof}%[Proof of Proposition \ref{pr:W_U_Utilde}]
See Appendix \ref{app:invariant_measure}.
\end{proof}

The remainder of this section shows that if the numerical error is small, the ergodic averages obtained through the measure $\widetilde{\mu}$ will be close to those obtained through the exact stationary measure $\mu$. Moreover, the discrepancy will be bounded from above by a monotone function of the numerical tolerance.

Let us now define the interpolation operator $\mathcal{J}$, which will be of distinct importance in our derivation, and prove two technical lemmas that will be used in Theorem \ref{th:mu_ergodic_measure}.

\begin{definition}[Interpolation operator]
Let $(\mathbf{x},\mathbf{v})$ and $(\mathbf{x}',\mathbf{v}')$ be two consecutive points in the skeleton chain $\mathbf{Y}_k$. Define the interpolation operator $\mathcal{J} : C ( E ) \to C \left( E \times E \right)$ by its action on functions as
\begin{equation}
    ( \mathcal{J} f ) (\mathbf{x}, \mathbf{v}, \mathbf{x}', \mathbf{v}') = \frac{\|\mathbf{x}' - \mathbf{x} \|_2}{\| \mathbf{v} \|_2} \int_0^1 f(\mathbf{x} + (\mathbf{x}' - \mathbf{x}) t, \mathbf{v} ) \, \dint t.
\end{equation}
\end{definition}

The significance of this operator is that given a pair of consecutive points in the skeleton chain and a function $f$, the interpolated function $\mathcal{J}f$ returns the integral of the function $f$ over the period of time during which the continuous-time process moves from the first point to the second point. As such, it is a natural tool for expressing ergodic averages of the continuous-time processes in terms of their skeleton 
chains.

\begin{lemma}
\label{lm:ergodic_avg_helper}
Let $f$ be such that $\mathcal{J}f \in L^1 (\zeta) \cap L^1 (\widetilde{\zeta})$. Then, the ergodic averages for the continuous processes $\mathbf{Z}_t$ and $\widetilde{\mathbf{Z}}_t$ can be given in terms of the chains $\mathbf{U}_k$ and $\widetilde{\mathbf{U}}_k$ and their stationary distributions as
\begin{align}
    \lim_{t \rightarrow \infty} \frac{1}{t} \int_0^t f(\mathbf{z}_s) \, \dint s 
    & = \frac{1}{H} \int_{E^2} ( \mathcal{J} f ) (\mathbf{u}) \, \zeta (\dint \mathbf{u}) \label{eq:erg_avg_u_exact} \\
    \lim_{t \rightarrow \infty} \frac{1}{t} \int_0^t f(\widetilde{\mathbf{z}}_s) \, \dint s 
    & = \frac{1}{\widetilde{H}} \int_{E^2} ( \mathcal{J} f ) (\mathbf{u}) \, \widetilde{\zeta}( \dint \mathbf{u}) \, , \label{eq:erg_avg_u_approx}
\end{align}
where
\begin{equation}
    \lim_{n \to \infty} \frac{T_n}{n} = H  \qquad \lim_{n \to \infty} \frac{\widetilde{T}_n}{n} = \widetilde{H} \, ,
\end{equation}
i.e.\ the average switching time for both the exact and numerical process along a trajectory is finite and constant.
\end{lemma}

\begin{proof}
See Appendix \ref{app:invariant_measure}.
\end{proof}

\begin{remark}
\label{rm:mu_tilde}
Note now that the right-hand side of Equation \eqref{eq:erg_avg_u_exact} can be rewritten as 
\begin{align}
    \frac{1}{H} \int_{E^2} ( \mathcal{J} f ) (\mathbf{u}) \, \zeta (\dint \mathbf{u}) 
    &= \int_{E^2} \nu(\dint \mathbf{y}) P(\mathbf{y}, \dint \mathbf{y}') (\mathcal{J} f) (\mathbf{y}, \mathbf{y}') H^{-1} \\
    &= \int_{E^2} \nu(\dint \mathbf{y}) P(\mathbf{y}, \dint \mathbf{y}') \, \frac{\|\mathbf{x}' - \mathbf{x} \|_2}{\| \mathbf{v} \|_2} \int_0^1 f( (1-t) \cdot \mathbf{y} + t \cdot  \mathbf{y}') \, \dint t H^{-1}\\
    &= \int_{E^3} \varphi (\dint \mathbf{y}, \dint \mathbf{y}', \dint \mathbf{y}'') f(\mathbf{y}'') \\
    &= \int_E f(\mathbf{z}) \, \mu(\dint \mathbf{z})
\end{align}
where
\begin{equation}
\varphi (\dint \mathbf{y}, \dint \mathbf{y}', \dint \mathbf{y}'') = \nu(\dint \mathbf{y}) P(\mathbf{y}, \dint \mathbf{y}') \frac{\|\mathbf{x}' - \mathbf{x} \|_2}{\| \mathbf{v} \|_2} K (\mathbf{y}'' | \mathbf{y}, \mathbf{y}') H^{-1} ,
\end{equation}
where $K ( \cdot | \mathbf{y}, \mathbf{y}')$ is the Markov kernel which selects a point uniformly at random from the line segment joining $\mathbf{y}$ and $\mathbf{y}'$, and $\mu$ is obtained by integrating $\mathbf{y}, \mathbf{y}'$ out of $\varphi$.

In particular, it follows that $\mu$ is the occupation (and stationary, in this case) measure of the exact process $\mathbf{Z}_t$ on the original state-space $E$. Likewise, we can define

\begin{align}
    \frac{1}{H} \int_{E^2} ( \mathcal{J} f ) (\widetilde{\mathbf{u}}) \, \widetilde{\zeta} (\dint \widetilde{\mathbf{u}}) 
    & = \int_{E^2} \widetilde{\nu}(\dint \widetilde{\mathbf{y}}) \widetilde{P}(\widetilde{\mathbf{y}}, \dint \widetilde{\mathbf{y}}') (\mathcal{J} f) (\widetilde{\mathbf{y}}, \widetilde{\mathbf{y}}') \widetilde{H}^{-1} \\
    & = \int_{E} f(\widetilde{\mathbf{z}}) \, \widetilde{\mu}(\dint \mathbf{z})
\end{align}
with
\begin{equation}
\widetilde{\varphi} \left( \dint \widetilde{\mathbf{y}}, \dint \widetilde{\mathbf{y}}', \dint \widetilde{\mathbf{y}}'' \right) = \widetilde{\nu} (\dint \widetilde{\mathbf{y}}) \widetilde{P}(\widetilde{\mathbf{y}}, \dint \widetilde{\mathbf{y}}') \frac{\| \widetilde{\mathbf{x}}' - \widetilde{\mathbf{x}} \|_2}{\| \widetilde{\mathbf{v}} \|_2} K (\widetilde{\mathbf{y}}'' | \widetilde{\mathbf{y}}, \widetilde{\mathbf{y}}') \widetilde{H}^{-1} \,
\end{equation}
in order to characterise the occupation measure of the approximating process.
\end{remark}

We can now prove the main result of this section, which provides an error bound on the ergodic averages which are obtained from the paths of the continuous-time approximate process.
\begin{theorem}[Bound on discrepancy of ergodic averages]
\label{th:mu_ergodic_measure}
Let $f : E \to \mathbb{R}$ be a bounded and Lipschitz function, and let $\mathcal{J}f \in L^1 (\zeta) \cap L^1 (\widetilde{\zeta})$. Then, there exists a probability measure $\widetilde{\mu}$ such that %with probability 1, it holds that the limit
the limit
\begin{align}
\label{eq:mu_tilde}
    \widetilde{\mu} (f) := \lim_{t \rightarrow \infty} \frac{1}{t} \int_0^t f(\widetilde{\mathbf{z}}_s) \, \dint s
\end{align}
exists, and satisfies
\begin{equation}
\label{eq:bound_th1}
    \left| \widetilde{\mu} (f) - \mu(f) \right| 
    \leqslant \frac{ d_W(\zeta, \widetilde{\zeta}) }{\sqrt{d} \widetilde{H}} \cdot \left ( |f|_\mathrm{Lip} \cdot \mathbf{E}_\zeta \left[ \| \mathbf{x}' - \mathbf{x} \|^2 \right]^{\frac{1}{2}} + \| f \|_{\infty} +  | \mu(f) | \right)
\end{equation}
where $\mu$ is the invariant distribution of the exact Zig-Zag Process, and the expectation is the expected square jump distance.
\end{theorem}

\begin{proof}
See Appendix \ref{app:invariant_measure}.
\end{proof}

Intuitively, Theorem \ref{th:mu_ergodic_measure} shows that despite $\widetilde{\mathbf{Z}}_t$ being non-Markovian, there exists a probability measure that describes its sojourn in the state-space, and that integration against this measure will reproduce the ergodic averages of the process. 

The second, and perhaps most important consequence of Theorem \ref{th:mu_ergodic_measure} is to establish a bound between the difference of ergodic averages between the exact process $\mathbf{Z}_t$ and ergodic averages of the approximate process $\widetilde{\mathbf{Z}}_t$, bound that depends on how close the skeleton chains $(\mathbf{Y}_k, \widetilde{\mathbf{Y}}_k)$ via the distance between the extended skeleton chains $(\mathbf{U}_k, \widetilde{\mathbf{U}}_k)$. Hence ergodic averages computed with the occupation measure $\widetilde{\mu}$ of the numerical process $\widetilde{\mathbf{Z}}_t$ will be close to those computed via the exact stationary measure $\mu$ of the exact process, and as the numerical tolerances $\varepsilon_\mathrm{Bre}$ and $\varepsilon_\mathrm{int}$ go to zero, the two ergodic averages coincide.

%We previously defined the class $\mathcal{F}$ to which $f$ belongs indirectly, as the class of bounded and Lipschitz functions s.t.\ $\mathcal{J}f \in L^1(\zeta) \cap L^1(\widetilde{\zeta})$. In order to give an explicit definition of $\mathcal{F}$, one can restrict the attention to the class of bounded and Lipschitz functions with compact support, where the requirement on $\mathcal{J}f$ is certainly satisfied, and is rich enough for the purposes of this work.

The following corollary is a somewhat restrictive but more explicit summary of the results presented in this section.

\begin{corollary}
\label{cor:final_results}
Let $f : E \to \mathbb{R}$ be a bounded and Lipschitz function with compact support. 
Let $\eta_\mathrm{int} \le \varepsilon_\mathrm{int}$, let every element of the Hessian of $\log \pi(\mathbf{x})$ be bounded by the constant $M$, let $\mathbf{Y}_k$ be Wasserstein ergodic with constants $\rho$ and $C$, and let $\widetilde{\mathbf{Y}}_k$ be Wasserstein ergodic.
Then by Theorem \ref{th:mu_ergodic_measure} and Propositions \ref{pr:W_U_Utilde}, \ref{prop:WdistYchain}, \ref{prop:wass_dist_P_Ptilde}, the discrepancy between ergodic averages from the Zig-Zag process, $\mathbf{Z}_t$, and the ergodic averages from NuZZ, $\widetilde{\mathbf{Z}}_t$, can be bounded as
\begin{align}
\label{eq:bound}
    & \left| \int_E f(\mathbf{z}) \, \mu (\dint \mathbf{z}) - \int_E f(\widetilde{\mathbf{z}}) \, \widetilde{\mu} (\dint \mathbf{z}) \right| \leqslant \frac{  M_3 \cdot M_4 }{\widetilde{H}} \left ( |f|_\mathrm{Lip} \cdot \mathbf{E}_\zeta \left[ \| \mathbf{x}' - \mathbf{x} \|^2 \right]^{1/2} + \| f \|_{\infty} +  | \mu(f) | \right)
\end{align}
with
\begin{align}
M_3 &= d^{-\frac{1}{2}} \epsilon = \sqrt{2} \, \Gamma^{-\frac{1}{2}} \left ( (\varepsilon_\mathrm{Bre} + \varepsilon_\mathrm{int})^2 + 16  M d (\varepsilon_\mathrm{Bre} + \varepsilon_\mathrm{int}) \right )^{1/2} \\
M_4 &= \left ( (1 + 2 C^2 \rho^2) \frac{C^2}{(1-\rho)^2} + 2 \right )^{1/2}
\, .
\end{align}
\end{corollary}

\begin{remark}
Note that while the left-hand side of Equation \eqref{eq:bound} is invariant under constant shifts of the form $f \mapsto f + c$, the right hand side is not. As such, one may sharpen the bound by taking the infimum of the existing bound over all such shifts.
\end{remark}

% Or, in simpler terms, see Corollary \ref{cor:final_results2}.

\section{Numerical experiments}
\label{sec:exp}

We now study how the Zig-Zag process compares to other popular MCMC algorithms on test problems with features that are often found in practice.  These features include high linear correlation, different length scales, high dimension, fat tails, and position-dependent correlation structure. 
%We also return to the real House Price dataset from Section \ref{sec:mot}.
Our results are shown in Figure \ref{fig:convergencePlots}.

%\subsection{Methodology}\label{sec:perf}

%\subsubsection{Algorithm choice}

In the following sections, we aim to compare the NuZZ method against three popular algorithms. In the figures that follow, the Random Walk Metropolis method is represented by the black $\bigcirc$ line labelled \emph{RWM} in plots, Hamiltonian Monte Carlo by the green $+$ line labelled \emph{HMC}, and the simplified Manifold MALA~\cite{girolami2011} with SoftAbs metric \cite{betancourt2013} by the red $\bigtriangleup$ line labelled \emph{sMMALA}. We also include a \emph{Direct} Monte Carlo sample (blue $\triangledown$ line) from the target as a term of comparison. Every example we examine has been chosen so that it is easy to obtain a direct Monte Carlo sample, and the marginal distributions are easily accessible.

For each example, we test two versions of NuZZ. 
The yellow $\Diamond$ line labelled \emph{NuZZroot}, represents the performance of NuZZ as if only one gradient evaluation was necessary to obtain each $\tau$, which would be the performance of the algorithm if an analytical solution to Equation \eqref{eq:sellketGeneral} were available. The yellow line therefore represents a lower bound to the performance of NuZZ, and indeed any numerical approximation of the Zig-Zag process.
The blue $\times$ line labelled \emph{NuZZepoch} represents the performance of the NuZZ algorithm when every gradient evaluation necessary to obtain $\tau$ is counted. 
In this context, the current implementation of NuZZ is expensive, as for each root, $\Lambda(\mathbf{x}(s),\mathbf{v})$ has to be evaluated for each root finding iteration, for each node of the numerical integral. 
%Therefore the \emph{NuZZepoch} line can be considered as a worst-case scenario approximation of the Zig-Zag dynamics. 

Lastly, the \emph{icdfZZ} algorithm mentioned in Section \ref{sec:stdNorm} is the exact process where $\tau$ is obtained via cdf inversion, while the \emph{NuZZtol} algorithm from Section \ref{sec:fatTails} is just the \emph{NuZZepoch} run with increased tolerances, as explained below.

\subsection{Tuning}
\label{sec:tuning}

The RWM, MALA and HMC algorithms used in this work are tuned to achieve the optimal acceptance rate of roughly 25\%, 50\% \parencite{roberts2001}, and 65\% \parencite{beskos2013}, with exceptions described in each relevant section.

%One of the many desirable features of Zig-Zag-like algorithms discussed here is that they only have a handful of parameters to tune. 

We run all the NuZZ-based algorithms with absolute tolerance of $\varepsilon_\mathrm{int} = 10^{-10}$ for the integration routine, and $\varepsilon_\mathrm{Bre} = 10^{-10}$ for Brent's method. 

The functions $\gamma_i(\mathbf{x},\mathbf{v})$ can be set to zero if the target distribution satisfies certain properties~\parencite{bierkens2019ergodicity}, or otherwise they can be made small enough such that they have negligible effects on the dynamics. 
Small values of $\gamma_i$ reduce diffusivity and guarantee better results in terms of asymptotic variance \parencite{andrieu2019peskuntierney}. 
As it is unknown what the effect of setting $\Gamma = 0$ for a numerical approximation to the Zig-Zag process, we set $\Gamma = 0.001$ in all of our examples, which is an empirical value small enough (in our case) that the impact on the switching rate is negligible\footnote{Practically speaking, one empirical way to select $\Gamma$ can consist in looking how the largest $\tau$ in a run evolves across runs with different values of $\Gamma$. 
The dynamics can reach furthest when decreasing the value of $\Gamma$ stops causing an increase in the value of the largest $\tau$ obtained from that run.}.

The mixing of Zig-Zag-based algorithms is affected by the velocity vector $\mathbf{v}$ in the same way the mixing of Metropolis-based MCMC algorithms is influenced by the covariance of the transition kernel (or the momentum variable).
Tuning $\mathbf{v}_{\mathrm{init}}$ is discussed in the relevant examples in Section \ref{sec:diffScales} and \ref{sec:linearCorr}.
Wherever possible, we used the values of the theoretical covariance matrix to adapt the proposal covariance, mass matrix, and velocity vectors of all the algorithms. We leave the study of how fast the methods adapt to each example using partial MCMC samples for future work.

\subsection{Performance comparison}

We chose as a measure of algorithm performance the largest Kolmogorov-Smirnov (KS) distance between the MCMC sample and true distribution amongst all the marginal distributions. 
The KS distance is particularly well suited to our situation as the marginal distributions of our test problems are known or easily obtainable. We chose not to utilise the Effective Sample Size (ESS) or other ESS-based metrics, as NuZZ is not an exact algorithm and a high ESS value does not necessarily reflect the quality of the sample (NuZZ with very high tolerances can have excellent ESS but be arbitrarily far from the target distribution). Moreover, the KS distance is still well defined even if the moments of the target are infinite as in Section \ref{sec:fatTails}, unlike the ESS.
Choosing the KS distance also reflects the requirements of practical Bayesian statistics, where each parameter will typically have a point estimate and credible region reported, meaning that a supremum distance such as the KS is appropriate in assessing the worst-case accuracy of such figures.

Formally, the Kolmogorov-Smirnov distance for the $i$th parameter is defined
as
\begin{equation*}
D_i = \sup_{x_i}|\widetilde{F}_i(x_i) - F_i(x_i)| \, , \qquad \qquad i\in\{1,\ldots,d\}
\, ,
\end{equation*}
i.e.\ the supremum of the absolute value of the difference between the exact cdf, $F_i$,  and the empirical cdf constructed from the MCMC sample, $\widetilde{F}_i$. The KS distance also has the advantage of being relatively easy to calculate, and it is an integral probability measure with good theoretical properties~\parencite{sriperumbudur2009}.  The largest KS distance on the marginal distributions, which we standardly refer to as the ``$D$-statistic'', is simply
\begin{equation*}
D = \max_{i=1,\ldots,d} D_i \, .
\end{equation*}
Since the empirical cdf is calculated via Monte Carlo simulation, we ran each algorithm 40 times and reported the mean and two-sided $95\%$ region of the KS distance in plots such as Figure \ref{fig:convergencePlots}.

In order to meaningfully compare the algorithms studied in this work, we gave ourselves a fixed computational budget. 
We defined one epoch as ``a complete passage of the algorithm through the whole dataset".
For each model, we gave ourselves a computational budget of $6\times 10^6$ epochs, and counted one epoch every time the full gradient or one likelihood evaluation is completed. For sMMALA, we ignored the computational cost of computing the Hessian matrix, since depending on the problem, the actual cost of this will be highly variable.

As each algorithm that we studied has a different computational cost, each of them must return a sample of different size. 
In order to meaningfully compare the $D$-statistic across samples of different size, we divided the sample from each algorithm in 40 equal batches, which is a large enough number to give a detailed account of convergence, but not too large for our computational resources. We then produced the lines in Figure \ref{fig:convergencePlots} by calculating the $D$-statistic on subsamples formed by an increasing number of batches, with the first value of each line calculated only on one batch, and the last value of each line calculated on the whole sample. 
The Zig-Zag-based algorithms that we study in this work return a trajectory which is represented by a series of switching points. These trajectories then have to be uniformly subsampled to obtain a sample from the distribution of interest $\pi(\mathbf{x})$. In our experiments we extracted $6 \times 10^6$ samples from each trajectory.

%\FloatBarrier

%\newgeometry{left=.8cm,right=.8cm, top=.1cm, bottom=.1cm} %affects whole page

\begin{figure}
	\begin{subfigure}[t]{.5\textwidth}
		\centering
		\includegraphics[width=.95\linewidth]{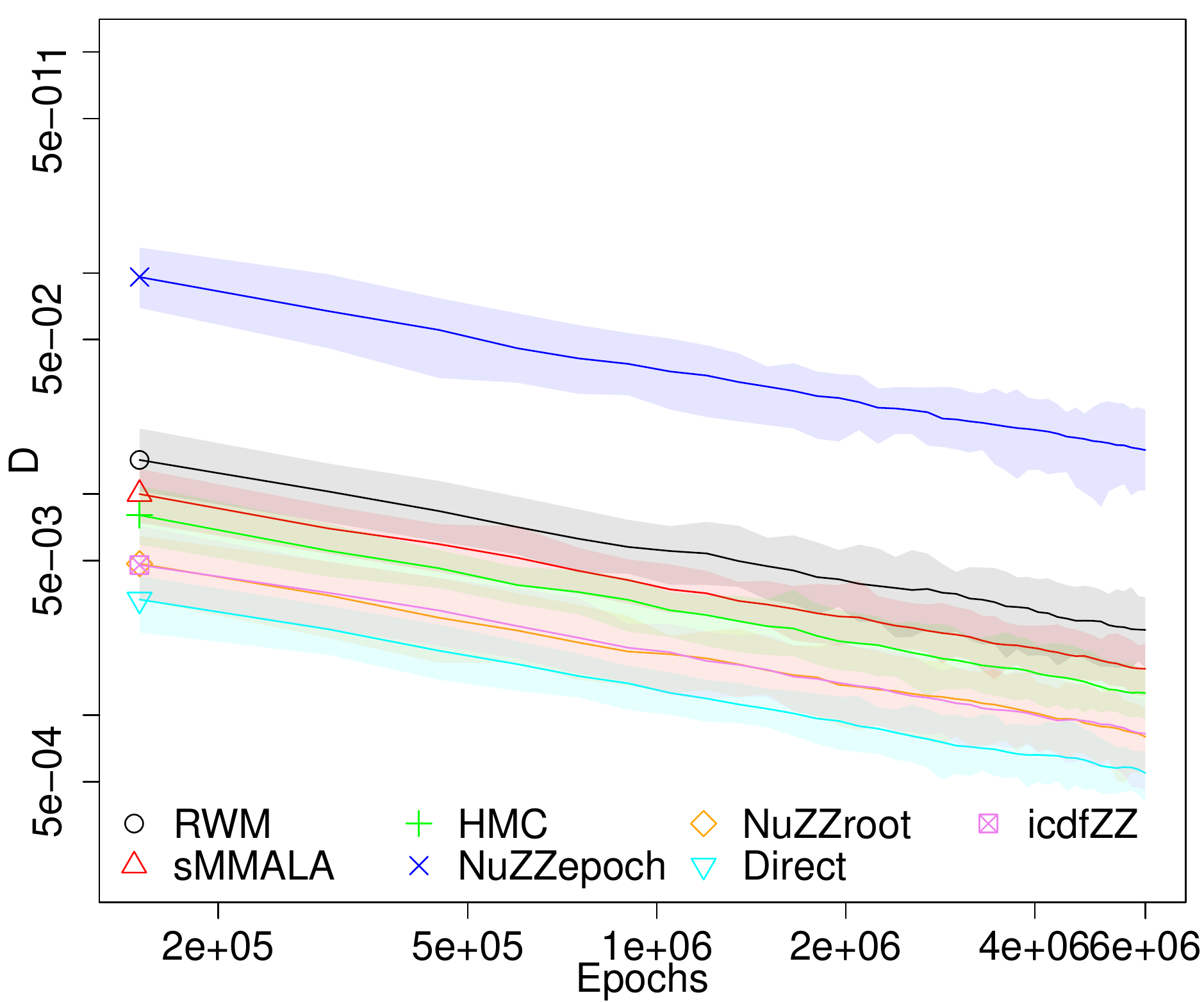}
		\caption{10-$d$ standard Normal.}
		\label{stdNorm10dconv}
	\end{subfigure}%
	\begin{subfigure}[t]{.5\textwidth}
		\centering
		\includegraphics[width=.95\linewidth]{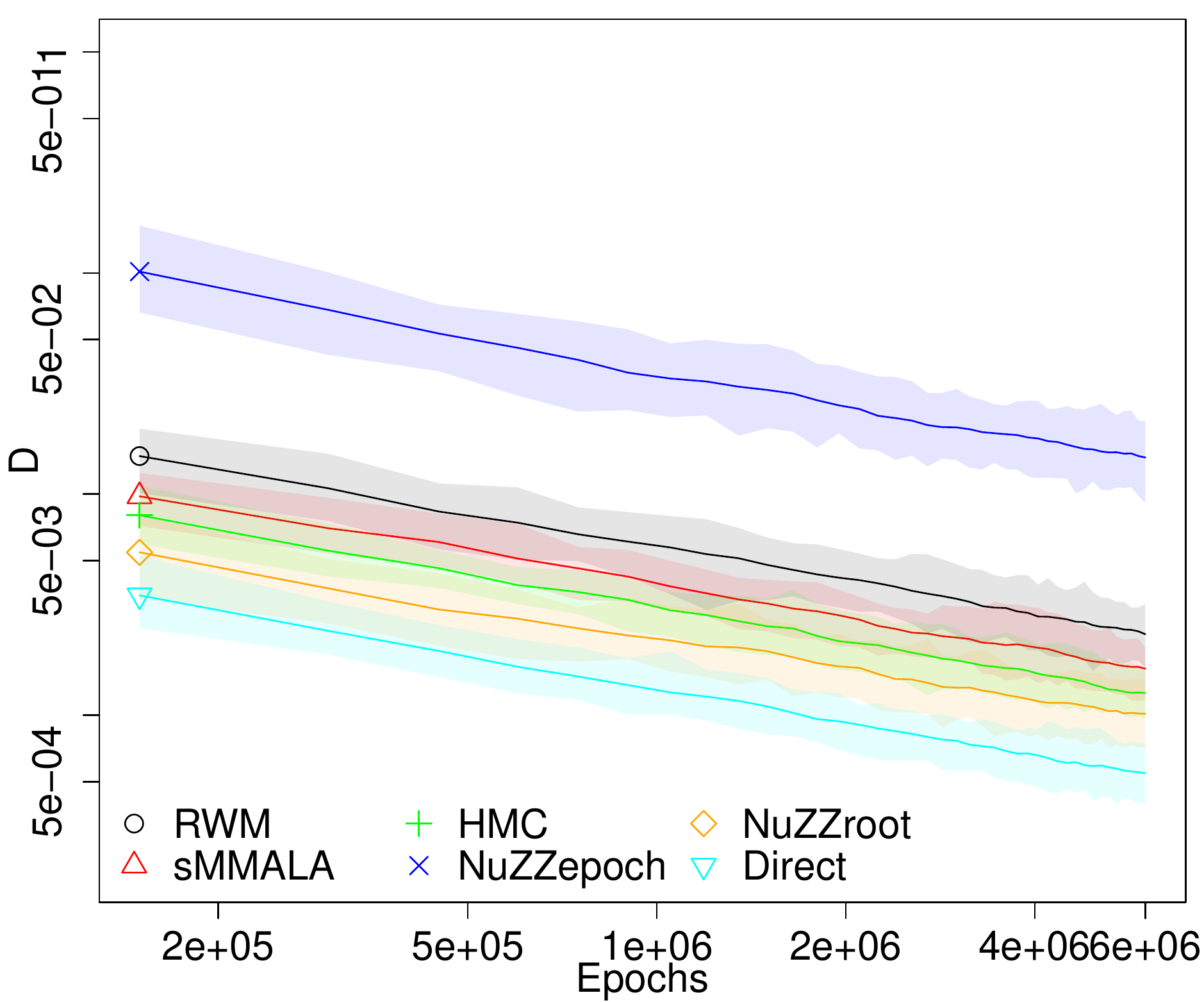}
		\caption{10-$d$ Normal, diagonal covariance with values $[1^2,2^2, \ldots,10^2]$.}
		\label{scaleNorm10dconv}
	\end{subfigure}
	
	\begin{subfigure}[t]{.5\textwidth}
		\centering
		\includegraphics[width=.95\linewidth]{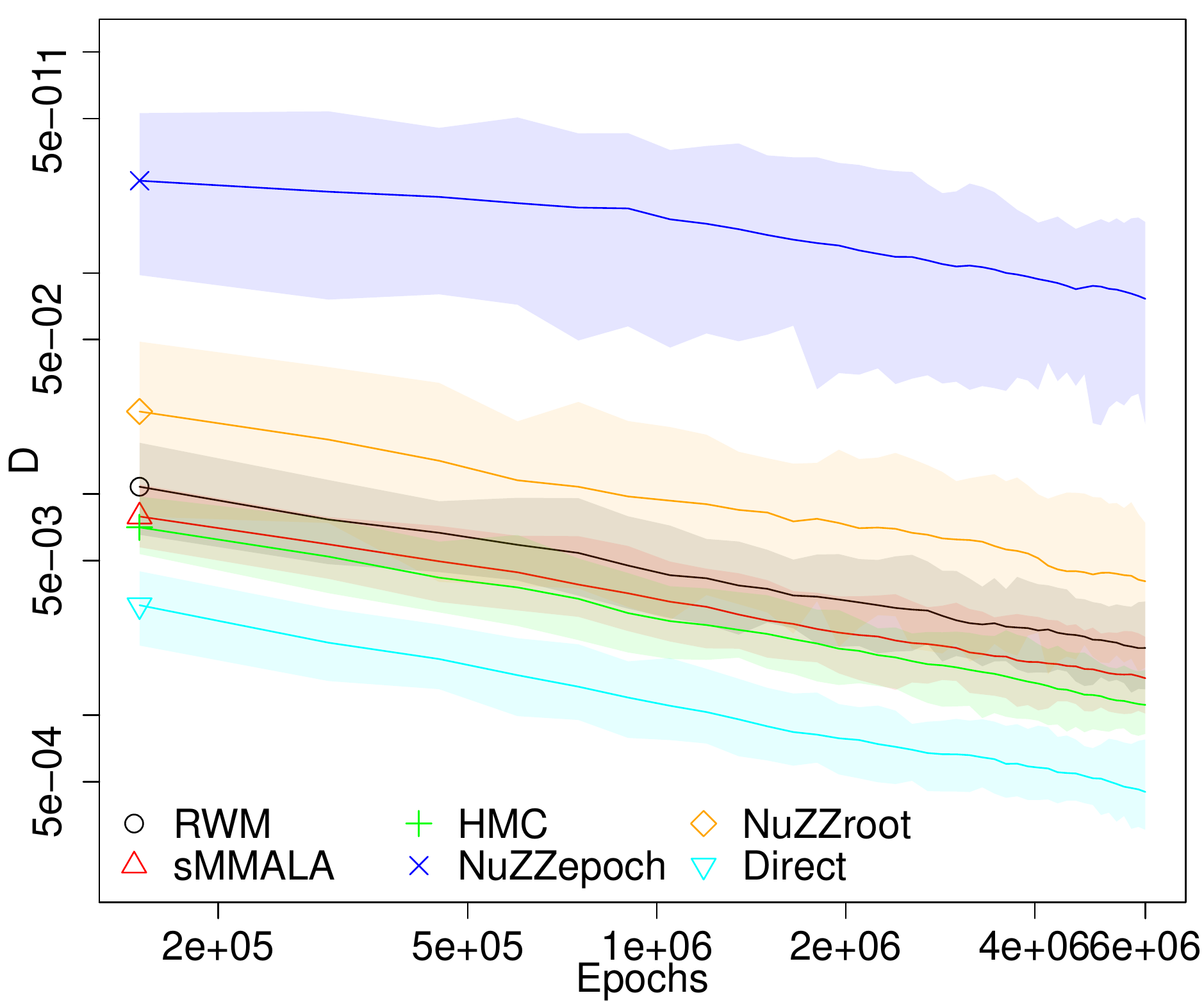}
		\caption{10-$d$ Normal with linear correlation.}
		\label{corrNorm10dconv}
	\end{subfigure}
	\begin{subfigure}[t]{.5\textwidth}
		\centering
		\includegraphics[width=.95\linewidth]{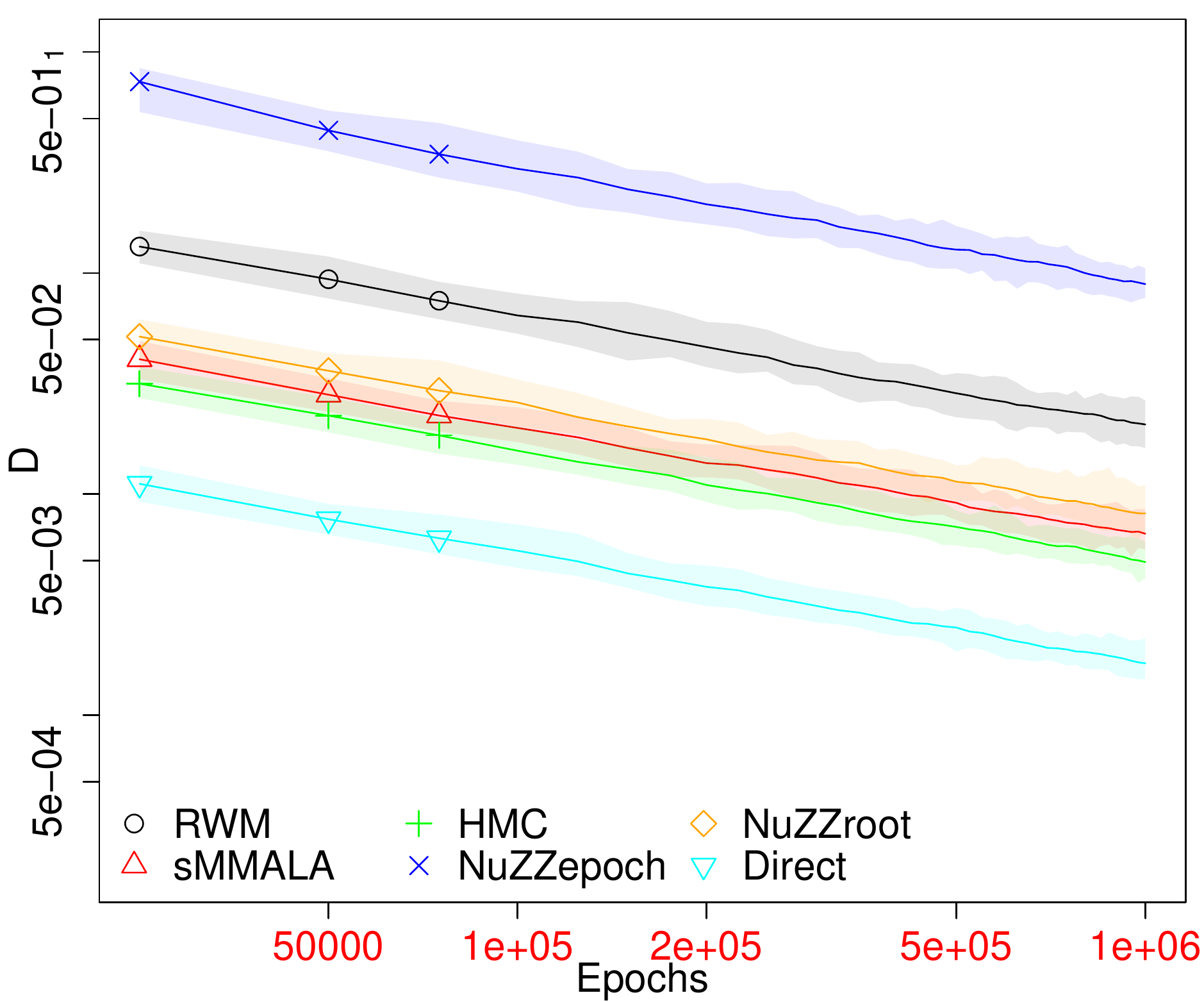}
		\caption{100-$d$ standard Normal.}
		\label{dimNorm100dconv} %8.3 x 6.81 inch
	\end{subfigure}
	
	\begin{subfigure}[t]{.5\textwidth}
		\centering
		\includegraphics[width=.95\linewidth]{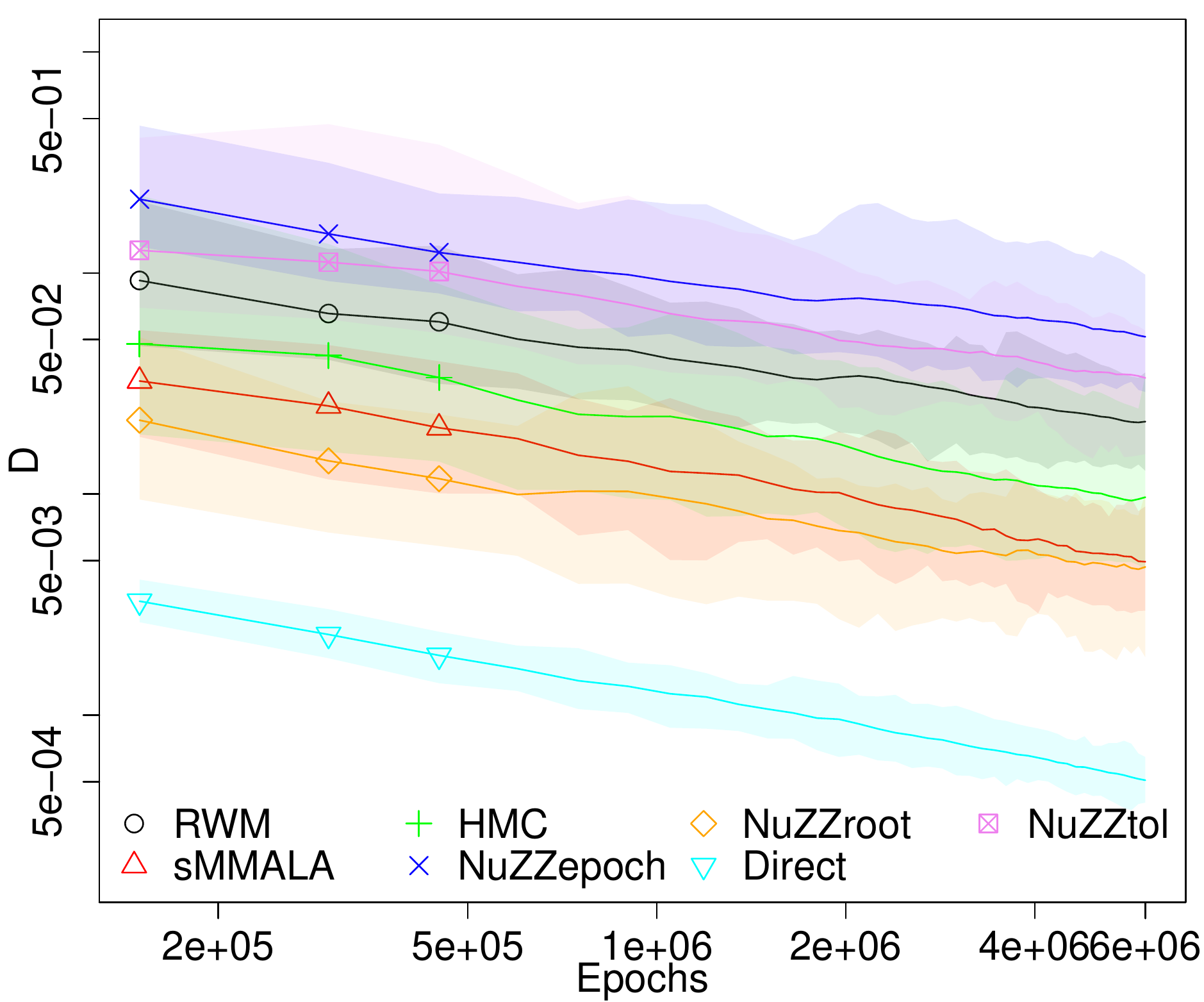}
		\caption{10-$d$ Student-\textit{t}, spherical covariance, one degree of freedom.}
		\label{tailsT10dconv}
	\end{subfigure}
	\begin{subfigure}[t]{.5\textwidth}
		\centering
		\includegraphics[width=.95\linewidth]{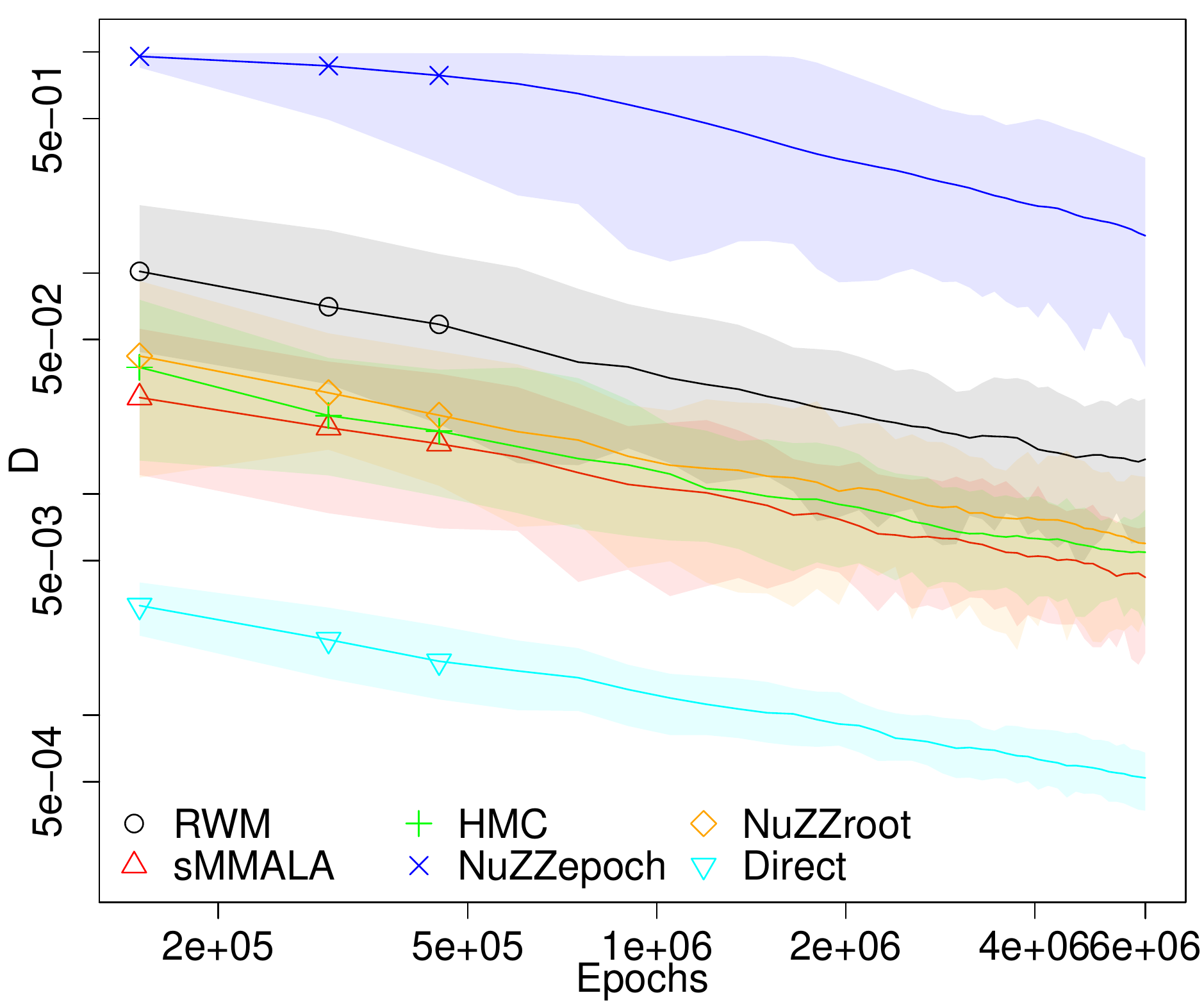}
		\caption{10-$d$ Hybrid Rosenbrock.}
		\label{curvRosen10dconvpch}
	\end{subfigure}

	\caption{Convergence plots. $D$ represents the largest Kolmogorov-Smirnov distance on the marginal distributions between the MCMC sample and the known target.} %plots of....
	\label{fig:convergencePlots}
\end{figure}

%\restoregeometry

%\FloatBarrier

%\subsection{Description of problems}

\subsection{Standard normal}
\label{sec:stdNorm}

The first model we study is a 10-dimensional standard normal distribution, which lacks any of the challenging features discussed later, and as such provides a baseline for algorithm performance. It is a model that has analytic solutions for the Zig-Zag Sampler, so it allows us to compare the performance of the Zig-Zag with cdf inversion (pink $\boxtimes$ line labelled \emph{icdfZZ} in Figure \ref{stdNorm10dconv}) with that of the NuZZroot algorithm. As one would expected, direct sampling from the target density shows the fastest convergence. 

The three popular algorithms RWM, sMMALA and HMC are all quite close to each other in terms of performance. HMC, which was tuned to take $3$ leapfrog steps with step size $0.6$ to achieve alternate autocorrelation (identity mass), performs slightly worse than sMMALA. Note that for normal distributions, the Hessian is constant, so there is no difference between sMMALA and standard MALA. The RWM converges quite well, but falls behind sMMALA and HMC in terms of performance.

The Zig-Zag-based algorithms are represented by the blue $\times$, yellow $\Diamond$, and pink $\boxtimes$ lines. As expected, NuZZepoch has the slowest convergence as root finding is expensive in terms of epochs. That leaves the algorithm with a budget of around $8\times 10^3$ switching points, which is not sufficient to explore the density as well as the other algorithms. If this algorithm represents a lower bound for the performance of numerically approximated Zig-Zag Samplers, the yellow $\Diamond$ line of NuZZroot represents the upper bound; its efficiency significantly outperforms every other MCMC algorithm tested.

Of particular interest is the convergence of the analytical version of the Zig-Zag Sampler (pink $\boxtimes$ line labelled icdfZZ) on this model. The algorithm's performance is practically indistinguishable from NuZZroot, providing further evidence supporting the discussion on numerical errors in this work.

\subsection{Different length scales}
\label{sec:diffScales}

Having measured the performance of the algorithms on a standard normal, we now proceed to introduce the features of interest which commonly appear in Bayesian inverse problems, starting with correlations.
Some algorithms, struggle on these targets, as the length of their trajectory is tuned on one particular scale. Therefore the trajectory will be too long for variables with smaller scales, doubling back and wasting computational resources, while it will be too short for variables with larger scales, giving a more correlated sample.

A common test problem with these feature is Neal's Normal~\parencite{neal2010}, an uncorrelated normal distribution (we take the dimension to be 10, to be consistent with the rest of the test problems) with variances $[1^2,2^2, \ldots,10^2]$. The results are shown in Figure \ref{scaleNorm10dconv}.

Tuning the components of the velocity vector $\mathbf{v}$ for Zig-Zag-based processes helps mixing for this problem.
The simplest choice for the initial velocities is $v_i=1 , i\in\{1,\ldots,d\}$, however it makes sense to have different velocities for components with different scales.
%The magnitudes of the elements in the vector $\mathbf{v}$, which do not change from their initial values during switches, specify the velocity of the process, and their signs correspond to the direction the process is following in each dimension. The sign of one component of the velocity is changed at every iteration, so that the process can properly explore the space.
When we tune the velocities of the NuZZ algorithm, we use the theoretical standard deviations $\sigma_i$ for each variable, and normalise them to have fixed speed $\Vert \mathbf{v} \Vert_2 = \sqrt{d}$, as the original process. Thus the individual values of the $v_i$ change, but the overall speed of the process remains the same. Explicitly, if we have a set of standard deviations, $\sigma_i$, then the velocity update we use is
\begin{equation}
v_i = \frac{\sigma_i}{\Vert \bm{\sigma} \Vert_2} \sqrt{d} \; .
\end{equation}
The velocities can also be tuned adaptively \parencite{roberts2007coupling} using the history of the process to calculate standard deviations $\sigma_i$ for each variable on the basis of observed sample properties during the adaptation phase. Our experiments suggest that the adaptively estimated velocities converge quickly to their final value, as they only need information about the relative scale of the components. 

As before, direct sampling is the most efficient method, while NuZZepoch is too expensive to provide a good sample. NuZZroot on the other hand, performs well, surpassing the Metropolis-based algorithms. RWM and sMMALA perform again on a similar level, as they apply the same kind of global conditioning to their proposal. HMC, outperforms both of them, with step size $0.6$, $3$ leapfrog steps, and mass matrix $M=\Sigma^{-1}$.

%\FloatBarrier

\subsection{Linear correlation}
\label{sec:linearCorr}

Another common issue that arises in practice is posteriors whose variables have strong linear correlation.
In this section we compare the performance of the algorithms on a 10-dimensional multivariate normal with covariance matrix
\begin{equation}
\label{Sigma}
\bf{\Sigma} = 
\begin{bmatrix} 
1 		& -\alpha 		& \cdots 	& \cdots	& -\alpha 		\\
-\alpha 	& 1		 	& \alpha		& \cdots	& \alpha 		\\
\vdots 	& \alpha		& \ddots	& 			& \vdots	\\ 
\vdots 	& \vdots	&			& \ddots	& \alpha		\\
-\alpha 	& \alpha		& \cdots	& \alpha		& 1
\end{bmatrix} \, ,
\end{equation}
where we pick $\alpha = 0.9$.
The results can be see in Figure \ref{corrNorm10dconv}.

Unsurprisingly, the performance of RWM and sMMALA is quite similar, as they both use rotation matrices to improve their proposals. 
%RWM was tuned to have acceptance rate $25\%$, while sMMALA was tuned to have $50\%$ acceptance. 
HMC was tuned with mass matrix $M=\Sigma^{-1}$, step size $0.6$ and $3$ Leapfrog steps, with a trajectory long enough to achieve alternate autocorrelation. HMC outperforms both the RW and MALA, even though the difference is not substantial. That is because as MALA is theoretically equivalent to HMC with one leapfrog step, naturally a HMC trajectory consisting of only three leapfrog steps will yield similar results.

The Zig-Zag algorithms do not perform well in this setting. The NuZZroot, performs worse than a well-tuned Random Walk. While a RW can take advantage of the whole matrix $\Sigma$, the strategy for tuning $\mathbf{v}$ applied in the previous section only utilises the marginal standard deviations, hence ignoring information on the correlation of the components.
%Hence, unless a more complete way to tune the velocities is found, Zig-Zag algorithms may not be a convenient choice if the target model has posteriors that are very strongly correlated.
This is largely due to the fact that the process can only travel in certain directions, and if these directions are not conducive with travelling along a ridge in the target density, then the process does not mix well. 
%Practically, if the posterior is very close to a multivariate normal then application of an appropriate global transform such as rotation might resolve the problem, but this is a matter for future work.
This highlights the importance of the work in \textcite{bertazzi2020}, where the authors obtain results on using the whole covariance matrix of the target to tune the dynamics of the process and achieve better mixing.
%Other results on linearly correlated targets and real data can be found in the Supplementary Material, Section \ref{app:HP}.
The theoretical linearly correlated target considered here is the closest example to the real-data example in the Supplementary Material, Section \ref{app:HP}, in terms of posterior properties and hence algorithm performance.

%\FloatBarrier

\subsection{High dimension}
\label{sec:largeDim}

The next example we study is a standard normal in 100 dimensions. Some algorithms are known to scale particularly poorly with dimension, e.g.\ RWM, while HMC is better adapted to this scenario. 
%We are not aware of any numerical experiments showing how well the Zig-Zag dynamics scales with dimension.
% That is a feature of interest, as the algorithm may need multiple switches to get from one side of the distribution to the other in one component, because the other components need change direction too. However, the more terms are summed into $\Lambda(\bm{x}(s),\bm{v})$, the smoother the function becomes, and the faster the integration routine will be.
For this particular experiment, we reduced the computational budget to $1 \times 10^6$ epochs, to be able to run the repeats in reasonable time. As a consequence, the KS distance is generally one order of magnitude larger than that for the other models. The result of our tests are shown in Figure \ref{dimNorm100dconv}.

The performance of HMC and sMMALA is quite similar, meaning that the use of gradient information has a positive effect on their performance on this problem. Conversely, the RWM performs significantly worse. 
%The acceptance rates were kept at 25\% for RW, 50\% for sMMALA and 65\% for HMC. 
No preconditioning was used for RWM, MALA  nor HMC, and HMC was tuned to have step size $.73$ and $3$ leapfrog steps ($\approx d^{1/4}$, \textcite{beskos2013}). The performance of HMC and sMMALA is still close as HMC only needs $3$ leapfrog steps to reach alternate autocorrelation. Moreover, it should be pointed out that even though theoretically MALA is equivalent to HMC with one leapfrog step, in practice due to the leapfrog scheme the cost of one HMC iteration with one leapfrog step is more expensive than one MALA iteration.

Interestingly, the yellow $\Diamond$ line corresponding to the NuZZroot is below the sMMALA and HMC line, suggesting that the Zig-Zag dynamics does not scale as well in high dimension, but it is still superior to RWM. 
%However, the yellow $\Diamond$ line is above both sMMALA and HMC. Hence gradient and Hessian information help the exploration more than having irreversible dynamics.

%\FloatBarrier

\subsection{Fat tails}
\label{sec:fatTails}

In this example we look at how the algorithms perform when the target has fat tails. This feature is particularly common in particle physics, where posteriors are often Cauchy-like, with power-law decay in their tails. The computational budget is set again to $6\times 10^6$ epochs, and the target distribution is a 10-dimensional student-$t$ with one degree of freedom, so that all moments are infinite.

In Figure \ref{tailsT10dconv}, the curves are tightly packed with overlapping confidence intervals. RWM, HMC and sMMALA display similar convergence speed, with sMMALA being the fastest. This may be due to the fact that the Hessian is not constant in this example, which may provide some helpful local information to the dynamics. HMC was tuned to take $20$ leapfrog steps, and a step size of $.3$. The mass matrix was taken to be the identity matrix, as there is no correlation in this example, and the marginals are all equal. As the moments of the distribution are infinite, we used a fixed identity matrix for the proposal.
%, and we tuned the step size to obtain an acceptance probability of about $25\%$ for RW, and $50\%$ for sMMALA.

The Zig-Zag-like algorithms perform quite well. The blue $\times$ line representing the NuZZepoch is near the others. The yellow $\Diamond$ line corresponding to the NuZZroot is the second lowest on the chart, outperformed only by direct sampling of the density. The good performance of the Zig-Zag-based algorithms is due to the fact that the process is more frequently able to make longer excursions into the tails than the other methods, due to the low gradient of the potential away from the mode.
%  That way the algorithm can effectively explore the tails without getting trapped in regions of low density, and speedily go back to the mode.

%In this example it was crucial to choose $\mathbf{\gamma}$ such that the exploration is dominated by the irreversible dynamics, as opposed to random switches. Using our pragmatic approach to this problem, as explained in section \ref{sec:tuning}, after choosing an initial value $\gamma = \gamma_i = 10^{-2}, \forall i=1,\ldots,d$, we run a few test simulations whereby we looked at the maximum switching time value $\tau_{\mathrm{max}}$ for each run. As $\gamma$ decreases, $\tau_{\mathrm{max}}$ tends to increase, as the dynamics becomes less diffusive. We noticed that in our example $\tau_{\mathrm{max}}$ stopped increasing at around $\gamma=10^{-4}$, suggesting there would not be any advantage in decreasing $\gamma$ (and the diffusivity of the process) further.

%This procedure for choosing $\mathbf{\gamma}$ is but a rule of thumb, and it
%can be thought of as similar to the procedure used when tuning the proposal
%distribution of a RWM using the sample covariance matrix obtained from a
%previous run.

%we choose a constant value for the sum of the elements of $\bm{\gamma}$, split equally among all components, so that reducing the sum of $\bm{\gamma}$ left the maximum switching time $\tau$ roughly constant across runs. 

In the Supplementary Material, Section \ref{app:tol}, we show numerical experiments to find the highest level of the numerical tolerances such that the the difference in the posterior sample in not detectable in this example. While one iteration of the adaptive integration routine is usually enough to fall within the tolerance $\varepsilon_\mathrm{int}$, preventing us from detecting its effect, we found that the numerical error overtakes the Monte Carlo error at a value of about $\varepsilon_\mathrm{Bre} = 10^{-2}$, in this example.  
The NuZZtol process simply corresponds to the NuZZepoch process run with an increased tolerance of $\varepsilon_\mathrm{Bre} = 10^{-2}$.

%\FloatBarrier

\subsection{Non-linear correlation}
\label{sec:curvedCorr}

We move on to study how the algorithms behave in the regime of curved (or non-constant, non-linear) correlation structure. A distribution belonging to this class cannot be indexed with a global covariance matrix, as global conditioning does not capture the local characteristics of the target. They commonly look like curved ridges, or `bananas', in the state space. These shapes are common in many fields of science, often in hierarchical models or models affected by some sort of degeneracy in their parameter space (e.g.\ \cite{house2016,des2017} ).

In this section we will use the Hybrid Rosenbrock distribution from \cite{pagani2021}, and more specifically the distribution
\begin{equation}
\label{ndrosen-hybrid2}
%\pi(\mathbf{x}) = \pi^{-\frac{d}{2}} \sqrt{a \cdot b_2 \cdot \ldots \cdot b_d } \; \exp \left \{ - a \, x_{1}^2 - \sum_{i=2}^{d} b_i (x_i - x_1^2)^2 \right \} \, ,
\pi(\mathbf{x}) = \sqrt{\frac{a \, b^{d-1}}{\pi^d}} \; \exp \left \{ - a \, x_{1}^2 - b \sum_{i=2}^{d} (x_i - x_1^2)^2 \right \} \, ,
\end{equation}
where $\mathbf{x} \in \mathbb{R}^d$, $a=2.5$, $b=50$. The parameters $a,b$ have been chosen so that the shape of the distribution does not pose an extreme challenge to the algorithms we are studying. The contours of this distribution generally look like those shown in Figure \ref{ZZrosenTraj}. The results of our experiments are shown in Figure \ref{curvRosen10dconvpch}.

It should be noted that the marginals of the Hybrid Rosenbrock are not known analytically, but it is possible to sample from them directly. To produce convergence plots based on the D statistics, we built empirical cumulative distribution functions for each marginal of our target, based on $2\times 10^{10}$ samples. The Monte Carlo error introduced in our analysis by this step is therefore negligible.

%besides the very large size of the sample for the distributions considered, all the algorithms involved in this test reach values of the Kolmogorov-Smirnov distance comparable to those achieved by the same algorithms in the other tests, which were performed on analytic marginal cdf's.

Even though the curvature is not extreme, the RWM performs significantly worse than the other standard MCMC algorithms. The sMMALA algorithm performs well, as it is designed to deal with this class of densities, and the gap with the other algorithms would be larger if the shape of the target was narrower and with longer tails. HMC works quite well, even though it is difficult to tune for this target. HMC was run with step size $0.03$, $20$ leapfrog steps, identity mass, aiming for an acceptance of 80\% \parencite{betancourt2015}. The HMC tuning was obtained by repeating the analysis with several different combinations of values for the step size and number of leapfrog steps, and it corresponds to the best ESS we were able to obtain, with speed of convergence close to sMMALA.

The first part of the NuZZepoch has KS distance close to one, which means that the process has not had enough time to escape the local region where it started. However, after $\approx 10^6$ epochs it starts converging like the other algorithms. The yellow $\Diamond$ line corresponding to the NuZZroot, is above sMMALA and HMC but close to them. This is quite a promising result, as HMC is notoriously difficult to tune, especially on this class of problems, while NuZZ requires little to no tuning at all. The velocity vector was not adapted from the initial values of $\mathbf{v} = (1,\ldots,1)^{\top}$ as done in Section \ref{sec:diffScales}, as the scales of the components are roughly the same.

\FloatBarrier

\section{Conclusions}\label{sec:con}

In this work we discussed how the standard Zig-Zag Sampler from \cite{bierkens2019bigdata} is limited in its applicability to models where an analytical solution is available, or when an efficient bound based on the Jacobian or Hessian is available. We proceeded to present the NuZZ algorithm, 
%which is based on a more efficient implementation of the Zig-Zag Sampler, and 
which carries out numerically the steps that would otherwise need analytical solutions or efficient bounds. As opposed to the standard Zig-Zag sampler, NuZZ can now be applied to general models, allowing us to study the properties of the Zig-Zag dynamics from the practitioner's point of view without having to compute potentially complicated thinning bounds beforehand.

Numerical errors in the sampling of the switching times in the NuZZ lead to a perturbation in the occupation measure of the process. We showed how a bound on the Wasserstein distance between the target and the invariant distribution is linearly dependent on the tolerances of both the numerical integration and root finding routines. This demonstrates that NuZZ can be used to sample from any given target distribution up to a prescribed degree of accuracy.

We tested Zig-Zag-based algorithms against some popular MCMC algorithms, i.e.\ Random Walk Metropolis, Hamiltonian Monte Carlo and Simplified Manifold MALA, on benchmark problems that display important features that often occur in practice. 
%Notably, NuZZ has significantly fewer parameters to tune than algorithms like HMC, and those that exist can be tuned quite simply, which makes it more robust to misspecification. 
NuZZ is an expensive algorithm to run, as the rate functions $\lambda_i$ have to be evaluated at each node of the numerical integrator, for each iteration of the root finding method. As such, NuZZ per se is not competitive for practical use with respect to other MCMC algorithms. However, NuZZ allowed us to study the Zig-Zag dynamics on important test problems. Our conclusion is that overall, the Zig-Zag dynamics are competitive with other popular MCMC dynamics, and in some cases outperforms them even in the absence of super-efficiency through sub-sampling, under the assumption that at most one epoch is used per switching event.
Importantly, the novel approach to quantify numerical errors in NuZZ is quite general, and can be applied to other PDMP-based MCMC algorithms.

While there are many questions remaining about the Zig-Zag process and PDMPs in MCMC in general, there are three main directions that naturally lead on from our work. A natural question is whether the numerical efficiency of NuZZ can be improved, for example by using more tightly integrated, bespoke integration and root-finding routines. Finally, the technique used to obtain bounds on the ergodic averages between the exact and numerically approximated can be extended to other PDMPs.

\section*{Acknowledgements}
SC and TH are supported by the Alan Turing Institute for Data Science and Artificial Intelligence. TH is additionally supported by the Royal Society. FP was supported with PhD funding by the Department of Mathematics at the University of Manchester. Research of FP and SP supported by EPSRC grant Bayes4Health, ‘New Approaches to Bayesian Data Science: Tackling Challenges from the Health Sciences’ (EP/R018561/1). Research of AC supported by EPSRC grant ‘CoSInES (COmputational Statistical INference for Engineering and Security)’ (EP/R034710/1).

\printbibliography

\newpage

\appendix

\section{The Sellke ZigZag process}
\label{app:SeZZ}

In this section we describe a method for augmenting the state space to relate the NuZZ to the construction~\parencite{sellke1983} discussed in Section
\ref{sec:SeZZ}.
This can be done by changing the structure of the Zig-Zag sampler, from a PDMP that jumps in conjunction with an arrival from a Poisson process, to a PDMP that jumps when one of the components touches an active boundary \parencite{davis1993}.  We call this Markov process the Sellke Zig-Zag (SeZZ), and this can be shown using methods from \textcite{davis1993} to target the correct invariant measure. All the processes considered in this section are one-dimensional
($d=1$).

%In the second augmentation, we can modify the generator of the Zig-Zag sampler %to accommodate the inclusion of the quantity $\eta_{\mathrm{Bre}}$ as defined %in the main paper. Despite the fact that SeZZ is an exact process and does not %include numerical errors, we can explicitly account for these in the 
%discrepancy variable $q(t)$.  In order to define the generator of NuZZ as a
%Markov process, further variables need to be added. In this framework, in order
%to bound the Wasserstein distance between the exact and approximate posterior
%with the Stein discrepancy, we formulated the generators for SeZZ and NuZZ as
%Stein operators, partly following \cite{huggins2017}. We abandoned this
%approach in favour of the indirect approach described in Section \ref{sec:err},
%although present it here since the reformulations may be of interest in other
%work.  

Let us define a discrepancy variable $q(t)$ as
\begin{equation}
\label{eq:qfunction}
q(t) = R - \int_0^t \Lambda(x(s),v) \dint{s} \, , \qquad t \in
[0,\tau] \, .
\end{equation}
The evolution of $q(t)$ in time can be seen as a jump process where the value of $q(t)$ decreases deterministically according to 
\begin{equation}
	\label{qDrift}
	\ddt{q} = - \Lambda(x(t),v)
\end{equation}
until it reaches $q(\tau^-)=0$. At that point it jumps to $q(\tau^+)=R \sim \mathrm{Exp}(1)$, and continues decreasing according to Equation \eqref{qDrift}, analogously to the way in which the differential equation \eqref{linearDyn} determines the dynamics of the state $x$. By adding the discrepancy variable $q(t)$ to the existing position $x(t)$ and velocity $v(t)$ variables, we can now define the SeZZ process on $(x,v,q)$ as a PDMP where a velocity flip is triggered by the variable $q(t)$ hitting the active boundary at 0.

In order to prove that this process targets the correct stationary distribution, we will prove that the forward Kolmogorov equation is equal to zero. However, as this process has an active boundary in the state space, the generator has to differ from Equation \eqref{eq:generatorZZ} using methods from \cite[p.118]{davis1993}.

Let the variable $q$ be defined on the space $\mathcal{Q} = [0,\infty)$, which can be partitioned into the subsets $\mathcal{Q}_0 = (0,\infty)$ and $\{0\}$. Let $E=\mathcal{X} \times \mathcal{V} \times \mathcal{Q}_0$ be the interior of the state space, with $\mathcal{X} = \mathbb{R}$ and $\mathcal{V}=\{-1,+1\}$, and $\mathcal{B} = \mathcal{X} \times \mathcal{V} \times \{0\}$ be an active boundary of the state space. 
While travelling through $E$, the process is described by the drift operator
\begin{equation}
\mathcal{D} f(x,v,q) = v \ppx{f}{x} - \Lambda(x,v)\ppx{f}{q} \; ,
\label{eg1}
\end{equation}
meaning that motion is completely deterministic and there are no jumps. When the variable $q(t)$ reaches zero, the process touches the active boundary $\mathcal{B}$ causing the velocity to switch, and $q$ is refreshed with a new exponential random variable. This change is expressed through the boundary operator
\begin{equation}
\label{eg2}
\mathcal{C} f(x,v,0) = \int_{\mathcal{Q}_0} {\rm e}^{-q} f(x,-v,q) \dint{q}
- f(x,v,0) \; . 
\end{equation}

The process SeZZ targets the stationary measure
\begin{equation}
\label{eq:stationary_measure_sezz}
	\mu( \dint{x} \times \dint{v} \times \dint{q}) = \pi(x) \, \psi(v) \, \varphi(q) \, \dint{x} \, \dint{v} \, \dint{q}
\end{equation}
if 
\begin{equation}
	\label{generatorBoundary}
	\int_E \mathcal{D} f(x,v,q) \, \mu( \dint{x} \times \dint{v} \times \dint{q})
	+ \int_\mathcal{B} \mathcal{C} f(x,v,0) \, \rho ( \dint{x} \times \dint{v}) = 0 \, ,
\end{equation}
where the measure $\rho$ is induced on the boundary by $\mu$ and the process.

\begin{theorem}
	\label{th:SeZZ}
	The one-dimensional SeZZ process with extended generator given by the operators 
	\begin{equation}
		\mathcal{D} f(x,v,q) = v \ppx{f}{x} - \Lambda(x,v)\ppx{f}{q}
	\end{equation}
	and
	\begin{equation}
		\mathcal{C} f(x,v,0) = \int_{\mathcal{Q}_0} {\rm e}^{-q} f(x,-v,q) \dint{q}
		- f(x,v,0)
	\end{equation}
	targets the marginal stationary distribution $\pi(x)$.
\end{theorem}

\begin{proof}

Let the function $f: E \to \mathbb{R}$ be absolutely continuous and measurable \cite[p.118, 82]{davis1993}, and let $\int_E | f(x,v,q)-f(x_{t^-},v_{t^-},q_{t^-}) | \dint{\varsigma} < \infty$, where $\varsigma$ is the measure $\varsigma = \sum_{T_k \le t} \delta ( (x_{T_k},v_{T_k},q_{T_k}) ,T_k)$, $\forall t \ge 0$, with $T_k$ representing the jumping times and $\delta_r$ being the Dirac measure at $r \in E \times \mathbb{R}$ \cite[p.367]{davis1984}. 
Let the stationary measure be as in Equation \eqref{eq:stationary_measure_sezz},
and following \cite{lopker2013}, let the boundary measure be
\begin{equation}
\rho\, ( \dint{x} \times \dint{v}) = \Lambda(x,v) \pi(x) \, \psi(v) \, \varphi(0) \, \dint{x} \, \dint{v}
\, .
\end{equation}
Following~\cite[p.118]{davis1993}, we will show that SeZZ has $\pi(x)$ as marginal stationary distribution by proving that Equation \eqref{generatorBoundary} is satisfied for
\begin{equation}
\pi(x) \, \psi(v) \, \varphi(q) = \frac{1}{Z} {\rm e}^{-U(x)} \times \frac{1}{2} \times {\rm e}^{-q}
\, .
\end{equation}

Substituting the quantities above into Equation \eqref{generatorBoundary} above, and assuming that  $\pi(x)=0$ at $x=\pm \infty$, we obtain
\begin{multline}
\label{generatorBoundary2}
\int_\mathcal{X} \sum_{v\in\mathcal{V}} \int_{\mathcal{Q}_0} 
\left( v \left. \ppx{f}{x}\right|_{(x,v,q)} -
\Lambda(x,v) \left. \ppx{f}{q}\right|_{(x,v,q)} \right)
{\rm e}^{-U(x)-q} \, \dint{q} \, \dint{x}  \\
+ \int_\mathcal{X} \sum_{v\in\mathcal{V}}
\left(  \int_{\mathcal{Q}_0}
f(x,-v,q){\rm e}^{-q} \dint{q}-f(x,v,0)\right)
\Lambda(x,v) \, {\rm e}^{-U(x)} \, \dint{x} = 0 \, ,
\end{multline}
where the common constant $1/(2Z)$ has been cancelled.  The equality follows by integration by parts and rearranging terms.
	
Our starting point is Equation \eqref{generatorBoundary2} above, which we will re-write as:
\begin{align}
\label{generatorBoundary3}
& \overbrace{\int_\mathcal{X} \sum_{v\in\mathcal{V}} \int_{\mathcal{Q}_0} 
	v \left. \ppx{f}{x}\right|_{(x,v,q)} 
	{\rm e}^{-U(x)-q} \, \dint{q} \, \dint{x}}^{\text{First term, }I_1} + \\
& -
\overbrace{
	\int_\mathcal{X} \sum_{v\in\mathcal{V}} \int_{\mathcal{Q}_0} 
	\Lambda(x,v) \left. \ppx{f}{q}\right|_{(x,v,q)} 
	{\rm e}^{-U(x)-q} \, \dint{q} \, \dint{x}}^{\text{Second term, }I_2} \nonumber \\
& +
\underbrace{\int_\mathcal{X} \sum_{v\in\mathcal{V}}
	\left(  \int_{\mathcal{Q}_0}
	f(x,-v,q){\rm e}^{-q} \dint{q}-f(x,v,0)\right)
	\Lambda(x,v) \, {\rm e}^{-U(x)} \, \dint{x}}_{\text{Third term, }I_3}
= 0 \, . \nonumber
\end{align}
To verify this equation, we proceed by analysing each of the three terms individually.  We will often exchange the order of integration in integrals, making use of the Fubini-Tonelli theorem.  Rearranging the factors and applying integration by parts, the first term becomes
\begin{equation}
\begin{aligned}
I_1 & =
\int_\mathcal{X} \sum_{v\in\mathcal{V}} \int_{\mathcal{Q}_0} 
v \left. \ppx{f}{x}\right|_{(x,v,q)} 
{\rm e}^{-U(x)-q} \, \dint{q} \, \dint{x} \\
%\theta
& = \sum_{v\in\mathcal{V}} \int_{\mathcal{Q}_0} v{\rm e}^{-q}\int_\mathcal{X}  
\left. \ppx{f}{x}\right|_{(x,v,q)} {\rm e}^{-U(x)}  \, \dint{x} \, \dint{q} \\
& = \sum_{v\in\mathcal{V}} \int_{\mathcal{Q}_0} v{\rm e}^{-q}
\left([f(x,v,q) {\rm e}^{-U(x)}]_{x=-\infty}^{\infty} + 
\int_{\mathcal{X}} f(x,v,q) \ddx{U} {\rm e}^{-U(x)}  \, \dint{x}\right)  \, \dint{q} \\
& = \int_{\mathcal{X}} \sum_{v\in\mathcal{V}} \int_{\mathcal{Q}_0}
v f(x,v,q) \ddx{U} {\rm e}^{-U(x)-q}  \, \dint{q} \, \dint{x} \, ,
\end{aligned}
\label{I1eq}
\end{equation}
where the term in square brackets on the third line is zero as we work with target densities
$\pi(x) \propto {\rm e}^{-U(x)}$ that tend to zero as $x \rightarrow \pm
\infty$.  Again, we rearrange the factors and apply integration by parts to the
second term:
\begin{equation}
\begin{aligned}
I_2 & =
\int_\mathcal{X} \sum_{v\in\mathcal{V}} \int_{\mathcal{Q}_0} 
\Lambda(x,v) \left. \ppx{f}{q}\right|_{(x,v,q)} 
{\rm e}^{-U(x)-q} \, \dint{q} \, \dint{x} \\
& = \int_\mathcal{X} \sum_{v\in\mathcal{V}} \Lambda(x,v) {\rm e}^{-U(x)}
\int_{\mathcal{Q}_0} \left. \ppx{f}{q}\right|_{(x,v,q)} {\rm e}^{-q}
 \, \dint{q} \, \dint{x} \\
& = \int_\mathcal{X} \sum_{v\in\mathcal{V}} \Lambda(x,v) {\rm e}^{-U(x)}
\left( [f(x,v,q) {\rm e}^{-q}]_{q=0}^{\infty} -
\int_{\mathcal{Q}_0} f(x,v,q) \ddxx{{\rm e}^{-q}}{q}
 \, \dint{q} \right) \, \dint{x} \\
& = \int_\mathcal{X} \sum_{v\in\mathcal{V}} \int_{\mathcal{Q}_0}
\Lambda(x,v) f(x,v,q) {\rm e}^{-U(x)-q} \, \dint{q} \, \dint{x}
- \int_\mathcal{X} \sum_{v\in\mathcal{V}} 
\Lambda(x,v) f(x,v,0) {\rm e}^{-U(x)} \, \dint{x} \, .
\end{aligned}
\label{I2eq}
\end{equation}
Finally, the third term reduces to
\begin{equation}
\begin{aligned}
I_3 & =
\int_\mathcal{X} \sum_{v\in\mathcal{V}} \left(  \int_{\mathcal{Q}_0} f(x,-v,q){\rm e}^{-q}  \, \dint{q} -f (x,v,0)\right) \Lambda(x,v) {\rm e}^{-U(x)} \, \dint{x} \\ 
& = \int_\mathcal{X} \sum_{v\in\mathcal{V}}
\int_{\mathcal{Q}_0} f(x,-v,q){\rm e}^{-q} \Lambda(x,v){\rm e}^{-U(x)} \, \dint{q}  \, \dint{x} - \int_\mathcal{X} \sum_{v\in\mathcal{V}} f(x,v,0) \Lambda(x,v) {\rm e}^{-U(x)} \, \dint{x} \\ 
& = \int_\mathcal{X} \sum_{v\in\mathcal{V}} \int_{\mathcal{Q}_0} \Lambda(x,-v)f(x,v,q){\rm e}^{-U(x)-q} \, \dint{q} \, \dint{x} - \int_\mathcal{X} \sum_{v\in\mathcal{V}} \Lambda(x,v) f(x,v,0) {\rm e}^{-U(x)} \, \dint{x} \, .
\end{aligned}
\label{I3eq}
\end{equation}
Substituting \eqref{I1eq}, \eqref{I2eq} and \eqref{I3eq} into
\eqref{generatorBoundary3}, we obtain
\begin{align}
\label{I123eq}
I_1 - I_2 + I_3 & = 
\int_\mathcal{X} \sum_{v\in\mathcal{V}}
\int_{\mathcal{Q}_0} \left(  
v \ddx{U} - (\Lambda(x,v) - \Lambda(x,-v)) \right) f(x,v,0) \, 
{\rm e}^{-U(x)-q}  \, \dint{q} \, \dint{x} \\
& = 0 \, ,
\end{align}
which will hold if
\begin{equation}
v \ddx{U} = \Lambda(x,v) - \Lambda(x,-v) \, .
\label{vU}
\end{equation}
Then from \eqref{lambda}, we have
\begin{equation}
\Lambda(x, v) = \left ( 0 \vee v \ddx{U} \right ) + \Gamma(x, v) \; .
\label{univlambda}
\end{equation}
Substituting \eqref{univlambda} into \eqref{vU}, we see that \eqref{vU} is satisfied and hence \eqref{I123eq} is satisfied, meaning that we have demonstrated Equation \eqref{generatorBoundary2} as required.

\end{proof}

The SeZZ process forms the basis for NuZZ, which uses numerical approximations to find the roots of \eqref{eq:qfunction}. Since these approximations are dependent on the last switching location through the generation of quadrature rules, NuZZ is not Markov. However expansion of the state space to include the last switching point does induce a Markov process, but the calculations that this gives rise to have proven unwieldy, hence the approach taken in this work.

\section{Proof of Proposition \ref{prop:gillespie}}
\label{gillespieProof}

Following \parencite{davis1984,davis1993}, the generator of the process
described in the statement of the result is
\begin{equation}
\label{generatorGill}
\mathcal{U}_{\Lambda}f(\mathbf{x},\mathbf{v}) = \sum_{i=1}^d
v_i \frac{\partial f}{\partial x_i} +
\Lambda(\mathbf{x},\mathbf{v}) \left( \sum_{r=1}^d
\frac{\lambda_r(\mathbf{x},\mathbf{v})}{\Lambda(\mathbf{x},\mathbf{v})}
( f(\mathbf{x},\mathbf{F}_r(\mathbf{v}))-f(\mathbf{x},\mathbf{v}) ) \right)
\, ,
\end{equation}
where the first term represents the drift and the second the velocity switch. The process switches velocities with overall rate $\Lambda$, whereby the index of the switching component $r$ is sampled according to the multinomial transition kernel in \eqref{multinomialindex}. By carrying out some standard manipulations, Equation \eqref{generatorGill} reduces to
\begin{equation*}
\mathcal{U}_{\Lambda}f(\mathbf{x},\mathbf{v}) = \sum_{i=1}^d \left( v_i
\frac{\partial f(\mathbf{x},\mathbf{v})}{\partial x_i} +
\lambda_i(\mathbf{x},\mathbf{v})
\left(f(\mathbf{x},\mathbf{F}_i(\mathbf{v}))-f(\mathbf{x},\mathbf{v})\right) \right)
\, ,
\end{equation*}
which is identical to the generator of the original Zig-Zag process, and therefore the two processes are equivalent.

\section{Linear bound for the ZZCV algorithm}
\label{proofZZCVbound}

In the case of a ZZCV algorithm on a Logit model, the bound $\bar{\lambda}_i$ is linear, i.e.  $\bar{\lambda}_i (\mathbf{x}(s),\mathbf{v}) = a_i(\mathbf{x},\mathbf{v}) + s \, b_i(\mathbf{x},\mathbf{v})$, and can be obtained as
\begin{align}
\lambda_i^j(t) & = (v_i \partial_i U^j(\mathbf{x} + \mathbf{v} t))^+ \\
%& = \left ( v_i \partial_i U(\mathbf{x}) + v_i \partial_i U(\mathbf{x}+\mathbf{v} t) -  v_i \frac{1}{n} \sum_j \partial_i U^j(\mathbf{x}) \right )^+ \\
& = \left ( v_i \partial_i U(\mathbf{x}^*) + v_i \partial_i U(\mathbf{x}+\mathbf{v} t) -  v_i \partial_i U^j(\mathbf{x}^*) \right )^+ \\
& = \left ( v_i \partial_i U(\mathbf{x}^*) + v_i \partial_i U^j(\mathbf{x}) - v_i \partial_i U^j(\mathbf{x}^*) + v_i \partial_i U(\mathbf{x}+\mathbf{v} t) - v_i \partial_i U^j(\mathbf{x}) \right )^+ \\
& \leqslant (v_i \partial_i U(\mathbf{x}^*) + C_i \Vert \mathbf{x} - \mathbf{x}^* \Vert_2 + t \, C_i \sqrt{d})^+ \\
& \leqslant C_i \Vert \mathbf{x} - \mathbf{x}^* \Vert_2 + t \, C_i \sqrt{d}
=: a_i + t \, b_i =: \bar{\lambda}_i(t) \, . \label{zzcv} %\\
%\bar{\Lambda}(t) & = \sum_{i=1}^d \bar{\lambda}_i(t) = \sum_i a_i + t \sum_i b_i = A + t \, B
\end{align}
In Equation \eqref{zzcv} above, we have applied control variates, added and subtracted the term $v_i \partial_i U^j(\mathbf{x})$, and applied the Lipschitz condition. We moved the terms $a_i := C_i \Vert \mathbf{x} - \mathbf{x}^* \Vert_2$ and $b_i := C_i \sqrt{d}$ out of the maximum $()^+$ function as they are both always positive, the Lipschitz constant $C_i$ being calculated as a $2-$norm. We then cancelled the term $(v_i \partial_i U(\mathbf{x}^*) )^+$
remaining in the bracket as the derivative at the MLE vanishes.

Now Equation \eqref{eq:sellketGeneral} can be rewritten as
\begin{equation}
\int_0^\tau \sum_{i=1}^d \left ( a_i + s \, b_i \right ) \dint{s}
=\int_0^\tau (A + s B) \dint{s}
= R  \label{eq:abR}
\, ,
\end{equation}
where $A = \sum_{i=1}^{d}a_i$ and $B= \sum_{i=1}^{d}b_i$.  Equation~\eqref{eq:abR}
can then be solved analytically, and $i^*$ can be sampled from a multinomial
with cell probability
$\bar{\lambda}_i(\mathbf{x}(\tau),\mathbf{v})/\Lambda(\mathbf{x}(\tau),\mathbf{v})=(a_i
+ \tau b_i)/(A + \tau B)$. Finally, the proposed $\tau$ is accepted as a
switching time with probability
$\lambda_i(\mathbf{x}(\tau),\mathbf{v})/\bar{\lambda}_i(\mathbf{x}(\tau),\mathbf{v})$,
otherwise the dynamics moves on to $\mathbf{x}+\tau \mathbf{v}$, $\mathbf{v}$
is unchanged, and a new potential switching time is sampled.

\section{Wasserstein distance, preliminary results}
\label{app:wasserstein}

\begin{definition}
\label{def:wass-erg-coef}
The generalised Wasserstein ergodicity coefficient in terms of the 2-Wasserstein distance is defined as
\begin{equation}
    \kappa(K) = \sup_{\mathbf{x}_1,\mathbf{x}_2} \frac{d_{W}(\delta_{\mathbf{x}_1} P, \delta_{\mathbf{x}_2} P)}{\| \mathbf{x}_1 - \mathbf{x}_2 \|_2} \, .
\end{equation}
\end{definition}
In this work we assume that both $\kappa(P) < +\infty$ and $\kappa(\widetilde{P}) < +\infty$. %Definition \ref{def:wass-erg-coef} holds.

\begin{proposition} 
\label{pr:W_contract_submult}
The 2-Wasserstein distance is contractive and submultiplicative, i.e. for any measures $\mu_1,\mu_2$ with second moments and any Markov kernels $P$ and $\widetilde{P}$ such that $\kappa(P) < +\infty$, $\kappa(\widetilde{P}) < +\infty$:
\begin{enumerate}
    \item $d_{W}(\mu_1 P,\mu_2 P) \leqslant \kappa(P) d_{W}(\mu_1,\mu_2)$,
    \item $\kappa(P \widetilde{P}) \leqslant \kappa(P) \kappa(\widetilde{P})$ \, .
\end{enumerate}
\end{proposition}

\begin{proof}[Proof of Proposition \ref{pr:W_contract_submult}]
Consider the two statements separately.

\begin{enumerate}

\item Let $(\mathbf{X}_{\mu_1},\mathbf{X}_{\mu_2})$ be a realisation of the optimal coupling between $\mu_1$ and $\mu_2$ for $d_{W}$. Let $\mathbf{Y},\mathbf{Y}'$ be such that $(\mathbf{Y},\mathbf{Y}') \, | \, (\mathbf{X}_{\mu_1},\mathbf{X}_{\mu_2})$ is a realisation of the optimal coupling between $\delta_{\mathbf{X}_{\mu_1}}P$ and $\delta_{\mathbf{X}_{\mu_2}}P$. Since $\mathbf{Y} \sim \mu_1 P$ and $\mathbf{Y'} \sim \mu_2 P$. Then
\begin{align}
    d_{W}^2(\mu_1 P,\mu_2 P) &\leqslant \mathbf{E}[\|\mathbf{Y}-\mathbf{Y}'\|_2^2] \\
    &\leqslant \mathbf{E}[\mathbf{E}[\|\mathbf{Y}-\mathbf{Y}'\|_2^2 \, | \, \mathbf{X}_{\mu_1},\mathbf{X}_{\mu_2}]] \nonumber \\
    &\leqslant \mathbf{E}[\|\mathbf{X}_{\mu_1}-\mathbf{X}_{\mu_2}\|_2^2 \kappa^2(P)] \nonumber \\
    &\leqslant d_{W}^2(\mu_1,\mu_2) \kappa^2(P), \nonumber
\end{align}
where we used the fact that $d_W$ is an infimum, the tower property of the expected value, and Definition \ref{def:wass-erg-coef}. On the last line, we used the fact that $\mathbf{X}_{\mu_1}$ and $\mathbf{X}_{\mu_2}$ are optimally coupled.

\item Using the contraction property on $\widetilde{P}$:
\begin{align}
    \kappa(P \widetilde{P}) 
    & = \sup_{\mathbf{x}_1,\mathbf{x}_2} \frac{d_{W}(\delta_{\mathbf{x}_1} P\widetilde{P},\delta_{\mathbf{x}_2} P\widetilde{P})}{\|\mathbf{x}_1-\mathbf{x}_2\|_2} \\
    &\leqslant \sup_{\mathbf{x}_1,\mathbf{x}_2} \frac{\kappa(\widetilde{P}) \, d_{W}(\delta_{\mathbf{x}_1} P,\delta_{\mathbf{x}_2} P)}{\|\mathbf{x}_1-\mathbf{x}_2\|_2} \nonumber \\
    &\leqslant \kappa(P) \kappa(\widetilde{P}) \, , \nonumber 
\end{align}
obtained using the repeated application of Definition \ref{def:wass-erg-coef}.

\end{enumerate}

\end{proof}

\subsection{Proof of Proposition \ref{pr:close_2_measures_n_step}}

\begin{proof}[Proof of Proposition \ref{pr:close_2_measures_n_step}]
Following the proof from \cite{rudolf2018}.
Recall that as $P$ is Wasserstein geometrically ergodic, $\kappa(P^n) \leqslant C \rho^n$, with $\rho < 1$.
By induction:
\begin{equation}
    \widetilde{\pi}_n - \pi_n = (\widetilde{\pi}_0 - \pi_0) P^n + \sum_{k=0}^{n-1}\widetilde{\pi}_k (\widetilde{P}-P) P^{n-k-1} \, .
\end{equation}
Using a coupling argument:
\begin{equation}
    d_{W}^2(\widetilde{\pi}_k P,\widetilde{\pi}_k \widetilde{P}) \leqslant \int d_{W}^2(\delta_\mathbf{x} P, \delta_\mathbf{x} \widetilde{P}) \, \dint \widetilde{\pi}_k \leqslant \epsilon^2 \, .
\end{equation}
Furthermore, using the contraction property from Proposition \ref{pr:W_contract_submult}:
\begin{align*}
    d_{W}(\widetilde{\pi}_k \widetilde{P} P^{n-k-1}, \widetilde{\pi}_k P P^{n-k-1}) 
    & \leqslant d_{W}(\widetilde{\pi}_k \widetilde{P},\widetilde{\pi}_k P) \kappa(P^{n-k-1}) \\
    & \leqslant \epsilon \, C \rho^{n-k-1} .
\end{align*}

Finally:
\begin{align*}
    d_{W}(\pi_n,\widetilde{\pi}_n) &\leqslant d_{W}(\pi_0 P^n,\widetilde{\pi}_0 P^n) + \sum_{k=0}^{n-1} d_{W}(\widetilde{\pi}_k \widetilde{P} P^{n-k-1}, \widetilde{\pi}_k P P^{n-k-1}) \\
    & \leqslant d_{W}(\pi_0,\widetilde{\pi}_0) C \rho^n + \epsilon C \sum_{k=0}^{n-1} \rho^k \\
    & = d_{W}(\pi_0,\widetilde{\pi}_0) C \rho^n + \epsilon C \frac{1-\rho^n}{1-\rho}
\end{align*}

\end{proof}

\subsection{Proof of Proposition \ref{pr:P_unique_pi}}

\begin{proof}[Proof of Proposition \ref{pr:P_unique_pi}]
Let $E$ be a complete metric space.
Using the Banach fixed point theorem (and Definition \ref{def:Wass_geom_ergo}), we have that there exists a measure $\pi$ with finite second moment such that $P^n(\mathbf{x},\cdot) \rightarrow \pi$ in Wasserstein distance. By using a coupling argument similar to the one found in the proof of Proposition \ref{pr:W_contract_submult}, we find:
\begin{align*}
    d_W(\pi P, \delta_\mathbf{x} P^{n+1}) \leqslant C \rho \cdot d_W(\pi, \delta_\mathbf{x} P^{n}) \rightarrow 0.
\end{align*}
It thus follows that $\pi P = \pi$ and hence that $\pi$ is invariant.

We will now prove that $\pi$ is ergodic. Let $A$ be an invariant set, that is for all $\mathbf{x} \in A$, $P(\mathbf{x},A) = 1$. Clearly, for any $\mathbf{x} \in A$ and for all $n > 0$, we have $P^n(\mathbf{x},A) = 1$. Since $P^n(\mathbf{x},\cdot) \rightarrow \pi$, this implies that $\pi(A) = 1$, which implies that $\pi$ is an ergodic measure.

Finally, by using again a coupling argument similar to the one found in the proof of Proposition \ref{pr:W_contract_submult}, we find that $\nu P^n \rightarrow \pi$ in Wasserstein distance for $\nu$ with finite second order moments.
\end{proof}

\section{Perturbation bound results}
\label{app:perturbation_bounds}

\subsection{Proof of Lemma \ref{lm:tau_tau_tilde}}

\begin{proof}[Proof of Lemma \ref{lm:tau_tau_tilde}]
Let the error in the integral of the switching rate due to having a numerically approximate switching time $\widetilde{\tau}$ be defined as
\begin{equation}
\label{eq:numerical_error}
    \eta_\mathrm{Root} = \int_0^\tau \Lambda^\mathbf{z}(s) \, \dint s - \int_0^{\tilde{\tau}} \Lambda^\mathbf{z}(s) \, \dint s \, ,
\end{equation}
with 
\begin{equation}
\label{eq:bound_eta_root}
    |\eta_\mathrm{Root}| \le \varepsilon_\mathrm{Bre} + \varepsilon_\mathrm{int} \, .
\end{equation}

Applying the Mean Value Theorem to $f(t) = \int_0^t \Lambda^\mathbf{z} (s) \, \dint s$, we obtain that for some $c \in \left[ \tau, \widetilde{\tau} \right]$, it holds that
\begin{align}
\label{eq:mean_val_th}
    \Lambda^\mathbf{z}(c) = \frac{\int_0^\tau \Lambda^\mathbf{z}(s) \, \dint s - \int_0^{\tilde{\tau}} \Lambda^\mathbf{z}(s) \, \dint s}{\tau - \widetilde{\tau} },
\end{align}
and hence that this expression is bounded from below as
\begin{align*}
    \Lambda^\mathbf{z}(c) &\geqslant \inf_{t \geqslant 0} \Lambda^\mathbf{z} (t) \\
    &= \inf_{t \geqslant 0} \Lambda^\mathbf{z} (\mathbf{x} + t\mathbf{v}, \mathbf{v}) \\
    &\geqslant \inf_{(\mathbf{x},\mathbf{v})\in E} \Lambda^\mathbf{z} (\mathbf{x}, \mathbf{v}) \\
    &=: \Lambda_\mathrm{min}.
\end{align*}
Combining Equations \eqref{eq:numerical_error}, \eqref{eq:bound_eta_root} and \eqref{eq:mean_val_th}, we obtain
\begin{equation}
    |\tau - \widetilde{\tau}| \leqslant \frac{|\eta_\mathrm{Root}|}{\Lambda_\mathrm{min}} \le \frac{\varepsilon_\mathrm{Bre} + \varepsilon_\mathrm{int}}{\Lambda_\mathrm{min}} \, .
\end{equation}
This result can be combined with Proposition \ref{prop:wass_dist_P_Ptilde} to show the dependence of the error bound $\epsilon$ on the numerical tolerances of the inner solvers. 
\end{proof}

\subsection{Proof of Lemma \ref{lem:lip_jump}}

\begin{proof}[Proof of Lemma \ref{lem:lip_jump}]
Let $\lambda_i(\mathbf{x},\mathbf{v}) = \max \left( 0, - v_i \partial_i \log \pi(\mathbf{x}) \right) + \gamma_i$ be the unnormalised probability of switching from velocity $\mathbf{v}$ to velocity $\mathbf{v}' = \mathbf{F}_i(\mathbf{v})$ at the time of an event at position $\mathbf{x}$. 
Write $\Lambda(\mathbf{x},\mathbf{v}) = \sum_{i=1}^d \lambda_i(\mathbf{x},\mathbf{v})$. Then, the normalised transition kernel at $\mathbf{x}$ from $\mathbf{v}$ to $\mathbf{v}'=\mathbf{F}_i(\mathbf{v})$ is
\begin{align*}
    Q(\mathbf{v}' | \mathbf{x},\mathbf{v}) = Q(\mathbf{v} \to \mathbf{v}'=\mathbf{F}_i(\mathbf{v})|\mathbf{x},\mathbf{v}) = \frac{\lambda_i(\mathbf{x},\mathbf{v})}{\Lambda(\mathbf{x},\mathbf{v})}.
\end{align*}
Now, fix $(\mathbf{v}, \mathbf{v}')$ and consider the function $h : \mathbf{x} \mapsto Q(\mathbf{v}'|\mathbf{x},\mathbf{v})$. One can write
\begin{equation}
    \label{eq:h_function}
    \left | \frac{\partial h(\mathbf{x})}{\partial x_j}  \right | 
    = \left | \frac{\partial \lambda_i (\mathbf{x},\mathbf{v})}{\partial x_j}  \frac{1}{\Lambda(\mathbf{x},\mathbf{v})} - \frac{\lambda_i(\mathbf{x},\mathbf{v})}{\Lambda(\mathbf{x},\mathbf{v})} \frac{1}{\Lambda(\mathbf{x},\mathbf{v})} \frac{\partial \Lambda(\mathbf{x},\mathbf{v})}{\partial x_j} \right | \, ,
\end{equation}
whose terms can be bounded as follows:
\begin{align}
    \frac{\partial \lambda_i (\mathbf{x},\mathbf{v})}{\partial x_j} & \leqslant \left | \frac{\partial \lambda_i (\mathbf{x},\mathbf{v})}{\partial x_j} \right | \leqslant | v_i | \left |\frac{\partial^2}{\partial x_i \partial x_j} \log \pi(\mathbf{x}) \right | \leqslant M \, , \\
    \frac{\lambda_i(\mathbf{x},\mathbf{v})}{\Lambda(\mathbf{x},\mathbf{v})} & \leqslant 1 \, , \\
    \Lambda(\mathbf{x},\mathbf{v})^{-1} & \leqslant \Gamma^{-1} \, , \\
    \frac{\partial \Lambda(\mathbf{x},\mathbf{v})}{\partial x_j} & = \sum_i \frac{ \partial \lambda_i(\mathbf{x},\mathbf{v})}{\partial x_j} \leqslant d M \, .
\end{align}
Combining these equations one can bound the term on the left-hand side of Equation \eqref{eq:h_function} as
\begin{equation}
    \left | \frac{\partial h(\mathbf{x})}{\partial x_j} \right | \leqslant \left| \Gamma^{-1} M \right| + \left| \Gamma^{-1} M d \right| \leqslant 2 \Gamma^{-1} M d \, .
\end{equation}
It follows that, along the same trajectory,
\begin{align*}
    \lvert Q(\mathbf{v}'|\mathbf{x},\mathbf{v}) - Q(\mathbf{v}'|\mathbf{x}',\mathbf{v}) \rvert &= \lvert h(\mathbf{x}) - h(\mathbf{x}') \rvert \\
    & = \left\lvert \int_0^1  \nabla_\mathbf{x} h ( t \mathbf{x} + (1 - t) \mathbf{x}' )^\top \; (\mathbf{x} - \mathbf{x}') \, \dint t \right\rvert \\
    & \leqslant \int_0^1 \left  \| \nabla_\mathbf{x} h ( t \mathbf{x} + (1 - t) \mathbf{x}') \right \|_2 \cdot \| \mathbf{x} - \mathbf{x}' \|_2 \, \dint t \\
    & \leqslant \int_0^1 \sqrt{d \cdot (2 \Gamma^{-2} M^2 d)^2}  \cdot \| \mathbf{x} - \mathbf{x}' \|_2 \, \dint t \\
    & \leqslant \int_0^1 2 \Gamma^{-1} M d^{3/2} \cdot \| \mathbf{x} - \mathbf{x}' \|_2 \, \dint t \\
    &= 2 \Gamma^{-1} M d^{3/2} \| \mathbf{x} - \mathbf{x}'\|_2 \, ,
\end{align*}
which concludes the proof.
\end{proof}

\subsection{Proof of Proposition \ref{prop:wass_dist_P_Ptilde}}

\begin{proof}[Proof of Proposition \ref{prop:wass_dist_P_Ptilde}]

Let us couple draws from $\delta_\mathbf{z} P, \delta_\mathbf{z} \widetilde{P}$ as follows: in each case, the stopping time is obtained by calculating the first time at which a (deterministic) integral exceeds a unit-mean exponential random variable $R$. We couple the stopping times $\tau, \widetilde{\tau}$ by sharing this variable $R$ across processes, so that
\begin{align*}
    R & = \int_0^\tau \Lambda^\mathbf{z} (s) \, \dint s \\
    &= \int_0^{\tilde{\tau}} \widetilde{\Lambda}^\mathbf{z} (s) \, \dint s,
\end{align*}
where $\widetilde{\Lambda}^\mathbf{z}$ is the polynomial approximation to the event rate which is constructed implicitly by the numerical integration routine. %in Brent's method.
Upon moving the $\mathbf{x}$ coordinates forward to $\mathbf{x}' = \mathbf{x} + \tau \cdot \mathbf{v}$ and $\widetilde{\mathbf{x}}' = \mathbf{x} + \widetilde{\tau} \cdot \mathbf{v}$, the new velocities must then be drawn from distributions with density functions $Q(\mathbf{v}'|\mathbf{x}',\mathbf{v})$ and $Q(\widetilde{\mathbf{v}}'|\widetilde{\mathbf{x}}', \mathbf{v})$ respectively. We couple these draws according to a maximal coupling, i.e.\ so as to maximise $\mathbb{P} \left( \mathbf{v}' = \widetilde{\mathbf{v}}' \, | \, \mathbf{x}', \widetilde{\mathbf{x}}', \mathbf{v} \right)$.

Under this coupling, it holds that the square 2-Wasserstein distance can be bounded as follows.
\begin{align*}
    d_W^2(\delta_\mathbf{z} P, \delta_\mathbf{z} \widetilde{P})
    & \leqslant \mathbf{E}[ \| \mathbf{z}' - \widetilde{\mathbf{z}}' \|_E^2 \, | \, \mathbf{x}, \mathbf{v}] \\
    & = \mathbf{E} [ ( \| \mathbf{x}' - \widetilde{\mathbf{x}}' \|_2 + \| \mathbf{v}' - \widetilde{\mathbf{v}}' \|_2 )^2 \, | \, \mathbf{x}, \mathbf{v} ] \\
    & \leqslant 2 \mathbf{E} [ \| \mathbf{x}' - \widetilde{\mathbf{x}}' \|_2^2 \, | \, \mathbf{x}, \mathbf{v} ] + 2 \mathbf{E} [ \| \mathbf{v}' - \widetilde{\mathbf{v}}' \|_2^2 \, | \, \mathbf{x}, \mathbf{v} ] \, .
\end{align*}
Consider the difference of velocities next. For all $\mathbf{x}$, it holds that for all $\mathbf{w}_1, \mathbf{w}_2$ in the support of $Q(\cdot | \mathbf{x}, \mathbf{v})$, the square norm $\| \mathbf{w}_1 - \mathbf{w}_2 \|_2^2$ can be bounded by a constant $8$, as each vector differs from the original $\mathbf{v}$ in at most one coordinate.
%That can in turn be bounded using Lemma \ref{lem:lip_jump}. 

Hence, multiplying by the probability of sampling different velocities, 
\begin{align}
    \mathbf{E}[ \| \mathbf{v}' - \widetilde{\mathbf{v}}' \|_2^2 \, | \, \mathbf{x}, \mathbf{v} ] 
    & = 8 \, \mathbf{E}[ \mathbb{P} (\mathbf{v}' \neq \widetilde{\mathbf{v}}' \, | \, \mathbf{x}, \mathbf{v}) \, | \, \mathbf{x}, \mathbf{v} ] \\
    & \leqslant 8 \, ( 2 \Gamma^{-1} M d^{3/2} ) \mathbf{E} [ \| \mathbf{x}' - \widetilde{\mathbf{x}}' \|_2 \, | \, \mathbf{x}, \mathbf{v} ] \, ,
\end{align}
where the inequality follows from Lemma \ref{lem:lip_jump}. Substituting this into the square Wasserstein distance equation and rearranging terms we obtain
\begin{equation}
    d_W^2(\delta_\mathbf{z} P, \delta_\mathbf{z} \widetilde{P}) \leqslant 
    2 \, \mathbf{E}[ \| \mathbf{x}' - \widetilde{\mathbf{x}}' \|_2^2 \, | \, \mathbf{x}, \mathbf{v} ] + 32 \Gamma^{-1} M d^{3/2}
    \mathbf{E}[ \| \mathbf{x}' - \widetilde{\mathbf{x}}' \|_2 \, | \, \mathbf{x}, \mathbf{v} ] \, ,
\end{equation}
where note that the first norm is squared while the second is not.
Now, the term in $\mathbf{x}$ can be bounded as follows
\begin{align}
    \mathbf{E}[ \| \mathbf{x}' - \widetilde{\mathbf{x}}' \|_2 \, | \, \mathbf{x}, \mathbf{v} ] 
    & = \mathbf{E}[ \| (\tau - \widetilde{\tau}) \mathbf{v} \|_2 \, | \, \mathbf{x}, \mathbf{v} ] \\
    & = \mathbf{E}[ |\tau - \widetilde{\tau}| \, | \, \mathbf{x}, \mathbf{v} ] \; \| \mathbf{v} \|_2 \\
    & \leqslant \Gamma^{-1} (\varepsilon_\mathrm{Bre} + \varepsilon_\mathrm{int}) \sqrt{d} \, ,
\end{align}
where we first used the relationship between velocity, space and time ($\mathbf{x}' - \widetilde{\mathbf{x}}' = \left( \tau - \widetilde{\tau} \right) \mathbf{v}$), then invoked Lemma \ref{lm:tau_tau_tilde}, and used the fact that the 2-norm of a $d$-vector with elements $\{\pm 1\}$ is $\sqrt{d}$, and the introduction of non-unit velocities for computational efficiency as we do in Section \ref{sec:exp} can be done in a length-preserving manner. . Substituting into the original equation we obtain the final bound
\begin{align}
    d_W^2(\delta_\mathbf{z} P, \delta_\mathbf{z} \widetilde{P}) 
    & \leqslant 2 \Gamma^{-1} d (\varepsilon_\mathrm{Bre} + \varepsilon_\mathrm{int})^2
    + 32 \Gamma^{-1} M d^2 (\varepsilon_\mathrm{Bre} + \varepsilon_\mathrm{int}) \, .
\end{align}
Taking the square root completes the proof.

\end{proof}

\subsection{Proof of Proposition \ref{prop:WdistYchain}}

\begin{proof}[Proof of Proposition \ref{prop:WdistYchain}]
Using the bound between $P$ and $\widetilde{P}$ given by Proposition \ref{pr:close_2_measures_n_step}, and Proposition \ref{prop:wass_dist_P_Ptilde}, we have:
\begin{equation}
    d_W(\widetilde{\nu} P^n, \widetilde{\nu} \widetilde{P}^n) \leq \epsilon C \frac{1-\rho^n}{1-\rho} \, ,
\end{equation}
for $C$ and $\rho$ as defined in Assumption \ref{as:Wergodic_Y_and_P} and $\epsilon$ as in Equation \eqref{eq:epsdef}.
Letting $n\rightarrow +\infty$ yields the result, since $\widetilde{\nu} \widetilde{P}^n = \widetilde{\nu}$ and $\widetilde{\nu} P^n \rightarrow \nu$ in Wasserstein distance by Proposition \ref{pr:P_unique_pi}. 
\end{proof}

\section{Invariant measure results}
\label{app:invariant_measure}

\subsection{Proof of Lemma \ref{lm:U_Utilde_unique_Pi}}
\begin{proof}[Proof of Lemma \ref{lm:U_Utilde_unique_Pi}]
For any invariant measure of $\mathbf{U}_k = ( \mathbf{Y}_{k-1}, \mathbf{Y}_k)$, the projection on the first coordinate is an invariant measure of $\mathbf{Y}_k$, which is therefore $\nu$ since it is unique. In addition, $\mathbb{P}(\mathbf{Y}_{k+1}|\mathbf{Y}_k) = P (\mathbf{Y}_k, \mathbf{Y}_{k+1})$, hence an invariant measure of $\mathbf{U}_k$ must be of the form $\zeta(A) = \int \mathbb{P} (\mathbf{y}_1, \{ \mathbf{y}_2 \, | \, (\mathbf{y}_1, \mathbf{y}_2) \in A \}) \dint \nu( \mathbf{y}_1 )$ for $A$ measurable. Since $\zeta$ is clearly invariant for $\mathbf{U}_k$, $\zeta$ is the unique invariant measure of $\mathbf{U}_k$. The same can be said for $\widetilde{\mathbf{U}}_k$ and $\widetilde{\zeta}$.

\end{proof}

\subsection{Proof of Proposition \ref{pr:W_U_Utilde}}
\begin{proof}[Proof of Proposition \ref{pr:W_U_Utilde}]

Let $\mathbf{Y}_1$ and $\widetilde{\mathbf{Y}}_1$ be random variables associated via the optimal coupling of $\nu$ and $\widetilde{\nu}$ such that $\mathbf{E} [\| \mathbf{Y}_1 - \widetilde{\mathbf{Y}}_1\|_2^2] = d_W(\nu,\widetilde{\nu})$. Let $\mathbf{Y}_2,\widetilde{\mathbf{Y}}_2$ be random variables such that for fixed $\mathbf{Y}_1,\widetilde{\mathbf{Y}}_1$, they follow the optimal coupling of $P(\mathbf{Y}_1,\cdot)$ and $P(\widetilde{\mathbf{Y}}_1,\cdot)$, i.e.\ $d_W(\delta_{\mathbf{Y}_1}P,\delta_{\tilde{\mathbf{Y}}_1}P) = \mathbf{E} [ \| \mathbf{Y}_2 - \widetilde{\mathbf{Y}}_2 \|_2^2 \; | \; \mathbf{Y}_1, \widetilde{\mathbf{Y}}_1]$.

Since
\begin{align*}
    d_W(\delta_\mathbf{y} P, \delta_{\tilde{\mathbf{y}}} \widetilde{P}) 
    & \le d_W(\delta_\mathbf{y} P, \delta_{\tilde{\mathbf{y}}} P) + d_W(\delta_{\tilde{\mathbf{y}}} P, \delta_{\tilde{\mathbf{y}}} \widetilde{P}) \\
    & \le C \rho \, \|\mathbf{y} - \widetilde{\mathbf{y}}\| + \epsilon
\end{align*}
for $C$ and $\rho$ as defined in Assumption \ref{as:Wergodic_Y_and_P} and $\epsilon$ as in Equation \eqref{eq:epsdef},
by applying the triangular inequality, Definition \ref{def:Wass_geom_ergo} and Proposition \ref{prop:wass_dist_P_Ptilde},
we obtain
\begin{align*}
    \mathbf{E} [\| \mathbf{Y}_2 - \widetilde{\mathbf{Y}}_2 \|^2] 
    & \le \mathbf{E} \bigg [ \mathbf{E} \Big [ \| \mathbf{Y}_2 - \widetilde{\mathbf{Y}}_2 \|^2 \; | \; \mathbf{Y}_1,\widetilde{\mathbf{Y}}_1 \Big ] \bigg ] \\
    & \le \mathbf{E} [ (C \rho \, \| \mathbf{Y}_1 - \widetilde{\mathbf{Y}}_1 \| + \epsilon )^2] \\
    & \le 2 C^2 \rho^2 \, d_W^2(\nu,\widetilde{\nu}) + 2 \epsilon^2
\end{align*}
The laws of $(\mathbf{Y}_1, \mathbf{Y}_2)$ and $(\widetilde{\mathbf{Y}}_1, \widetilde{\mathbf{Y}}_2)$ are $\zeta$ and $\widetilde{\zeta}$ respectively, which gives us a coupling between $\zeta$ and $\widetilde{\zeta}$. We can therefore deduce that
\begin{equation*}
    d_W^2(\zeta, \widetilde{\zeta}) \leqslant \mathbf{E} [\| (\mathbf{Y}_1, \mathbf{Y}_2) - (\widetilde{\mathbf{Y}}_1, \widetilde{\mathbf{Y}}_2) \|^2] \le (1 + 2 C^2 \rho^2) \, d_W^2(\nu, \widetilde{\nu}) + 2 \epsilon^2.
\end{equation*}

\end{proof}

\subsection{Proof of Lemma \ref{lm:ergodic_avg_helper}}

\begin{proof}[Proof of Lemma \ref{lm:ergodic_avg_helper}]
The ergodic averages for a function $f$ of the process $\mathbf{Z}_t$ up to time $T_n$ can be written in terms of the interpolation operator $\mathcal{J}$, and a set of velocities $\mathbf{u}_k$, as
\begin{equation}
\label{eq:ergodic_avg_Tn}
        \frac{1}{T_n} \int_0^{T_n} f(\mathbf{z}_s) \, \dint s = \frac{n}{T_n} \cdot \frac{1}{n}\sum_{k=1}^n \left( \mathcal{J} f \right) (\mathbf{u}_k) \, .
\end{equation}
This result is obtained by starting with the integral on the left-hand side of this equation, then applying the transformation $t = (s-T_{k-1})/(T_k - T_{k-1})$, and defining each velocity as equal to displacement divided by time, and multiplying both sides by $T_n^{-1}$ %, and the RHS by $n/n$.
Since by Lemma \ref{lm:U_Utilde_unique_Pi} the chain $\mathbf{U}_k$ has a unique invariant measure, and $\mathcal{J}f \in L^p(\zeta)$ for $p \in [1,\infty)$, we use the Birkhoff theorem \cite{sandric2017} to take the limit and write
\begin{equation}
    \lim_{n\rightarrow \infty} \left( \frac{1}{T_n} \int_0^{T_n} f(\mathbf{z}_s) \, \dint s \right )
     = \left( \lim_{n \to \infty} \frac{n}{T_n} \right ) \left( \lim_{n \to \infty} \frac{1}{n} \sum_{k=1}^n \left ( \mathcal{J} f \right ) (\mathbf{u}_k) \right )  = \frac{1}{H} \int_{E^2} \left ( \mathcal{J} f \right ) (\mathbf{u}) \, \zeta ( \dint \mathbf{u}) \quad \mbox{a.s.} \, .
\end{equation}
%
%which holds with probability 1.
The same calculations apply mutatis mutandis to $\widetilde{\mathbf{Z}}_t$ and $\widetilde{\mathbf{U}_k}$. 
%(with appropriate modifications)
\end{proof}

\subsection{Lemma \ref{lem:bound_helper_function} with proof.}

\begin{lemma}
\label{lem:bound_helper_function}
For every $\mathbf{u}_1 = (\mathbf{x}_1, \mathbf{v}_1, \mathbf{x}_1', \mathbf{v}_1')$, $\mathbf{u}_2 = (\mathbf{x}_2, \mathbf{v}_2, \mathbf{x}_2', \mathbf{v}_2')$ in $E \times E$, the state space of $\mathbf{U}_k$, and every bounded, Lipschitz function $f : E \to \mathbb{R}$, it holds that
\begin{equation}
    \left| \left( \mathcal{J} f \right) (\mathbf{u}_1) - \left( \mathcal{J} f \right) (\mathbf{u}_2)  \right| \leqslant \frac{1}{\sqrt{d}}\left\{ |f|_\mathrm{Lip} \| \mathbf{x}_1' - \mathbf{x}_1 \|_2 + \| f \|_{\infty} \right\} \| \mathbf{u}_1 - \mathbf{u}_2 \|_{E^2} \, ,
\end{equation}
where $\| \cdot \|_{E^2}$ is defined in Equation \eqref{eq:2-norm-E}. %-square
\end{lemma}

\begin{proof}[Proof of Lemma \ref{lem:bound_helper_function}]
Write directly that
\begin{equation*}
    \left( \mathcal{J} f \right) (\mathbf{u}_1) - \left( \mathcal{J} f \right) (\mathbf{u}_2) = 
\frac{1}{\sqrt{d}} \int_0^1 \| \mathbf{x}_1' - \mathbf{x}_1 \|_2 f( \mathbf{x}_{1, t}, \mathbf{v}_1) \, \dint t 
    - {\sqrt{d}} \int_0^1 \| \mathbf{x}_2' - \mathbf{x}_2 \|_2 f( \mathbf{x}_{2, t}, \mathbf{v}_2) \, \dint t\, ,
\end{equation*}
where $\mathbf{x}_{k, t} = (1 - t) \mathbf{x}_k + t \mathbf{x}_k'$, using that $\| \mathbf{v} \|_2 \equiv \sqrt{d}$. Now, define $f_k (t) = f(\mathbf{x}_{k, t}, \mathbf{v}_k)$, and decompose the difference of the integrands above as 
\begin{align*}
    \Delta(t)
    & = \| \mathbf{x}_1' - \mathbf{x}_1 \|_2  f_1 (t) - \| \mathbf{x}_2' - \mathbf{x}_2 \|_2  f_2 (t) \\
    & = \| \mathbf{x}_1' - \mathbf{x}_1 \|_2  ( f_1 (t) - f_2 (t) ) + \left( \| \mathbf{x}_1' - \mathbf{x}_1 \|_2 - \| \mathbf{x}_2' - \mathbf{x}_2 \|_2 \right)  f_2 (t)\, .
\end{align*}
One can then bound for $t\in [0,1]$
\begin{align*}
    | f_1 (t) - f_2 (t) | 
    & \leqslant | f |_\mathrm{Lip}  \left \| (\mathbf{x}_{1, t}, \mathbf{v}_1) - (\mathbf{x}_{2, t}, \mathbf{v}_2) \right \|_{E} \\
    & = | f |_\mathrm{Lip}  \left( \| \mathbf{x}_{1, t} - \mathbf{x}_{2,t} \|_2 + \| \mathbf{v}_1 - \mathbf{v}_2 \|_2 \right) \\
    & = | f |_\mathrm{Lip}  \left( \| (1-t) (\mathbf{x}_1 - \mathbf{x}_2) + t (\mathbf{x}_1' - \mathbf{x}_2') \|_2 + \| \mathbf{v}_1 - \mathbf{v}_2 \|_2 \right) \\
    & \leqslant | f |_\mathrm{Lip}  \left ( (1-t)  \|\mathbf{x}_1 - \mathbf{x}_2\|_2 + t \|\mathbf{x}_1' - \mathbf{x}_2'\|_2 + \| \mathbf{v}_1 - \mathbf{v}_2 \|_2 \right ) \\
    & \leqslant | f |_\mathrm{Lip}  \left ( \|\mathbf{x}_1 - \mathbf{x}_2\|_2 + \|\mathbf{x}_1' - \mathbf{x}_2'\|_2 + \| \mathbf{v}_1 - \mathbf{v}_2 \|_2 \right ) \\
    & \leqslant | f |_\mathrm{Lip}  \| \mathbf{u}_1 - \mathbf{u}_2 \|_{E^2}\, .
\end{align*} 
Additionally, the triangle inequality shows that
\begin{align*}
    \| \mathbf{x}_1' - \mathbf{x}_1 \|_2 - \| \mathbf{x}_2' - \mathbf{x}_2 \|_2 
    & \leqslant \| (\mathbf{x}_1' - \mathbf{x}_1) - (\mathbf{x}_2' - \mathbf{x}_2) \|_2  \\
    & \leqslant \| \mathbf{x}_1 - \mathbf{x}_2 \|_2 + \| \mathbf{x}_1' - \mathbf{x}_2' \|_2 \\
    & \leqslant \|\mathbf{u}_1 - \mathbf{u}_2\|_{E^2} \, .
\end{align*} 
Thus,
\begin{align*}
    | \Delta (t) | 
    & \leqslant \| \mathbf{x}_1' - \mathbf{x}_1 \|_2 \cdot | f |_\mathrm{Lip} \cdot \| \mathbf{u}_1 - \mathbf{u}_2 \|_{E^2} + \| \mathbf{u}_1 - \mathbf{u}_2 \|_{E^2} \cdot \| f \|_\infty \, ,
\end{align*}
from which the result follows by integration.
\end{proof}

\subsection{Proof of Theorem \ref{th:mu_ergodic_measure}}

\begin{proof}[Proof of Theorem \ref{th:mu_ergodic_measure}]
The measure $\widetilde{\mu}$ as defined in Equation \eqref{eq:mu_tilde} exists as shown in Lemma \ref{lm:ergodic_avg_helper} and Remark \ref{rm:mu_tilde}.

Our argument to derive the bound in Inequality \eqref{eq:bound_th1} consists in proving that the right-hand sides of Equations \eqref{eq:erg_avg_u_exact} and \eqref{eq:erg_avg_u_approx} from Lemma \ref{lm:ergodic_avg_helper} are close to each other. In order to do that we need to prove that $H$ and $\widetilde{H}$ are close, and that
\begin{equation}
    \int_{E^2} ( \mathcal{J} f ) (\mathbf{u}) \, \zeta (\dint \mathbf{u}) 
    \quad  \text{and} \quad
    \int_{E^2} ( \mathcal{J} f ) (\mathbf{u}) \, \widetilde{\zeta} (\dint \mathbf{u}) \
\end{equation}
are close. Let us start with the two integrals.

\begin{enumerate}
    \item 

Let $\phi$ be an optimal coupling of $( \zeta, \widetilde{\zeta})$ under the 2-Wasserstein distance. By the coupling property, it holds that
\begin{equation*}
    \zeta( \mathcal{J} f ) - \widetilde{\zeta} \left( \mathcal{J} f \right) = \int \left \{ \left( \mathcal{J} f \right) (\mathbf{u}_1)- \left( \mathcal{J} f \right) (\mathbf{u}_2) \right \} \phi(\dint \mathbf{u}_1, \dint \mathbf{u}_2).
\end{equation*}
From Lemma \ref{lem:bound_helper_function}, we recall that for all $\mathbf{u}_1, \mathbf{u}_2 \in E^2$ and $f$ bounded and Lipschitz,
\begin{equation}
    \left| \left( \mathcal{J} f \right) (\mathbf{u}_1) - \left ( \mathcal{J} f \right) ( \mathbf{u}_2)  \right| \leqslant \frac{1}{\sqrt{d}} \left ( |f|_\text{Lip} \| \mathbf{x}_1' - \mathbf{x}_1\| + \| f \|_{\infty} \right ) \cdot \| \mathbf{u}_1 - \mathbf{u}_2 \|_{E^2}.
\end{equation}
We can thus bound
\begin{align*}
    \left| \zeta( \mathcal{J} f ) - \widetilde{\zeta} (\mathcal{J} f) \right| &
 \leqslant \int \left( \frac{1}{\sqrt{d}} \left( |f|_\text{Lip} \|
\mathbf{x}_1' - \mathbf{x}_1\| + \| f \|_{\infty} \right) \|
\mathbf{u}_1 - \mathbf{u}_2 \|_{E^2} \right) \phi \left( \dint \mathbf{u}_1,
\dint \mathbf{u}_2 \right) \\
    & \leqslant \frac{1}{\sqrt{d}} |f|_\mathrm{Lip} \int \| \mathbf{x}_1' - \mathbf{x}_1\| \cdot \| \mathbf{u}_1 - \mathbf{u}_2 \|_{E^2} \phi \left( \dint \mathbf{u}_1, \dint \mathbf{u}_2 \right) \\
    & \qquad + \sqrt{\frac{1}{d}} \| f \|_{\infty} \int  \| \mathbf{u}_1 - \mathbf{u}_2 \|_{E^2}  \phi \left( \dint \mathbf{u}_1, \dint \mathbf{u}_2 \right) \\
    &\leqslant \frac{1}{\sqrt{d}} |f|_\mathrm{Lip} \left( \int \| \mathbf{x}_1' - \mathbf{x}_1\|^2 \phi \left( \dint \mathbf{u}_1, \dint \mathbf{u}_2 \right) \right)^{1/2} \left( \int  \| \mathbf{u}_1 - \mathbf{u}_2 \|_{E^2}^2  \phi \left( \dint \mathbf{u}_1, \dint \mathbf{u}_2 \right) \right)^{1/2} \\
    & \qquad + \sqrt{\frac{1}{d}}  \| f \|_{\infty} \left(\int  \| \mathbf{u}_1 - \mathbf{u}_2 \|_{E^2}^2  \phi \left( \dint \mathbf{u}_1, \dint \mathbf{u}_2 \right) \right)^{1/2} \\
    & = \frac{1}{\sqrt{d}}  |f|_\mathrm{Lip}  \mathbf{E}_\zeta \left[ \| \mathbf{x}_1' - \mathbf{x}_1 \|^2 \right]^{1/2}  \mathbf{E}_\phi \left[ \| \mathbf{u}_1 - \mathbf{u}_2 \|^2 \right]^{1/2} +  \sqrt{\frac{1}{d}}  \| f \|_{\infty}  \mathbf{E}_\phi \left[ \| \mathbf{u}_1 - \mathbf{u}_2 \|^2 \right]^{1/2} \\
    & = d_W ( \zeta, \widetilde{\zeta} )  \left( \frac{1}{\sqrt{d}}  |f|_\mathrm{Lip}  \mathbf{E}_\zeta \left[ \| \mathbf{x}_1' - \mathbf{x}_1 \|^2 \right]^{1/2} +  \sqrt{\frac{1}{d}}  \| f \|_{\infty} \right),
\end{align*}
where the final step uses that $\mathbf{E}_\phi \left[ \| \mathbf{u}_1 - \mathbf{u}_2 \|^2 \right]^{1/2} = d_W ( \zeta, \widetilde{\zeta} )$.

\item

Let us now prove that $H$ and $\widetilde{H}$ are close. Define the function $g : (\mathbf{x},\mathbf{v},\mathbf{x}',\mathbf{v}') \mapsto \| \mathbf{x} - \mathbf{x}' \| / \| \mathbf{v} \|$, and note that $H = \zeta(g)$, $\widetilde{H} = \widetilde{\zeta}(g)$. As $g$ is $d^{-\frac{1}{2}}$-Lipschitz (noting that $\| \mathbf{v} \|$ is constant on our domain of interest), we can apply Kantorovich duality to bound
\begin{align*}
    | H - \widetilde{H} | = | \zeta(g) - \widetilde{\zeta} (g) |
    & \leqslant | g |_{\mathrm{Lip}} \cdot d_{W,1} (\zeta, \widetilde{\zeta}) \\
    & = d^{-\frac{1}{2}} d_{W,1} (\zeta, \widetilde{\zeta}) \\
    & \leqslant d^{-\frac{1}{2}} d_{W} (\zeta, \widetilde{\zeta}),
\end{align*}
where the final step uses an elementary inequality comparing the 1-Wasserstein distance $d_{W,1}$ and the 2-Wasserstein distance $d_{W}$.

\end{enumerate}

Finally, combining the results,
\begin{align*}
    |\widetilde{\mu}(f) - \mu(f)| &= \left| \frac{1}{\widetilde{H}} \int_{E^2} ( \mathcal{J} f ) (\mathbf{u}) \, \widetilde{\zeta} (\dint \mathbf{u}) - \frac{1}{H} \int_{E^2} ( \mathcal{J} f ) (\mathbf{u}) \, \zeta (\dint \mathbf{u})   \right| \\
    &= \frac{1}{\widetilde{H}} \left| \int_{E^2} ( \mathcal{J} f ) (\mathbf{u}) \, \widetilde{\zeta} (\dint \mathbf{u}) - \left(1 - \frac{H - \widetilde{H}}{H}\right) \int_{E^2} ( \mathcal{J} f ) (\mathbf{u}) \, \zeta (\dint \mathbf{u})   \right| \\
    &= \frac{1}{\widetilde{H}} \left| \int_{E^2} ( \mathcal{J} f ) (\mathbf{u}) \, \widetilde{\zeta} (\dint \mathbf{u}) - \int_{E^2} ( \mathcal{J} f ) (\mathbf{u}) \, \zeta (\dint \mathbf{u})  \right| + \frac{| H - \widetilde{H}|}{\widetilde{H}} |\mu(f)|\\
    &\leqslant \frac{1}{\widetilde{H}} d_{W}(\zeta, \widetilde{\zeta}) \frac{1}{\sqrt{d}} \left(  |f|_\mathrm{Lip} \mathbf{E}_\zeta \left[ \| \mathbf{x}_1' - \mathbf{x}_1 \|^2 \right]^{1/2} +  \| f \|_{\infty} + |\mu(f)| \right).
\end{align*}
\end{proof}

%%%
%%%
%%%
%%% SUPPLEMENT
%%%
%%%
%%%

\newpage

\vspace*{1cm}
\begin{center}
\textsc{\LARGE Supplementary Material}
\end{center}
\vspace*{.5cm}

\setcounter{figure}{0}
\setcounter{section}{0}
\setcounter{equation}{0}

\renewcommand{\theequation}{S.\arabic{equation}}
\renewcommand{\thesection}{S.\arabic{section}}
\renewcommand{\thefigure}{S.\arabic{figure}}

\section{House prices data}
\label{app:HP}

In this section we compare our algorithms on real-world data, specifically on the House-Price dataset from \cite{housePriceKaggle}. 
%The target is a normal distribution, $\bm{\beta}|\{X, \mathbf{y}, \hat{\sigma}\} \sim \mathcal{N}(\bm{\hat{\beta}}, \hat{\sigma}^2 (\bm{X}^{\top} \bm{X})^{-1} )$. 

Consider a Bayesian linear regression model, $\mathbf{y} = \bm{X} \bm{\beta} + \bm{\varepsilon}$. We are interested in finding the posterior distribution over the length-$d$ parameter vector $\bm{\beta}$, given the $n \times d$ matrix of independent variables $\bm{X}$, and the length-$n$ vector of dependent variables, $\mathbf{y}$. The errors in the length-$n$ vector $\bm{\varepsilon}$ are taken to be normally distributed, $\bm{\varepsilon} \sim \mathcal{N}(\mathbf{0}, \sigma^2 \bm{I})$, where $\mathbf{0}$ is a length-$n$ vector with all zero elements and $\bm{I}$ is the $n\times n$ identity matrix. Taking the prior on $\bm{\beta}$ to be improper, i.e.\ $\pi_0(\bm{\beta}) \propto 1$, the marginal posterior on the parameter vector given the observed data and standard deviation of the error has the simple form $\bm{\beta}|\{X, \mathbf{y}, \sigma\} \sim \mathcal{N}(\bm{\hat{\beta}}, \sigma^2 (\bm{X}^{\top} \bm{X})^{-1} )$, where we use hats to denote maximum-likelihood estimates (MLEs). 

The House-Prices dataset is composed of $n=1459$ house sales prices, along with $d=80$ variables for each house, describing features such as position, type of heating, and type of bath.  We selected this dataset because it has a reasonably high number of variables, many of which are categorical, as is often the case in real datasets. 
Experiments performed on synthetic datasets, other simple targets and different performance metrics suggest that these conclusion hold in more general cases.
Missing values in the data were dealt with as suggested in~\cite{gaudreau2017}, and we coded categorical variables using dummy variables. We discarded the variables \textit{Utilities} (type of utilities available), \textit{TotalBsmtSF} (total square feet of basement area), and \textit{GrLivArea} (above grade (ground) living area square feet), to avoid singularities in the matrix $(\bm{X}^{\top} \bm{X})^{-1}$, and the variable \textit{Id} (identity number), which is not part of the analysis. Finally, we make the standard choice to use the logarithm of the sales price as the variable $\mathbf{y}$.  This processing step left us with $d=77$ variables. 

To simplify running the algorithms, we do not estimate the error variance parameter $\sigma^2$ as part of the MCMC algorithm. Rather, we set it exactly equal to its MLE, $\sigma = \hat{\sigma}$. This model has the advantage, therefore, of having known posterior, which allows us to use the same performance metric that we used in the previous tests. 

The House Prices dataset is quite challenging for all the algorithms involved in this test. Its posterior combines multiple features that we previously explored separately. The number of variables is large ($d=77$), their marginals have very different scales with a difference up to several orders of magnitude, and some of the variables are significantly correlated. All of these features commonly occur in real world datasets, which further increases the importance of this test. Due to the challenging nature of this model, the computational budget for this test was increased to $2.4 \times 10^7$ epochs, thinned down to $6 \times 10^6$ points, and the algorithms were started at the Maximum Likelihood Estimate. The results can be seen in Figure \ref{HousePricesconvpch}.

\begin{figure}[!ht] %[!h] meaning here if possible, but not absolutely
	\includegraphics[scale=0.6]{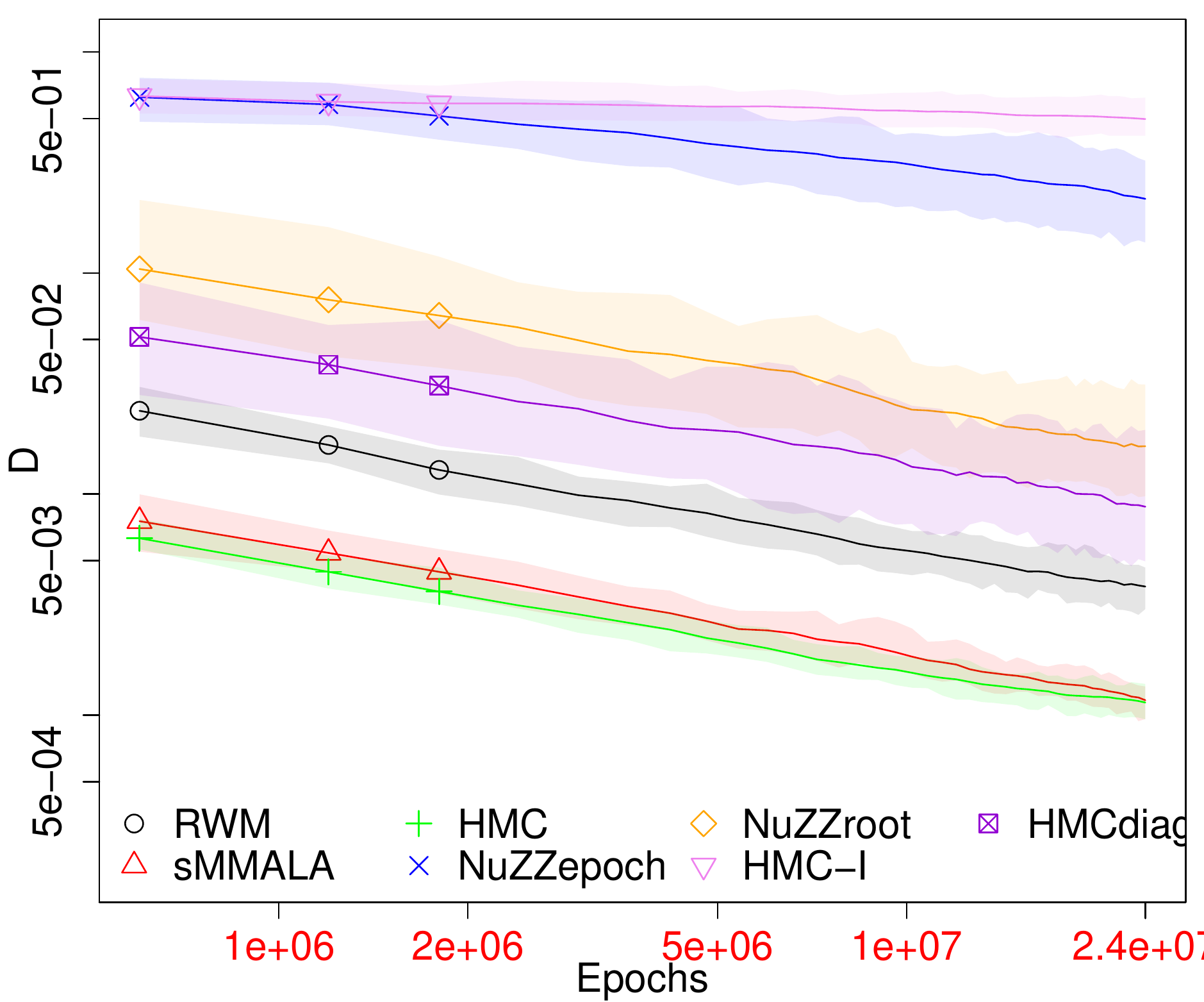}
	\centering
	\caption{Convergence of MCMC algorithms in the Kolmogorov-Smirnov metric on a linear regression model estimated on the House Prices datase. }
	\label{HousePricesconvpch}
\end{figure}

sMMALA performs well, surpassing RW by good measure %They were tuned to have an acceptance of $25\%$ for RW and $50\%$ for sMMALA. 
and performs similarly to HMC, again due to the low number of leapfrog steps that HMC needs, and to the fact that the Hessian is independent of $\bm{\beta}$. %the value of

Once again the green $+$ line representing HMC with mass matrix $M=\hat{\sigma}^{-2}(\bm{X}^{\top} \bm{X})$, step size $.83$ and $3$ leapfrog steps, performs very well, with acceptance close to $67\%$. 

The pink line $\bigtriangledown$ at the top of the graph on the other hand, shows the performance of HMC should the mass matrix be left as the identity matrix. That requires HMC to take an extreme $2400$ leapfrog steps with step size $.2$, with poor results compared to the other algorithms. This may serve as a caveat to practitioners on the importance of tuning for HMC.

As in the other examples, the blue line $\times$ corresponding to NuZZepoch is far above the other algorithms, as the cost of the root finding and integration is very high.

The yellow $\Diamond$ line corresponding to NuZZroot is above the black $\bigcirc$ line corresponding to the Random Walk. The Zig-Zag dynamics in this example is hampered by the presence of features such as high dimension and high correlation together. For comparison, we added the purple line $\boxtimes$ representing HMC run with $100$ leapfrog steps, step size $.2$, and mass matrix equal to the diagonal of $\hat{\sigma}^{-2}(\bm{X}^{\top} \bm{X})$, i.e.\ eliminating information about the correlation from the HMC setup, to make it more comparable to how we tune the velocities for NuZZ. The NuZZroot performance is still inferior to that of HMC in this case, likely due to the presence of correlation between regressors, once again highlighting the importance of the work in \cite{bertazzi2020}.

\section{Effects of tolerances}
\label{app:tol}

As mentioned in the text, in most experiments we set both the integration routine absolute tolerance and Brent's method absolute tolerance to $10^{-10}$. 

In the example from Section \ref{sec:fatTails} we changes these tolerances systematically and assessed their impact on accuracy of the posterior samples.

\begin{figure}[!ht]
	\centering
	\includegraphics[width=.9\linewidth]{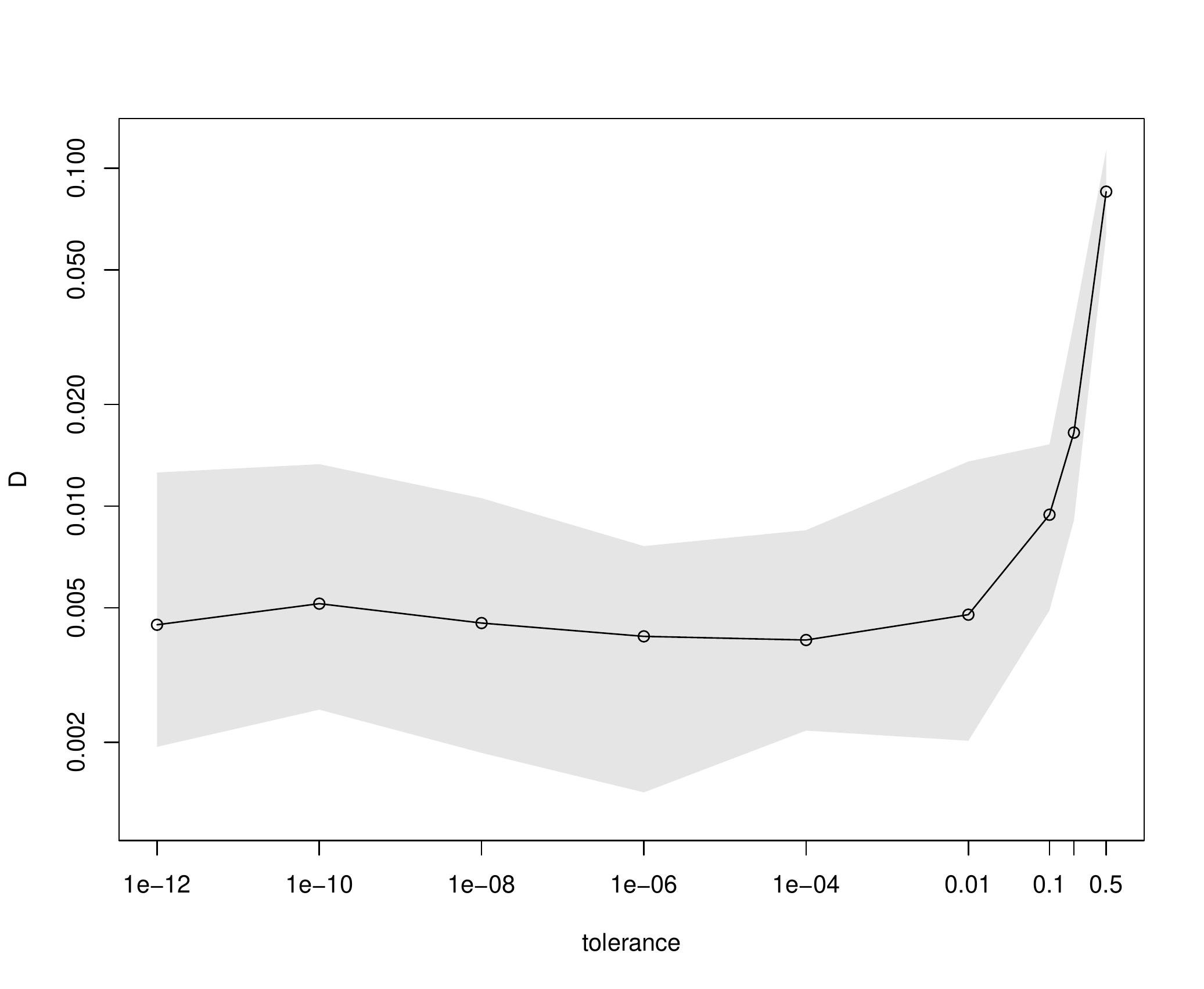}
	\caption{Dependence between the largest Kolmogorov distance on
		the marginals and the root finding tolerance $\varepsilon_{\mathrm{Bre}}$, for a fixed integration tolerance $\varepsilon_{\mathrm{int}}=10^{-10}$. The computational budget was $6 \times 10^6$ epochs, which was processed into $6 \times 10^6$ NuZZ samples. }
	\label{fig:toleranceRootStudT10d}
\end{figure} 

Figure \ref{fig:toleranceRootStudT10d} shows how $D$, the worst Kolmogorov-Smirnov distance on the marginals, increases as we increase $\varepsilon_{\mathrm{Bre}}$.  In this example, with $6 \times 10^6$ samples, the Monte Carlo error dominates the numerical error for all $\varepsilon_{\mathrm{Bre}}<10^{-2}$.
This result is model dependent, but it supports the point that even though NuZZ uses numerical approximations to sample the switching times, if the numerical error is small, the difference in the posterior is not detectable.

We also tested the influence of the integration error on the quality of the sample from NuZZ, but since the QAGS integration routine performs a minimum of a 15-point Gaussian integration steps, this is often more than enough to reach the condition $\eta_{\mathrm{int}} \le \varepsilon_{\mathrm{int}}=10^{-10}$. Therefore we were unable to coarsen this part of the approximation sufficiently in order to produce the analogous plot to Figure \ref{fig:toleranceRootStudT10d} for $\varepsilon_{\mathrm{int}}$, but simply note that for all the examples we considered, one integration step is usually sufficient, and the integration tolerance simply acts as a fail-safe for when they are not.

\end{document}